\PassOptionsToPackage{hyphens}{url}
\PassOptionsToPackage{dvipsnames,svgnames,x11names}{xcolor}
\documentclass[12pt]{article}

\usepackage[round, authoryear]{natbib}

\usepackage{xr-hyper}
\usepackage{hyperref} 

\usepackage{amsmath}
\usepackage{graphicx, psfrag, epsf, orcidlink, setspace,enumitem}
\usepackage[font=small]{caption}
\usepackage{url} 
\usepackage[margin=1in, showframe=false]{geometry}
\usepackage{float,bm,booktabs,makecell,amsthm,amssymb,amsfonts,dsfont,mathtools,authblk,subcaption}

\usepackage{titlesec}

\titlespacing*{\section}{0pt}{8pt plus 4pt minus 2pt}{6pt plus 2pt}
\titlespacing*{\subsection}{0pt}{6pt plus 3pt minus 2pt}{4pt plus 2pt}
\titlespacing*{\subsubsection}{0pt}{6pt plus 2pt minus 2pt}{4pt plus 1pt}

\usepackage{algorithm}
\usepackage{algorithmic}
\allowdisplaybreaks

\usepackage[font=small]{caption}
\usepackage{tabularx}
\usepackage{ragged2e}
\graphicspath{{plots/}}

\usepackage{amssymb}
\usepackage{hyperref}
\usepackage{multirow}
\usepackage{booktabs}
\usepackage{mathrsfs, xcolor} 
\definecolor{darkblue}{rgb}{0,0,.6}
\hypersetup{colorlinks = true, linkcolor=darkblue, urlcolor=darkblue, citecolor=darkblue}
\usepackage[patch=none]{microtype}

\usepackage{iftex}
\ifPDFTeX
  \usepackage[T1]{fontenc}
  \usepackage[utf8]{inputenc}
  \usepackage{textcomp} 
\else 
  \usepackage{unicode-math}
  \defaultfontfeatures{Scale=MatchLowercase}
  \defaultfontfeatures[\rmfamily]{Ligatures=TeX,Scale=1}
\fi
\usepackage{lmodern}
\ifPDFTeX\else  
\fi
\IfFileExists{upquote.sty}{\usepackage{upquote}}{}
\IfFileExists{microtype.sty}{
  \usepackage[]{microtype}
  \UseMicrotypeSet[protrusion]{basicmath} 
}{}
\makeatletter
\@ifundefined{KOMAClassName}{
  \IfFileExists{parskip.sty}{%
    \usepackage{parskip}
  }{
    \setlength{\parindent}{0pt}
    \setlength{\parskip}{6pt plus 2pt minus 1pt}}
}{
  \KOMAoptions{parskip=half}}
\makeatother
\usepackage{xcolor}
\makeatletter
\ifx\paragraph\undefined\else
  \let\oldparagraph\paragraph
  \renewcommand{\paragraph}{
    \@ifstar
      \xxxParagraphStar
      \xxxParagraphNoStar
  }
  \newcommand{\xxxParagraphStar}[1]{\oldparagraph*{#1}\mbox{}}
  \newcommand{\xxxParagraphNoStar}[1]{\oldparagraph{#1}\mbox{}}
\fi
\ifx\subparagraph\undefined\else
  \let\oldsubparagraph\subparagraph
  \renewcommand{\subparagraph}{
    \@ifstar
      \xxxSubParagraphStar
      \xxxSubParagraphNoStar
  }
  \newcommand{\xxxSubParagraphStar}[1]{\oldsubparagraph*{#1}\mbox{}}
  \newcommand{\xxxSubParagraphNoStar}[1]{\oldsubparagraph{#1}\mbox{}}
\fi
\makeatother

\usepackage{longtable,booktabs,array}
\usepackage{calc} 
\usepackage{etoolbox}
\makeatletter
\patchcmd\longtable{\par}{\if@noskipsec\mbox{}\fi\par}{}{}
\makeatother
\IfFileExists{footnotehyper.sty}{\usepackage{footnotehyper}}{\usepackage{footnote}}
\makesavenoteenv{longtable}
\usepackage{graphicx}
\makeatletter
\def\maxwidth{\ifdim\Gin@nat@width>\linewidth\linewidth\else\Gin@nat@width\fi}
\def\maxheight{\ifdim\Gin@nat@height>\textheight\textheight\else\Gin@nat@height\fi}
\makeatother
\setkeys{Gin}{width=\maxwidth,height=\maxheight,keepaspectratio}
\makeatletter
\def\fps@figure{htbp}
\makeatother

\makeatletter
\@ifpackageloaded{caption}{}{\usepackage{caption}}
\AtBeginDocument{%
\ifdefined\contentsname
  \renewcommand*\contentsname{Table of contents}
\else
  \newcommand\contentsname{Table of contents}
\fi
\ifdefined\listfigurename
  \renewcommand*\listfigurename{List of Figures}
\else
  \newcommand\listfigurename{List of Figures}
\fi
\ifdefined\listtablename
  \renewcommand*\listtablename{List of Tables}
\else
  \newcommand\listtablename{List of Tables}
\fi
\ifdefined\figurename
  \renewcommand*\figurename{Figure}
\else
  \newcommand\figurename{Figure}
\fi
\ifdefined\tablename
  \renewcommand*\tablename{Table}
\else
  \newcommand\tablename{Table}
\fi
}
\@ifpackageloaded{float}{}{\usepackage{float}}
\floatstyle{ruled}
\@ifundefined{c@chapter}{\newfloat{codelisting}{h}{lop}}{\newfloat{codelisting}{h}{lop}[chapter]}
\floatname{codelisting}{Listing}

\makeatother
\makeatletter
\@ifpackageloaded{caption}{}{\usepackage{caption}}
\@ifpackageloaded{subcaption}{}{\usepackage{subcaption}}
\makeatother

\ifLuaTeX
  \usepackage{selnolig}  
\fi
\usepackage[]{natbib}
\usepackage{bookmark}

\IfFileExists{xurl.sty}{\usepackage{xurl}}{} 
\hypersetup{
  pdftitle={Title},
  pdfauthor={Author 1; Author 2},
  pdfkeywords={3 to 6 keywords, that do not appear in the title},
  colorlinks=true,
  linkcolor={blue},
  filecolor={Maroon},
  citecolor={Blue},
  urlcolor={Blue},
  pdfcreator={LaTeX via pandoc}}

\newcommand{\anon}{1}

\DeclareMathOperator{\ran}{ran}

\DeclareMathOperator{\op}{op}

\newcommand{\PP}{\mathcal{P}}

\renewcommand{\hat}{\widehat}

\newcommand{\HX}{\mathcal H_{\mathbf{x}}}
\newcommand{\HY}{\mathcal H_{Y}}
\newcommand{\KX}{d_{\mathbf{x}}}
\renewcommand{\hat}{\widehat}

\numberwithin{equation}{section}

\newcommand{\Bpre}{\check B}
\newcommand{\Upre}{\check U}

\newtheorem{proposition}{Proposition}

\newcommand{\Ito}{It\^{o}}
\newcommand{\ei}{i}

\begin{document}

\def\spacingset#1{\renewcommand{\baselinestretch}%
{#1}\small\normalsize} \spacingset{1}

\def\be{\begin{equation}}
\def\ee{\end{equation}} 
\def\ben{\begin{equation*}}
\def\een{\end{equation*}}
\def\bea{\begin{eqnarray}}
\def\eea{\end{eqnarray}}
\def\bda{\begin{eqnarray*}}
\def\eda{\end{eqnarray*}}
\numberwithin{equation}{section}

\newtheorem{theorem}{Theorem}
\newtheorem{corollary}{Corollary}
\newtheorem{assumption}{Assumption}
\renewcommand\theassumption{A\arabic{assumption}}
\newtheorem{lemma}{Lemma}
\newtheorem{remark}{Remark}
\newtheorem{example}{Example}

\theoremstyle{definition}
\newtheorem{exmp}{Example}[section]
\AtEndDocument{\refstepcounter{theorem}\label{finalthm}}
\AtEndDocument{\refstepcounter{proposition}\label{finalprop}}
\newcommand{\pkg}[1]{{\normalfont\fontseries{b}\selectfont #1}} \let\proglang=\textsf \let\code=\texttt
\setlength{\abovedisplayskip}{4pt}
\setlength{\belowdisplayskip}{4pt}

\date{}

\if1\anon
{
  \title{\bf Anthropogenic Forcing, Climate Change, and the Shape of Warming: Statistical Inference for Distributional Cointegration}
  \author{\normalsize
    Kyungsik Nam\\ \vspace{-0.8em}    Division of Climate Change, Hankuk University of Foreign Studies \\ \vspace{1.2em}    
    Won-Ki Seo\thanks{The data and computational code used to reproduce the reported results are available in the public repository at \url{https://anonymous.4open.science/r/FRSTAT_Program-5F12/}.} \\
    School of Economics, University of Sydney
  }
  \maketitle
}
\fi

\if0\anon
{
  \bigskip
  \bigskip
  \bigskip
  \begin{center}
    {\LARGE\bf Anthropogenic Forcing, Climate Change, and the Shape of Warming:\\ Statistical Inference for Distributional Cointegration}
\end{center}
  \medskip
} \fi

\begin{abstract}
Anthropogenic forcing components follow different long-run paths, while persistent temperature change can involve distributional changes beyond the mean. Scalar regressions aggregate these components and retain only mean temperature, obscuring how distinct forcing paths relate to persistent distributional change. We develop new testing, estimation, and inference methods for long-run relations between an integrated predictor vector and a density-valued response. These comprise a residual-based test of between-cointegration (whether predictor trends account for all stochastic trends in the response density), a fully modified least-squares estimator of predictor-specific functional responses, and simulation-based inference for interpretable projections. We apply the methods to densities of observed local temperature anomalies and anthropogenic effective radiative forcing divided into CO$_2$ and non-CO$_2$ portfolios. The test results are consistent with persistent movements in these portfolios statistically accounting for the persistent evolution of the anomaly distribution, with no additional stochastic trend detected in the residual. A joint test rejects the common-response restriction imposed by aggregating the two portfolios. The fitted CO$_2$ response mainly shifts mass toward warmer anomalies and increases central concentration, whereas the non-CO$_2$ response produces a smaller shift but greater dispersion and off-center reshaping. Positive fitted mean responses for both portfolios conceal these contrasts, demonstrating the information lost through scalar aggregation.
\end{abstract}

\noindent%
\vfill

\newpage
\spacingset{1.7} 

\section{Introduction}\label{sec:intro}

\noindent Standard climate econometric models relate the global mean temperature anomaly to aggregate radiative forcing, reducing both sides of the long-run relation to scalars. This scalar-on-scalar formulation targets persistent movements in mean temperature but cannot recover distributional shape on the response side or forcing composition on the predictor side. Similar mean movements may coexist with changes in dispersion, asymmetry, and tail mass, while forcing components combined in an anthropogenic forcing index may have distinct distributional response profiles. Our analysis therefore focuses on stochastic trends, by which we mean persistent nonstationary components arising from the accumulation of stationary increments. We ask whether the stochastic trends in an anthropogenic forcing vector account for all persistent nonstationary variation in annual distributions of observed local temperature anomalies and how the component-specific response profiles vary across anomaly states.

\indent The scalar time-series literature provides a natural starting point for this analysis. Research on temperature persistence examines whether global and hemispheric anomaly series contain stochastic trends or are stationary around deterministic trends subject to structural breaks \citep{estrada2017extracting,kim2020inference}. Related work studies long-run relations between aggregate temperature and radiative forcing \citep{dergiades2016long}. Despite differing in their representations of persistence, these approaches retain a scalar temperature response. Specifications with multiple forcing variables relax aggregation on the predictor side, but their conclusions remain confined to persistent movements in the mean anomaly.

\indent Distributional evidence shows why the scalar response is restrictive. \citet{rivas2020trends} document changes across moments and quantiles of temperature distributions. Treating cross-sectional anomaly densities as functional time series, \citet{chang2020evaluating} find that stochastic trends operate through location and dispersion and differ across hemispheric distributions. Persistent temperature change may involve changes in shape as well as location. A scalar response cannot determine whether a long-run forcing relation operates through the center of the distribution or through dispersion, asymmetry, and tail mass.

\indent
Aggregation on the predictor side imposes a restriction. Effective radiative forcing accounts distinguish anthropogenic components with different signs and historical paths \citep{IPCC_AR6_WGI_AnnexIII_2021}. Unit-root evidence for related greenhouse-gas, sulfur-forcing, and total-forcing measures motivates an integrated long-run treatment of radiative forcing \citep{kaufmann2002cointegration,pretis2020econometric}. Components may differ in spatial incidence: using station-level data, \citet{magnus2011global} distinguish broadly distributed warming associated with greenhouse gases from more localized aerosol-related cooling. These differences make a common distributional response an empirical restriction rather than a consequence of forcing accounting. Once components are combined into an index, component responses cannot be identified.

\indent
The response-side and predictor-side extensions are therefore complements
rather than substitutes. Retaining the anomaly distribution preserves
information about shape, while retaining a forcing vector preserves information
about composition. Our empirical specification takes the annual distribution
of observed local temperature anomalies as a density-valued response and
separates anthropogenic forcing into carbon dioxide and the net contribution of
the remaining components. The conventional aggregate-mean regression serves as
the scalar benchmark, allowing the information recovered from distributional
shape and forcing composition to be assessed separately and jointly.

\indent We formulate this relation as distributional cointegration between an integrated predictor vector and a density-valued response, which we call \emph{between-cointegration}. The centered log-ratio (CLR) transformation represents the anomaly densities in a linear Hilbert space using the geometry of density-valued functional data \citep{Egozcue2006,petersen2016}. Between-cointegration holds when a vector-to-function long-run relation between the forcing variables and the transformed density process leaves a stationary functional residual. The restriction therefore asks whether persistent variation in the anomaly distribution is spanned by the forcing trends through a stable long-run operator. By contrast, we use the term \emph{within-cointegration} for cointegration arising from stationary linear combinations within a single process. This is the principal object of the existing Hilbert-space cointegration literature \citep{Chang2016152,BSS2017,NSS}, which does not directly address the cross-process testing, operator estimation, and response-function inference required here.

\indent The procedure follows the order in which the long-run relation is assessed and interpreted. We first construct residual-based tests for a functional response and an integrated vector predictor by extending the null-of-cointegration approach of \citet{shin1994residual}. Because their null distributions are nonpivotal, critical values are obtained by plug-in Monte Carlo. We then estimate the vector-to-function operator by fully modified least squares, extending \citet{phillips1995fully} to correct for long-run endogeneity and serial dependence. The limit theory yields simulation-based marginal inference for forcing-specific responses averaged over local anomaly intervals, following the projection approach of \citet{seong2021functional} and \citet{Nam2025}.

\indent 
We apply the approach to anomaly densities from the HadCRUT5 monthly gridded temperature record and historical effective radiative forcing reconstructions split into CO$_2$ and non-CO$_2$ portfolios. A forcing-block diagnostic supports their treatment as distinct persistent coordinates. At the 5\% level, the between-cointegration results are consistent with a long-run relationship in which the distinct forcing trends account for the persistent evolution of the anomaly distribution. We then ask how strongly the distribution responds to each portfolio across anomaly states. The estimated response functions differ substantially in magnitude and shape, and inference on their difference rejects the common-response restriction imposed by aggregation, revealing heterogeneity concealed by scalar aggregation.

Substantively, an increase in the CO$_2$ forcing portfolio is associated mainly with mean and location shifts and greater central concentration, whereas an increase in the non-CO$_2$ portfolio is associated with greater dispersion. For CO$_2$, the cold- and warm-tail probability changes nearly offset; for non-CO$_2$, combined extreme-state probability increases primarily on the warm side. The mean response alone therefore misses these contrasting changes in dispersion and the allocation of probability across extreme anomaly states.

\indent The paper proceeds as follows. Section~\ref{sec:empirical_motivation} describes data, documents the distributional evolution of temperature anomalies, and characterizes the information lost through scalar aggregation. Section~\ref{sec:Metric_model} develops testing, estimation, and inference methods. Section~\ref{sec:distributional_responses} presents specification evidence, forcing-specific responses, fixed-basis common-response comparison, and implications for distributional shape and tail mass. Section~\ref{conclude} concludes. The Supplement contains proofs, technical assumptions, implementation details, and robustness checks.

\section{Data and the cost of scalar aggregation}
\label{sec:empirical_motivation}

\noindent This section specifies the two sides of the long-run relation and formalizes what is assumed away when either side is reduced to a scalar. The benchmark throughout is the scalar climate-econometric specification linking the global mean temperature anomaly to aggregate radiative forcing \citep[e.g.,][]{dergiades2016long,pretis2020econometric}. Section~\ref{subsec:response_data} constructs the density-valued response and forcing predictor vector, while Section~\ref{subsec:aggregation_cost} discusses the restrictions implicit in scalar reductions of the response and predictor.

\subsection{Response and predictor construction}
\label{subsec:response_data}

\noindent\textbf{The density-valued response.} We construct annual densities of temperature anomalies from observed grid-cell-month data using the ensemble mean of the non-infilled HadCRUT5 record for 1850--2024 \citep{morice2021updated}. The anomalies are measured in degrees Celsius relative to the 1961--1990 climatology. For each year $t$, we pool the available anomalies with equal weights and estimate $f_t$ by kernel smoothing. Each estimate is restricted to and renormalized on a common, tail-trimmed support. Section~\ref{sec_app_density_error} of the Supplement provides the data and numerical details.

\indent Following \citet{chang2020evaluating}, we use non-infilled observations because our estimand is the observed cell-month anomaly distribution; infilling would make dispersion and tail behavior reflect the spatial reconstruction model as well as the observations \citep{morice2021updated}. Observed spatial coverage expands over time, however, changing the locations represented in the annual densities. Section~\ref{sec_app_coverage} of the Supplement therefore examines sensitivity to this changing coverage.

\indent To use these densities in a functional linear model, we map them into a Hilbert space by the CLR transformation. For any positive density $f$ on the common support $[a,b]$, define
\begin{equation}\label{eq_clr_response} \operatorname{clr}(f)(s)=\log f(s)-\frac{1}{b-a}\int_a^b\log f(u)\,du,\qquad f(s)=\frac{\exp{\operatorname{clr}(f)(s)}}{\int_a^b\exp{\operatorname{clr}(f)(u)}\,du}. \end{equation}
For each year, the empirical functional response is $Y_t=\operatorname{clr}(f_t)\in\HY$. The CLR transformation is an isometry from the Bayes Hilbert space of densities onto the closed linear subspace $\{g\in L^2[a,b]:\int_a^b g(s)\,ds=0\}$ \citep{Egozcue2006}. The inverse in \eqref{eq_clr_response} shows that this representation retains all distributional information.
\begin{remark}[Density-estimation error]\label{rem_density_error}
Let $f_t^\circ$ denote the underlying density of temperature anomalies and $Y_t^\circ=\operatorname{clr}(f_t^\circ)$. We represent density-estimation error multiplicatively as $f_t(s)\propto f_t^\circ(s)\exp\{\eta_t(s)\}$, which is general for strictly positive densities. Equation~\eqref{eq_clr_response} then gives $Y_t=Y_t^\circ+e_t$, where $e_t(s)=\eta_t(s)-(b-a)^{-1}\int_a^b\eta_t(u)\,du$. Thus, the error enters additively on the CLR scale and is absorbed into the regression disturbance under the conditions in Section~\ref{sec_app_density_error} of the Supplement.
\end{remark}
\indent Panels~(a) and~(c)--(e) of Figure~\ref{Fig:Data_Desc} show that the evolution of the anomaly distribution is not a pure translation. The densities shift persistently to the right, but re-centering each year at its own cross-sectional mean does not collapse them onto a common shape: later densities are less concentrated than mid-sample densities and carry more mass on the warm shoulder. The difference panel makes the reallocation explicit, with a long-run loss of mass over moderately negative anomalies and a gain over near-zero, positive, and upper-tail states. Location, dispersion, and tail mass therefore all move over the sample, and a specification that tracks only the first moment cannot represent the last two.

\begin{figure}[t]
\centering
\captionsetup[subfigure]{font=small,skip=2pt}
\begin{subfigure}[b]{0.48\textwidth}
\centering
\includegraphics[width=\linewidth,trim={1.2cm 0.1cm 0.2cm 0.1cm},clip]{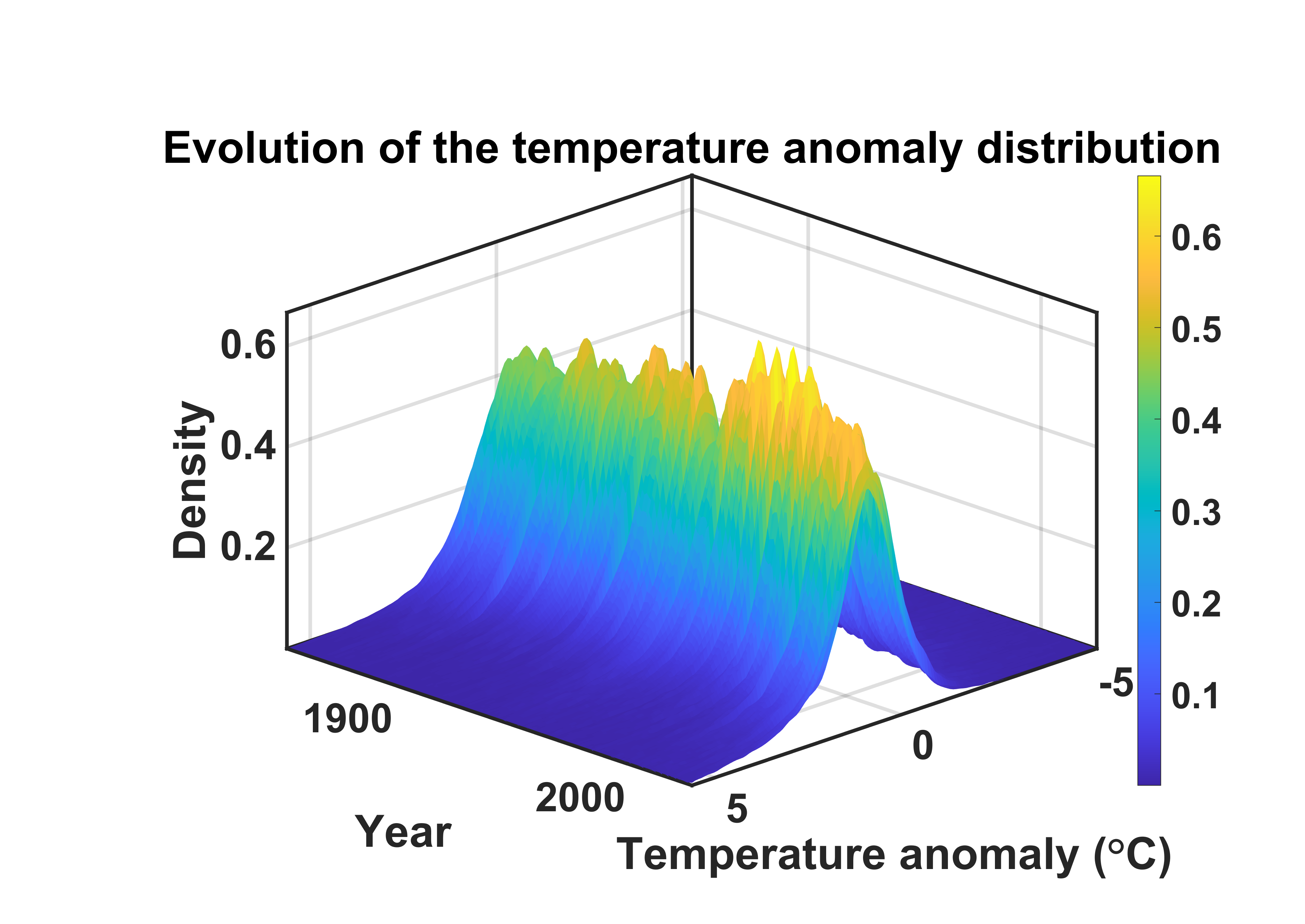}
\caption{Annual density surface}
\end{subfigure}
\hfill
\begin{subfigure}[b]{0.48\textwidth}
\centering
\includegraphics[width=\linewidth,trim={1.2cm 0.1cm 0.2cm 0.4cm},clip]{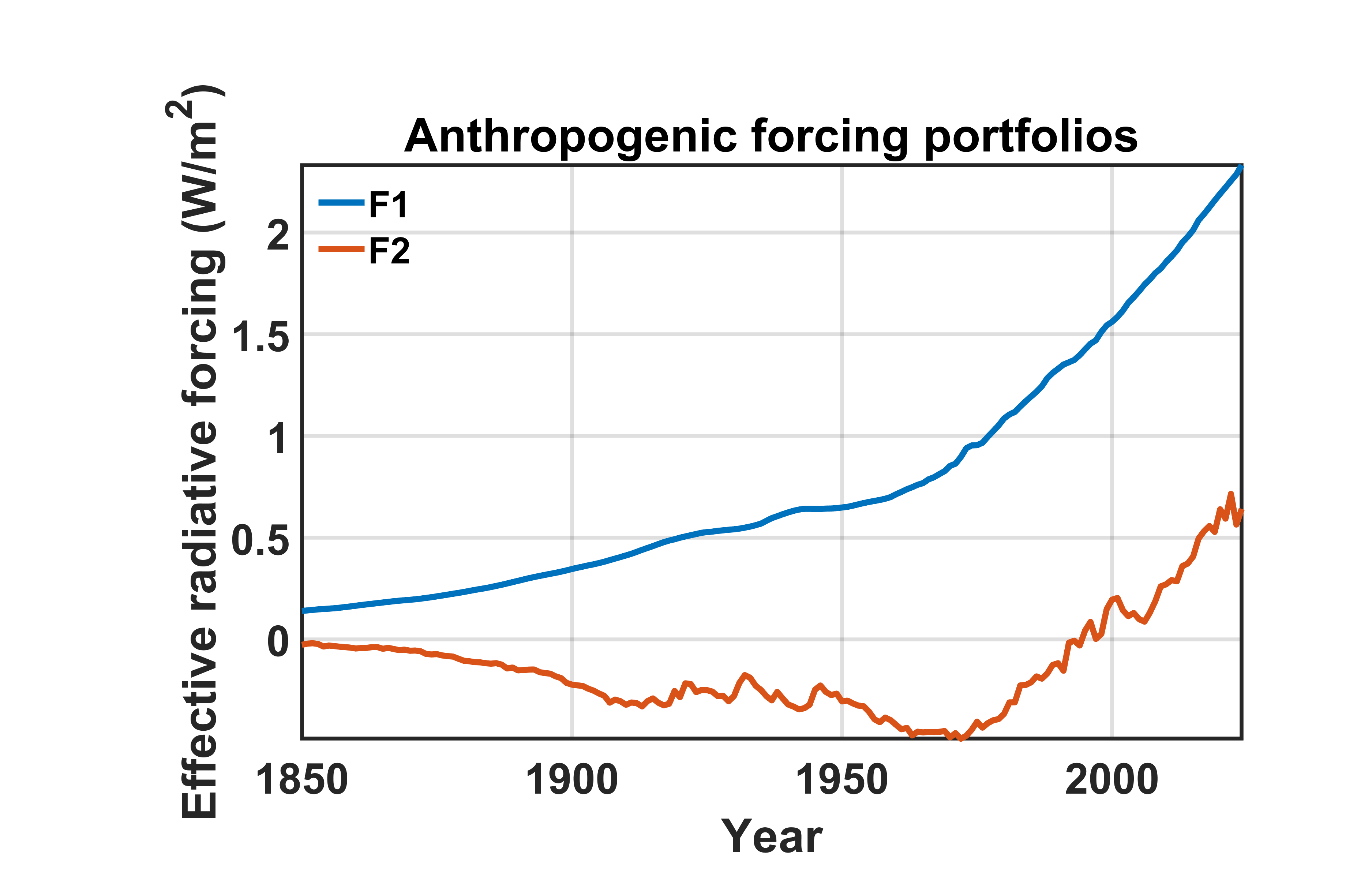}
\caption{Grouped anthropogenic ERF}
\end{subfigure}

\vspace{0.4em}

\begin{subfigure}[b]{0.327\textwidth}
\centering
\includegraphics[width=\linewidth,trim={1.3cm 0.1cm 0.3cm 0.1cm},clip]{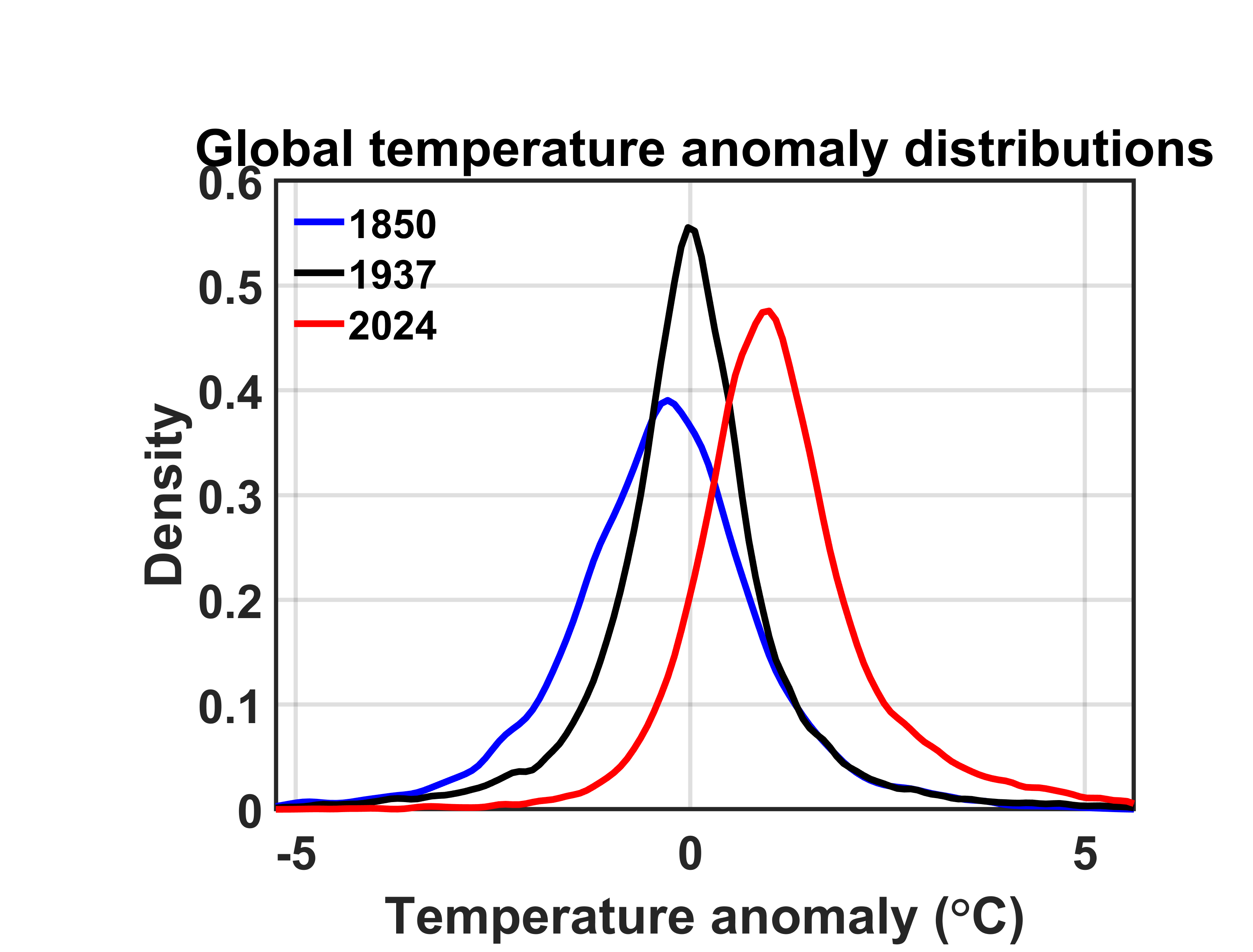}
\caption{Representative years}
\end{subfigure}
\hfill
\begin{subfigure}[b]{0.327\textwidth}
\centering
\includegraphics[width=\linewidth,trim={1.3cm 0.1cm 0.3cm 0.1cm},clip]{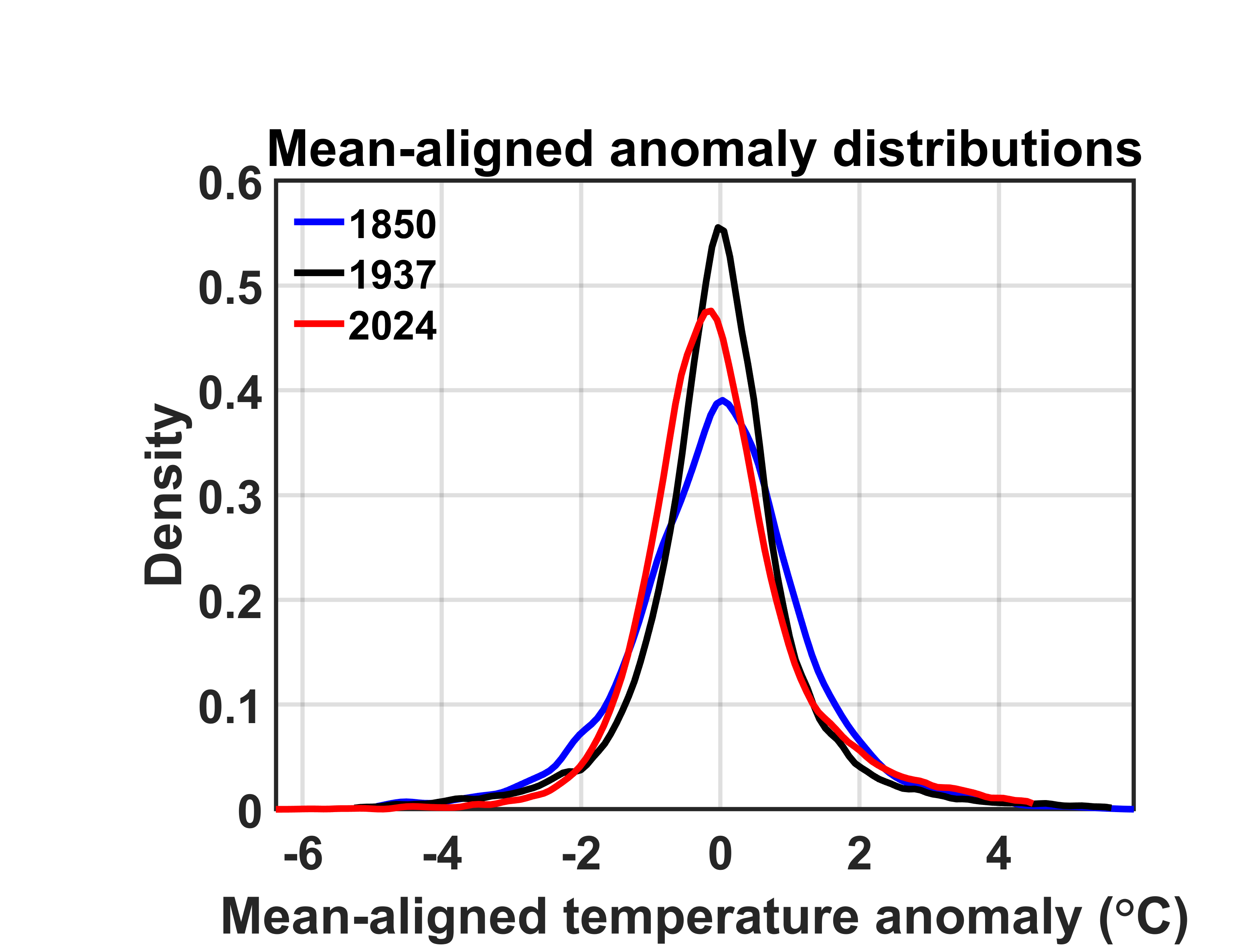}
\caption{Mean-aligned densities}
\end{subfigure}
\hfill
\begin{subfigure}[b]{0.327\textwidth}
\centering
\includegraphics[width=\linewidth,trim={1.3cm 0.1cm 0.3cm 0.1cm},clip]{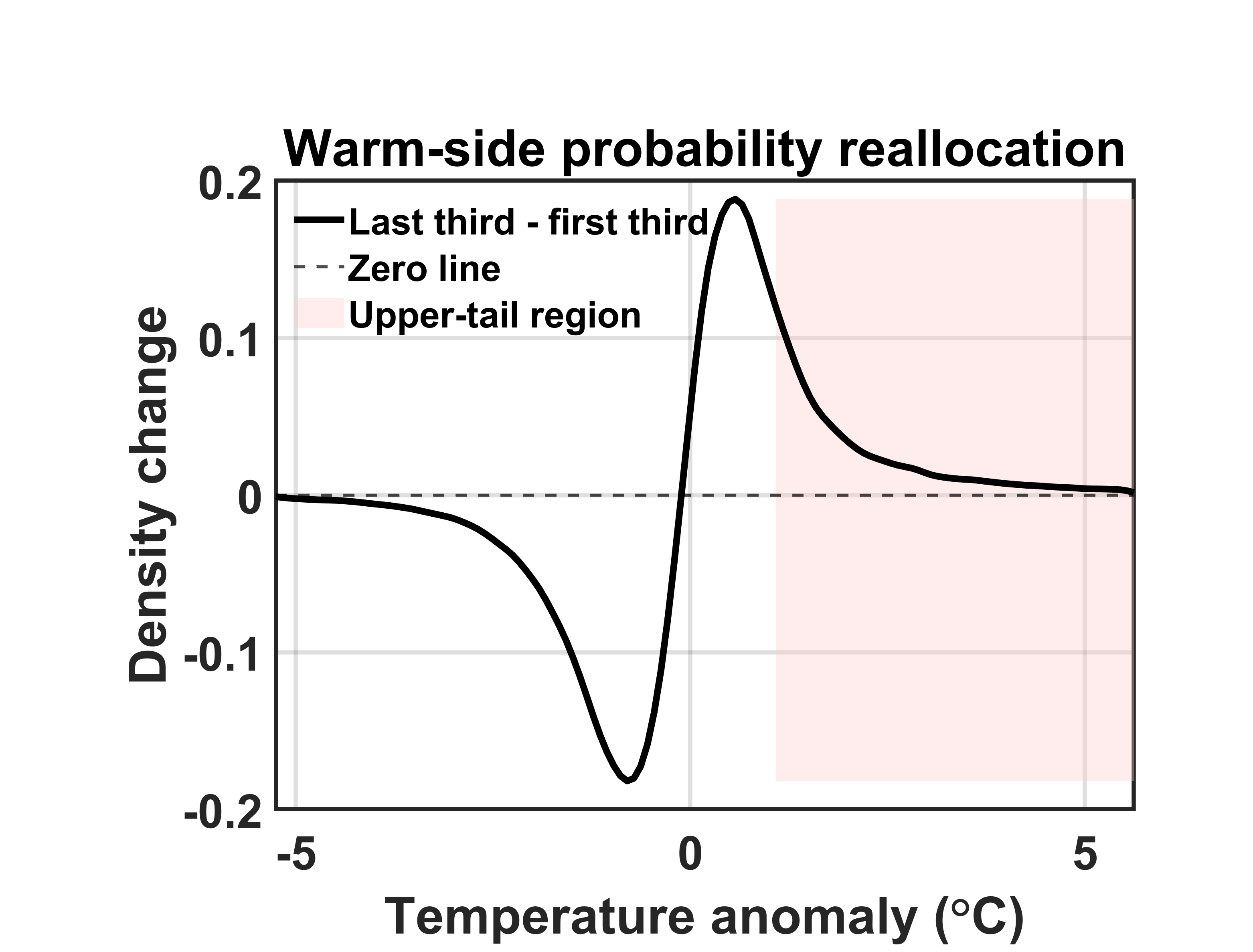}
\caption{Late-minus-early difference}
\end{subfigure}
\caption{Temperature-anomaly densities and anthropogenic forcing, 1850--2024. F1 is CO$_2$ forcing; F2 is the sum of the non-CO$_2$ anthropogenic components listed in Table~\ref{Tab:ForcingAggregation}.}
\label{Fig:Data_Desc}
\end{figure}

\medskip
\noindent\textbf{The forcing predictor vector.} The forcing series are annual best-estimate reconstructions of Effective Radiative Forcing (ERF) from \citet{forster2024indicators}, as updated by \citet{Smith2025IndicatorsGCC} (v2025.06.25), over 1850--2024. The series follow the disaggregated accounting of IPCC AR6 \citep{IPCC_AR6_WGI_AnnexIII_2021} and are measured in W/m$^2$ relative to the 1750 baseline. Solar variability and volcanic aerosol forcing are excluded, so the predictor vector isolates persistent anthropogenic forcing rather than externally driven or episodic natural variation.

\begin{table}[t]
\centering
\caption{Anthropogenic forcing portfolios}
\label{Tab:ForcingAggregation}
\renewcommand{\arraystretch}{1.18}
\setlength{\tabcolsep}{5pt}
\small
\begin{tabularx}{\textwidth}{
>{\RaggedRight\arraybackslash}p{3.3cm}
>{\RaggedRight\arraybackslash}p{4.9cm}
>{\RaggedRight\arraybackslash\footnotesize}X}
\toprule
Group & Forcing components & Statistical role in the empirical design \\
\midrule
\makecell[tl]{F1: CO$_2$\\ forcing}
&
\makecell[tl]{CO$_2$}
&
The dominant anthropogenic greenhouse-gas component in the IPCC effective radiative forcing accounting. This portfolio provides the main low-frequency CO$_2$ forcing benchmark.
\\
\midrule
\makecell[tl]{F2: Non-CO$_2$\\ anthropogenic\\ forcing}
&
\makecell[tl]{CH$_4$, N$_2$O,\\ aerosol--radiation interactions,\\ aerosol--cloud interactions,\\ O$_3$, contrails, land-use\\ change, BC on snow,\\ H$_2$O$_{\text{strat}}$, halogenated species}
&
Anthropogenic forcing outside the CO$_2$ block. This portfolio collects non-CO$_2$ greenhouse gases, reactive-chemistry, aerosol, cloud-adjustment, surface-albedo, snow-albedo, and aviation-related forcing channels.
\\
\bottomrule
\end{tabularx}
\end{table}

With $T=175$ annual time points, estimating eleven component-specific functional responses would be imprecise because the forcing series share substantial low-frequency variation. We therefore group the components into the two portfolios in Table~\ref{Tab:ForcingAggregation}. F1 contains CO$_2$ forcing, whereas F2 sums the remaining ten anthropogenic components. Keeping CO$_2$ separate also provides a forcing-based counterpart to the CO$_2$-only greenhouse-gas specification of \citet{magnus2011global}, while F2 preserves a separate non-CO$_2$ forcing margin. This specification allows F1 and F2 to have distinct response functions; it does not identify component-specific responses within F2.

\indent Panel~(b) of Figure~\ref{Fig:Data_Desc} shows that F1 and F2 contain distinct low-frequency information. F1 remains positive and follows a smooth, nearly monotone upward path, whereas F2 is nonmonotone: negative aerosol forcing dominates much of the twentieth century, while non-CO$_2$ greenhouse gases and other positive components dominate in the recent period. Thus, despite a sample correlation in levels of $0.6109$, their historical paths differ substantially. The no-within-cointegration diagnostic reported in Section~\ref{sec_test_stat} formally examines whether F1 and F2 carry two stochastic trends. Aggregating them into an aggregate anthropogenic forcing index would collapse the two coordinates into their sum, netting their paths where their signs differ and imposing a single response function on the temperature density.

\subsection{What scalar aggregation misses: a functional alternative}
\label{subsec:aggregation_cost}

\noindent The matched scalar benchmark applies the aggregate anthropogenic forcing regression to the mean of the same observed temperature anomaly density:
\begin{equation} \bar y_t=\alpha+\beta^{\mathrm{agg}}F_t^{\mathrm{agg}}+u_t,\label{eq:linear_regression_sec6} \end{equation}
where $\alpha$ is an intercept, $u_t$ is a scalar disturbance, $\bar y_t=\int_a^b s f_t(s)\,ds$, $F_t^{\mathrm{agg}}$ is aggregate anthropogenic radiative forcing, and $\beta^{\mathrm{agg}}$ is the long-run scalar slope, measured in $^\circ$C per W/m$^2$, associated with a 1 W/m$^2$ increase in aggregate anthropogenic forcing.

\indent Our formulation generalizes both sides of \eqref{eq:linear_regression_sec6} by replacing the scalar forcing index with a $\KX$-dimensional vector of forcing portfolios and the mean response with the full anomaly density in its CLR representation. Let $Y_t$ denote the CLR-transformed density and $\mathbf{x}_t=(x_{1,t},\ldots,x_{\KX,t})'$ the forcing vector, with $\KX=2$ in our application. The model developed in Section~\ref{sec:Metric_model} takes the form
\begin{equation} Y_t(s)=\beta_0(s)+\sum_{k=1}^{\KX}\beta_k(s)x_{k,t}+U_t(s),\qquad s\in[a,b].\label{eq:scalar_to_function_sec6} 
\end{equation} 
Here $\beta_0$ is a functional intercept, $\beta_k$ is the CLR response function of forcing portfolio $k$, and $U_t$ collects distributional variation not explained by the reconstructed forcing series. 


\indent Aggregating the forcing vector replaces $\sum_k\beta_k(s)x_{k,t}$ with $\beta^{\mathrm{clr}}(s)F_t^{\mathrm{agg}}$, where $F_t^{\mathrm{agg}}=\sum_{k=1}^{\KX}x_{k,t}$. The aggregate specification reproduces the vector model for every forcing path if and only if $\beta_k=\beta^{\mathrm{clr}}$ in $\HY$ for $k=1,\ldots,\KX$, with equality understood almost everywhere. Hence, predictor aggregation is lossless only if CO$_2$ and non-CO$_2$ forcing share a common CLR response; otherwise, the distributional response depends on forcing composition. Reducing the response to its mean entails a separate loss, retaining only the first-moment response and leaving distributional reshaping unidentified. Section~\ref{sec:distributional_responses} examines this restriction empirically.


The central statistical question is whether the stochastic trends in the CLR-transformed anomaly distribution are spanned by those in the forcing vector, leaving a stationary functional residual. The next section formalizes this long-run restriction as between-cointegration and develops procedures for testing it, estimating the long-run response operator, and conducting inference on the forcing-specific responses.

\section{Statistical Framework}\label{sec:Metric_model}
\noindent We begin from the possibility that both the transformed anomaly distribution and the forcing vector contain stochastic trends, a feature examined empirically in Section~\ref{sec_test_stat}. In that case, the regression relation in \eqref{eq:scalar_to_function_sec6} has a long-run interpretation only if its residual is stationary. In the present application, this condition means that no persistent movement in the anomaly distribution remains after accounting for the anthropogenic forcing vector. We formalize this restriction as \emph{between-cointegration}. The statistical analysis follows the same logic: we first test whether between-cointegration holds, then estimate the long-run response operator under that restriction, and finally conduct inference on the forcing-specific response functions.

\indent Let $Y_t$ denote the CLR-transformed anomaly density, taking values in the Hilbert space $\HY\subseteq L^2[a,b]$, and let $\mathbf{x}_t=(x_{1,t},\ldots,x_{\KX,t})'$ denote the forcing vector, taking values in the $\KX$-dimensional Euclidean space $\HX=\mathbb{R}^{\KX}$. For the theoretical development, we write the empirical specification in \eqref{eq:scalar_to_function_sec6} in operator form as
\begin{equation}\label{eqmodel1}
Y_t=\beta_0+B(\mathbf{x}_t)+U_t,\quad B(\mathbf{v})=\sum_{j=1}^{\KX}\beta_jv_j,\quad \mathbf{v}=(v_1,\ldots,v_{\KX})' \in\HX,
\end{equation}
where $B:\HX\rightarrow\HY$ is a linear response operator, $\beta_0\in\HY$ is a functional intercept, $U_t\in\HY$ is the unexplained component, and $\beta_j\in\HY$ is the CLR response function associated with forcing portfolio $j$. Between-cointegration holds when $U_t$ is stationary. Under this restriction, the stochastic trends in the anomaly distribution are fully spanned by those in the forcing vector; otherwise, a persistent distributional component remains unexplained and $B$ cannot be interpreted as a long-run response operator.

\indent The remainder of this section is organized as follows. Section~\ref{sec_model1} states the assumptions and formally defines between-cointegration; Section~\ref{sec_betcointeg} develops the residual-based test; and Section~\ref{sec_esti} presents estimation and inference under the maintained long-run relation. The paper then applies these procedures to the climate data in Section~\ref{sec:distributional_responses}.

\subsection{Model and assumption}\label{sec_model1}

\noindent Recall that $Y_t$ takes values in $\HY\subseteq L^2[a,b]$ and $\mathbf{x}_t$ in $\HX=\mathbb{R}^{\KX}$, equipped with the $L^2$ and standard Euclidean inner products, respectively. We write $\langle\cdot,\cdot\rangle$ for either inner product. For $v_1\in\mathcal H_1$ and $v_2\in\mathcal H_2$, with $\mathcal H_1,\mathcal H_2\in\{\HX,\HY\}$, define the rank-one operator $v_1\otimes v_2:\mathcal H_1\rightarrow\mathcal H_2$ by $(v_1\otimes v_2)(\cdot)=\langle v_1,\cdot\rangle v_2$. The identity operator on the relevant space is denoted by $I$. For any bounded linear operator $A:\mathcal H_1\rightarrow\mathcal H_2$, let $A^\ast:\mathcal H_2\rightarrow\mathcal H_1$ denote its adjoint, defined by $\langle Au,v\rangle=\langle u,A^\ast v\rangle$ for all $u\in\mathcal H_1$ and $v\in\mathcal H_2$. Further notation and mathematical preliminaries are collected in Section~\ref{Sec_prelim}.

\subsubsection{Stochastic trends and within-cointegration}\label{sec_basic}
Following work on nonstationary functional time series \citep[e.g.,][]{Chang2016152,BSS2017}, we assume a finite-dimensional stochastic-trend decomposition for the CLR-transformed anomaly-density process; Section~\ref{AP_FTS} of the Supplement gives formal conditions. Specifically, let $P^N$ and $P^S=I-P^N$ denote orthogonal projections on $\HY$, with $P^N$ of finite rank. We write
\begin{equation}\label{eqtimedecom}
Y_t=\mu_Y+Y_t^N+Y_t^S,\qquad Y_t^N=P^N(Y_t-\mu_Y),\qquad Y_t^S=P^S(Y_t-\mu_Y).
\end{equation}
Here, $\mu_Y$ is a deterministic functional level, representing the common or initial shape of the functional observations. The component $Y_t^N$ captures the persistent, nonstationary dynamics driven by stochastic trends arising from the accumulation of stationary random elements, while $Y_t^S$ denotes the stationary, mean-reverting component. Throughout this paper, we use the terms nonstationarity and persistence interchangeably. Although $Y_t$ resides in an infinite-dimensional space $\HY$, following the literature, we assume that its nonstationary component $Y_t^N$ is finite-dimensional (i.e., $P^N$ is a finite-rank projection). This assumption is not only empirically relevant (see e.g., \citealp{NSS}), but theoretically necessary to ensure feasible statistical inference with finite samples based on eigenanalysis to be discussed. Note that for any $v \in \ran P^S$, the inner product $\langle Y_t, v \rangle = \int Y_t(s)v(s)ds$, which can be viewed as a continuous linear combination, consists of a stationary sequence. In this context, $Y_t$ is said to be \textit{within-cointegrated}, reflecting an internal synchronization where individual nonstationary behaviors cancel each other out to reveal a stable long-run component.

\indent For the predictor $\mathbf{x}_t$, we assume that it is an $I(1)$ process of full nonstationary rank with a possibly nonzero deterministic level $\mu_X\in\HX$:
\begin{equation}\label{eqtimedecomx}
\mathbf{x}_t=\mu_X+\mathbf{x}_t^N.
\end{equation}
Here $\mathbf{x}_t^N:=\mathbf{x}_t-\mu_X$ denotes its nonstationary component, generated by the accumulation of stationary increments. Specifically, we assume that $\mathbf{x}_t$ admits no stationary linear combination $\langle\mathbf{x}_t,v\rangle$ for any nonzero $v\in\HX$. This implies that $\mathbf{x}_t$ is characterized by $\KX$-dimensional persistent dynamics that are potentially aligned with the nonstationary behavior of $Y_t$. We focus on this purely nonstationary case because it matches our empirical application, where the no-within-cointegration diagnostic in Section~\ref{sec_test_stat} supports treating the forcing portfolios as distinct persistent coordinates.


\subsubsection{Statistical formulation of between-cointegration}\label{sec_between}

\noindent We next define the long-run relation between the density-valued response and the forcing vector. Within-cointegration concerns stationary linear combinations within a single multivariate or functional process. By contrast, between-cointegration concerns whether the stochastic trends in one process are accounted for by those in another. In the present setting, this means that the persistent component of $Y_t$ is spanned by the persistent components of $\mathbf{x}_t$.

Formally, we say that $Y_t$ and $\mathbf{x}_t=(x_{1,t},\ldots,x_{\KX,t})'$ are between-cointegrated if there exist nonrandom elements $\gamma_1,\ldots,\gamma_{\KX}\in\HY$ such that \begin{equation}\label{eq:bet_cointeg} Y_t^N-\sum_{j=1}^{\KX}\gamma_jx_{j,t} \end{equation} is stationary in $\HY$. Let $x_{j,t}^N$ denote the $j$th coordinate of $\mathbf{x}_t^N$. Since $\mu_X$ is deterministic, between-cointegration can equivalently be defined using $x_{j,t}^N$ in place of $x_{j,t}$ in \eqref{eq:bet_cointeg}. Under this condition, the long-run relation takes the form \eqref{eqmodel1}, with $\beta_j=\gamma_j$ for $j=1,\ldots,\KX$ and stationary $U_t$. In the climate application, this means that after accounting for the anthropogenic forcing vector, no unexplained persistent component remains in the CLR-transformed temperature-anomaly distribution. Thus, testing between-cointegration provides a specification check for the long-run relation before the forcing-specific response functions are estimated.

\subsection{Statistical test for between-cointegration}\label{sec_betcointeg}

\noindent We develop a residual-based test of between-cointegration. The central idea is to determine whether the stochastic trends in $Y_t$ are fully accounted for by projecting $Y_t$ on the nonstationary forcing vector $\mathbf{x}_t$. Under between-cointegration, the resulting residuals should contain only stationary variation; otherwise, they should retain an unexplained stochastic trend. Section~\ref{sec_app_test} of the Supplement presents the asymptotic theory and examines finite-sample properties by simulation (Section~\ref{sec_sim_between}); here we focus on practical implementation in three steps.
\begin{description}[style=nextline, leftmargin=1em, font=\bfseries,topsep=-2pt, itemsep=0pt]
\item[Step 1: Computing residuals from the least-squares projection] \hfill
We first compute the least-squares projection of $Y_t$ on $\mathbf{x}_t$, using demeaned variables to allow for a nonzero intercept. Let $\bar Y_T=T^{-1}\sum_{t=1}^T Y_t$ and $\bar{\mathbf{x}}_T=T^{-1}\sum_{t=1}^T\mathbf{x}_t$. Define $\widehat{C}_{\mathbf{x}\mathbf{x}}=T^{-1}\sum_{t=1}^T(\mathbf{x}_t-\bar{\mathbf{x}}_T)\otimes(\mathbf{x}_t-\bar{\mathbf{x}}_T)$ and $\widehat{C}_{Y\mathbf{x}}=T^{-1}\sum_{t=1}^T(\mathbf{x}_t-\bar{\mathbf{x}}_T)\otimes(Y_t-\bar Y_T)$. The least-squares projection map is $\Bpre=\widehat C_{Y\mathbf{x}}\widehat C_{\mathbf{x}\mathbf{x}}^{-1}$. Let $(\hat\lambda_{\mathbf{x},j},\hat v_{\mathbf{x},j})$, $j=1,\ldots,\KX$, denote the eigenvalue--eigenvector pairs of $\widehat C_{\mathbf{x}\mathbf{x}}$. The fitted component can then be computed as
\begin{equation}\label{eqprelimest}
\Bpre(\mathbf{x}_t-\bar{\mathbf{x}}_T)=\frac{1}{T}\sum_{s=1}^T\sum_{j=1}^{\KX}\hat\lambda_{\mathbf{x},j}^{-1}\langle \mathbf{x}_t-\bar{\mathbf{x}}_T,\hat v_{\mathbf{x},j}\rangle\langle \mathbf{x}_s-\bar{\mathbf{x}}_T,\hat v_{\mathbf{x},j}\rangle(Y_s-\bar Y_T).
\end{equation}
Intuitively, the fitted component represents the part of the transformed anomaly distribution that is accounted for by the forcing vector; the projection residuals $\Upre_t=(Y_t-\bar Y_T)-\Bpre(\mathbf{x}_t-\bar{\mathbf{x}}_T)$ summarize the remaining variation. Under between-cointegration, the disturbance $U_t$ is stationary, so $\Upre_t$ should contain no unexplained stochastic trend. The fitted residuals need not themselves be stationary in finite samples because they depend on an estimated projection and sample means; the plug-in calibration accounts for these effects.

\item[Step 2: Constructing residual-persistence diagnostics] \hfill
To measure the persistence remaining in $\Upre_t$, define the residual partial-sum operator
\[ \widehat{\mathcal K}=\frac{1}{T}\sum_{t=1}^T\left(\sum_{s=1}^t\Upre_s\right)\otimes\left(\sum_{s=1}^t\Upre_s\right). \]
This operator accumulates residual variation over time. Under between-cointegration, $T^{-1}\widehat{\mathcal K}$ remains stochastically bounded despite the estimation and demeaning effects in $\Upre_t$; under an alternative with a remaining stochastic trend, it diverges. We first consider the unnormalized diagnostic
\begin{equation}\label{eqteststat_K_main} \widehat{\mathcal T}_K=\Lambda_{\max}\left(T^{-1}\widehat{\mathcal K}\right), \end{equation}
where $\Lambda_{\max}(\cdot)$ denotes the largest eigenvalue. For a covariance-normalized diagnostic, we use the residual covariance operator $\widehat{\mathcal V}=T^{-1}\sum_{t=1}^T\Upre_t\otimes\Upre_t$ both to identify the direction of greatest contemporaneous residual variation and to normalize the residual scale. Let $\hat v_V$ be the unit eigenvector associated with the largest eigenvalue of $\widehat{\mathcal V}$. We define
\begin{equation}\label{eqteststat_V_main} \widehat{\mathcal T}_V=T^{-1}{\langle\widehat{\mathcal K}\hat v_V,\hat v_V\rangle}/{\langle\widehat{\mathcal V}\hat v_V,\hat v_V\rangle}. \end{equation}
This statistic measures accumulated residual persistence along the direction of greatest contemporaneous residual variation, normalized by the residual variance in that direction. The statistic $\widehat{\mathcal T}_K$ depends on the scale of $\Upre_t$. By contrast, $\widehat{\mathcal T}_V$ is invariant to rescaling of the residual process: replacing $\Upre_t$ by $c\Upre_t$ for any $c>0$ leaves it unchanged. Although the scale invariance of $\widehat{\mathcal T}_V$ can be convenient, both statistics are calibrated against their respective plug-in Monte Carlo null distributions and are reported as complementary diagnostics.
\item[Step 3: Interpreting the diagnostics and implementing the tests] \hfill
For practical interpretation, both diagnostics are right-tailed: large values indicate residual persistence inconsistent with between-cointegration. Under the null, $\widehat{\mathcal T}_K$ and $\widehat{\mathcal T}_V$ converge to finite limits; under the residual $I(1)$ alternative, they diverge at rates $T^2$ and $T$, respectively. The null limits are nonpivotal because they depend on the unknown long-run covariance structure and the estimation effect induced by projecting $Y_t$ on the integrated forcing vector. We therefore use the plug-in Monte Carlo procedure described in Section~\ref{sec_app_test_implementation} of the Supplement. At significance level $\alpha$, this procedure yields $\widehat q_{K,\alpha}$ and $\widehat q_{V,\alpha}$, the simulated $(1-\alpha)$ critical values for $\widehat{\mathcal T}_K$ and $\widehat{\mathcal T}_V$, respectively, together with the corresponding upper-tail $p$-values. The procedure uses kernel estimates of the relevant long-run and one-sided covariance operators, constructed from $\Delta\mathbf{x}_t$ and $\Upre_t$, to simulate the joint Brownian processes entering the null limits. The tests reject when $\widehat{\mathcal T}_K>\widehat q_{K,\alpha}$ and $\widehat{\mathcal T}_V>\widehat q_{V,\alpha}$, respectively. Section~\ref{sec_app_test} of the Supplement establishes the validity of this approximation and the consistency of both tests.

\end{description}

\noindent Although these diagnostics use the partial-sum logic familiar from existing stationarity tests for functional time series, their target differs. They do not test whether $Y_t$ is stationary; rather, they test whether any stochastic trend remains after projecting $Y_t$ on the forcing vector. They thus provide a specification check for the long-run relation. When between-cointegration is not rejected, we estimate and conduct inference on the forcing-specific response functions under this maintained restriction. Rejection would indicate that the forcing vector does not account for all persistent variation in the temperature-anomaly distribution.

\subsection{Estimation under between-cointegration}\label{sec_esti}
\noindent We now turn to estimation under the maintained hypothesis of between-cointegration, continuing to assume that $\mathbf{x}_t$ has full nonstationary rank and hence no within-cointegration. This condition is assessed for the forcing data in Section~\ref{sec_test_stat}. Under between-cointegration, $U_t$ in \eqref{eqmodel1} is stationary and $B$ represents the long-run response operator. To remove the functional intercept $\beta_0$, write the model in demeaned form as
\begin{equation}\label{eqmodel1a}
Y_t-\bar Y_T=B(\mathbf{x}_t-\bar{\mathbf{x}}_T)+(U_t-\bar U_T),
\end{equation}
where $\bar Y_T$, $\bar{\mathbf{x}}_T$, and $\bar U_T$ denote the corresponding temporal sample means.
A natural starting point is the least-squares projection estimator $\Bpre=\widehat C_{Y\mathbf{x}}\widehat C_{\mathbf{x}\mathbf{x}}^{-1}$. The centered analogue of the result in Remark~\ref{rem_ls_limit} of the Supplement gives $\|\Bpre-B\|_{\op}=O_p(T^{-1})$. The corresponding asymptotic expansion of $T(\Bpre-B)$ contains nuisance-dependent bias terms arising from endogeneity between $U_t$ and $\Delta\mathbf{x}_t$ and serial dependence in their joint process. We therefore extend the fully modified least-squares estimator of \citet{phillips1995fully} to this setting, removing these terms and obtaining a limit suitable for feasible simulation-based inference.

\subsubsection{Computation of the proposed estimator}\label{sec_compest}

Formal assumptions, consistency results, and the limiting distribution of the estimator are established in Section~\ref{sec_app_est1} of the Supplement, while Section~\ref{sec_app_det2} details the intercept extension used here. We focus on its practical computation, which proceeds in three steps:

\begin{description}[style=nextline, leftmargin=1em, font=\bfseries,topsep=-2pt, itemsep=0pt]
\item[Step 1: Preliminary estimation and residuals] \hfill
We begin with the preliminary least-squares estimator $\Bpre$ and compute the projection residuals $\Upre_t=(Y_t-\bar Y_T)-\Bpre(\mathbf{x}_t-\bar{\mathbf{x}}_T)$. These residuals are used in place of the unobserved disturbance $U_t$ when estimating the covariance operators for the fully modified correction.

\item[Step 2: Estimation of long-run covariance operators] \hfill
The fully modified estimator uses the long-run and one-sided long-run covariance operators $\Omega_{\mathbf{x}\mathbf{x}}=\sum_{j=-\infty}^{\infty}\mathbb E[\Delta\mathbf{x}_t\otimes\Delta\mathbf{x}_{t+j}]$, $\Omega_{U\mathbf{x}}=\sum_{j=-\infty}^{\infty}\mathbb E[\Delta\mathbf{x}_t\otimes U_{t+j}]$, $\Omega_{\mathbf{x}\mathbf{x}}^+=\sum_{j=0}^{\infty}\mathbb E[\Delta\mathbf{x}_t\otimes\Delta\mathbf{x}_{t+j}]$, and $\Omega_{U\mathbf{x}}^+=\sum_{j=0}^{\infty}\mathbb E[\Delta\mathbf{x}_t\otimes U_{t+j}]$. These operators determine the nuisance-dependent terms in the limiting distribution of $\Bpre$ and the corrections used by the fully modified estimator, as developed in Sections~\ref{sec_app_est1} and~\ref{sec_app_det2} of the Supplement. We estimate them using
\begin{align}
\widehat{\Omega}_{\mathbf{x}\mathbf{x}}&=\sum_{|j|\leq h}\mathrm{k}(j/h)\widehat{\Gamma}_{\mathbf{x}\mathbf{x}}^{(j)},&
\widehat{\Omega}_{U\mathbf{x}}&=\sum_{|j|\leq h}\mathrm{k}(j/h)\widehat{\Gamma}_{U\mathbf{x}}^{(j)},\label{eqsample1}\\
\widehat{\Omega}_{\mathbf{x}\mathbf{x}}^+&=\sum_{j=0}^{h}\mathrm{k}(j/h)\widehat{\Gamma}_{\mathbf{x}\mathbf{x}}^{(j)},&
\widehat{\Omega}_{U\mathbf{x}}^+&=\sum_{j=0}^{h}\mathrm{k}(j/h)\widehat{\Gamma}_{U\mathbf{x}}^{(j)}.\label{eqsample2}
\end{align}
Under our rank-one-operator convention, the lag-$j$ sample autocovariance and cross-covariance operators are $\widehat{\Gamma}_{\mathbf{x}\mathbf{x}}^{(j)}=T^{-1}\sum_{t=(1-j)\vee1}^{T\wedge(T-j)}\Delta\mathbf{x}_t\otimes\Delta\mathbf{x}_{t+j}$ and $\widehat{\Gamma}_{U\mathbf{x}}^{(j)}=T^{-1}\sum_{t=(1-j)\vee1}^{T\wedge(T-j)}\Delta\mathbf{x}_t\otimes\Upre_{t+j}$, respectively, so $\widehat{\Gamma}_{U\mathbf{x}}^{(j)}$ maps $\HX$ to $\HY$. The function $\mathrm{k}(\cdot)$ is a kernel and $h$ is its bandwidth. Assumption~\ref{assum_test_kernel} of the Supplement gives the corresponding conditions. Proposition~\ref{prop1} establishes operator-norm consistency for the baseline specification, and Section~\ref{sec_app_det2} provides the demeaned extension used here. We use the Parzen kernel in the empirical analysis.

\item[Step 3: Computation of the proposed estimator] \hfill
Using the covariance estimators from Step~2, define the modified response and bias-correction operator as
\begin{equation}\label{eqest01_main}
Z_{1,t}=(Y_t-\bar Y_T)-\widehat{\Omega}_{U\mathbf{x}}\widehat{\Omega}_{\mathbf{x}\mathbf{x}}^{-1}\Delta\mathbf{x}_t,\qquad \widehat{\Upsilon}=\widehat{\Omega}_{U\mathbf{x}}^+-\widehat{\Omega}_{U\mathbf{x}}\widehat{\Omega}_{\mathbf{x}\mathbf{x}}^{-1}\widehat{\Omega}_{\mathbf{x}\mathbf{x}}^+.
\end{equation}
The modification of $Y_t-\bar Y_T$ removes the estimated long-run covariance between the disturbance component and the predictor innovations, while $\widehat{\Upsilon}$ corrects the remaining bias. Under the maintained full-rank condition, $\widehat{\Omega}_{\mathbf{x}\mathbf{x}}^{-1}$ is computed as an ordinary matrix inverse on the finite-dimensional space $\HX$. The proposed fully modified estimator is
\begin{equation}\label{eq_B_hat}
\widehat B=(\widehat C_{Z_1\mathbf{x}}-\widehat{\Upsilon})\widehat C_{\mathbf{x}\mathbf{x}}^{-1},\qquad \widehat C_{Z_1\mathbf{x}}=\frac{1}{T}\sum_{t=1}^T(\mathbf{x}_t-\bar{\mathbf{x}}_T)\otimes Z_{1,t}.
\end{equation}
For explicit computation, let $\widehat C_{\mathbf{x}\mathbf{x}}=\sum_{r=1}^{\KX}\widehat\lambda_{\mathbf{x},r}\widehat v_{\mathbf{x},r}\otimes\widehat v_{\mathbf{x},r}$ be the eigendecomposition used in Section~\ref{sec_betcointeg}. Then, for any $v\in\HX$,
\begin{equation}\label{eq_B_hat_explicit}
\widehat B(v)=\sum_{r=1}^{\KX}\widehat\lambda_{\mathbf{x},r}^{-1}\langle v,\widehat v_{\mathbf{x},r}\rangle\left\{\frac{1}{T}\sum_{t=1}^T\langle\mathbf{x}_t-\bar{\mathbf{x}}_T,\widehat v_{\mathbf{x},r}\rangle Z_{1,t}-\widehat{\Upsilon}(\widehat v_{\mathbf{x},r})\right\}.
\end{equation}
This expression parallels the preliminary least-squares projection in \eqref{eqprelimest}, with the modified response and bias correction incorporated explicitly. The forcing-specific response functions are obtained as $\widehat\beta_j=\widehat B(e_j)$, where $e_j\in\HX$ is the $j$th standard basis vector, with one in its $j$th coordinate and zeros elsewhere.
\end{description}

\noindent Theorem~\ref{thmapp2} in the Supplement establishes that $\|\widehat B-B\|_{\op}=O_p(T^{-1})$ under the specification with an intercept and derives the corresponding nonstandard limiting distribution. The next subsection uses a feasible simulation of this limit to conduct inference on local averages of the forcing-specific response functions.



\subsubsection{Simulation-based inference for local average responses}\label{sec_inference}
\noindent A direct confidence region for the full response operator $B$ would consist of maps from the forcing space $\HX$ into the CLR-response space $\HY$, while a coefficient-specific region for $\beta_j=B(e_j)$ lies in the infinite-dimensional space $\HY$. Neither is readily interpretable in applied work. Constructing a conventional uniform confidence band for $\beta_j$ would require additional mathematical structure and stronger assumptions. We therefore conduct inference on average CLR responses over selected regions of the temperature-anomaly support, which provide directly interpretable scalar summaries of the forcing-specific response functions.

\indent For a subinterval $[a_k,b_k]$ of the full temperature-anomaly support $[a,b]$, define the rectangular weight
\begin{equation}\label{eq_local_weight} w_k(s)=(b_k-a_k)^{-1}\mathbf{1}\{s\in[a_k,b_k]\}. \end{equation}
The weight itself need not be centered: for any $g\in\HY$, $\langle g,w_k\rangle=\langle g,w_k-(b-a)^{-1}\rangle$, so only its centered projection enters the inner product. Hence, $\langle\beta_j,w_k\rangle=(b_k-a_k)^{-1}\int_{a_k}^{b_k}\beta_j(s)\,ds$ is the average CLR response to forcing portfolio $j$ over the anomaly range $[a_k,b_k]$ and has a direct interpretation. Its estimator is $\langle\widehat\beta_j,w_k\rangle$. Let $e_j$ be the $j$th standard basis vector of $\HX$, so that $\beta_j=B(e_j)$. Applying Theorem~\ref{thmapp2} in the Supplement with $v=e_j$ and the centered projection of $w_k$, and using the preceding equality, gives
\begin{equation}\label{eq_local_limit} T\langle\widehat\beta_j-\beta_j,w_k\rangle\to_d\left\langle\left(\int_0^1 W^c_{\mathbf{x}}(r)\otimes dW_{U|\mathbf{x}}(r)\right)\left(\int_0^1 W^c_{\mathbf{x}}(r)\otimes W^c_{\mathbf{x}}(r)\,dr\right)^{-1}e_j,w_k\right\rangle. \end{equation}
Here $\to_d$ denotes convergence in distribution. The processes $W_{\mathbf{x}}$ and $W_U$ are the joint Brownian limits of the $T^{-1/2}$-scaled partial sums of $\Delta\mathbf{x}_t$ and $U_t$, taking values in $\HX$ and $\HY$, respectively. The demeaned path is $W^c_{\mathbf{x}}(r)=W_{\mathbf{x}}(r)-\int_0^1W_{\mathbf{x}}(u)\,du$, whereas $W_{U|\mathbf{x}}(r)=W_U(r)-\Omega_{U\mathbf{x}}\Omega_{\mathbf{x}\mathbf{x}}^{-1}W_{\mathbf{x}}(r)$ is the uncentered error Brownian motion independent of $W_{\mathbf{x}}$.

\indent Although the limiting distribution in \eqref{eq_local_limit} does not have closed-form quantiles, Theorem~\ref{thmapp3} in the Supplement shows that it can be consistently approximated using estimated covariance eigenelements and simulated scalar Brownian motions. Let $\widehat q_{\tau}(e_j,w_k)$ denote the simulated $\tau$ quantile of this approximation. An asymptotic $(1-\alpha)$ equal-tailed confidence interval for $\langle\beta_j,w_k\rangle$ is
\begin{equation}\label{eqlocalci} \operatorname{CI}_j(1-\alpha,w_k)=\left[\langle\widehat\beta_j,w_k\rangle-{\widehat q_{1-\alpha/2}(e_j,w_k)}/T,\ \langle\widehat\beta_j,w_k\rangle-{\widehat q_{\alpha/2}(e_j,w_k)}/{T}\right]. \end{equation}
\indent Repeating this calculation over a collection of subintervals produces a local-average confidence display for the forcing-specific response function \citep{seong2021functional,Nam2025}. It shows where the average CLR response over an anomaly range is individually distinguishable from zero at the specified confidence level. Each interval has asymptotic marginal coverage for its corresponding local average; without a multiplicity adjustment, the collection should not be interpreted as a simultaneous confidence band for the entire function.

\smallskip
\noindent\textbf{Monte Carlo implementation.} The feasible quantiles $\widehat q_{\tau}(e_j,w_k)$ are computed using the following three-step simulation. The asymptotic justification for this procedure is provided in Section~\ref{sec_app_det3} of the Supplement.
\begin{description}[style=nextline, leftmargin=1em, font=\bfseries,topsep=-2pt, itemsep=0pt]
\item[Step 1: Spectral decomposition of covariance operators] \hfill
Using the residuals $\Upre_t$ defined above, we estimate their long-run covariance operator by
\begin{equation}\label{eq_lrv_UU} \widehat{\Omega}_{UU}=\sum_{|j|\leq h}\mathrm{k}(j/h)\widehat{\Gamma}_{UU}^{(j)},\qquad \widehat{\Gamma}_{UU}^{(j)}=\frac{1}{T}\sum_{t=(1-j)\vee1}^{T\wedge(T-j)}\Upre_t\otimes\Upre_{t+j}. \end{equation}
We then construct the conditional long-run covariance operator $\widehat{\Omega}_{U|\mathbf{x}}=\widehat{\Omega}_{UU}-\widehat{\Omega}_{U\mathbf{x}}\widehat{\Omega}_{\mathbf{x}\mathbf{x}}^{-1}\widehat{\Omega}_{\mathbf{x}U}$, where $\widehat{\Omega}_{\mathbf{x}U}=\widehat{\Omega}_{U\mathbf{x}}^*$. Let $(\widehat\lambda_j,\widehat v_j)$ and $(\widehat\mu_j,\widehat w_j)$ denote the eigenpairs of $\widehat{\Omega}_{U|\mathbf{x}}$ and $\widehat{\Omega}_{\mathbf{x}\mathbf{x}}$, respectively. We use the decompositions
\begin{equation}\label{eqdecomstep} \widehat{\Omega}_{U|\mathbf{x}}\approx\sum_{j=1}^{M}\widehat\lambda_j\widehat v_j\otimes\widehat v_j,\qquad \widehat{\Omega}_{\mathbf{x}\mathbf{x}}=\sum_{j=1}^{\KX}\widehat\mu_j\widehat w_j\otimes\widehat w_j. \end{equation}
The first decomposition retains the leading $M$ eigenpairs of an operator acting on the infinite-dimensional space $\HY$, whereas the second requires no truncation because $\HX$ has dimension $\KX$. These eigenelements are obtained by functional eigenanalysis and matrix eigendecomposition, respectively. In finite samples, $M$ is chosen to capture the dominant variation. Asymptotically, $M=M_T$ increases slowly enough with $T$ to control eigenelement estimation error, as required by Theorem~\ref{thmapp3} of the Supplement.

\item[Step 2: Generation of synthetic Brownian paths] \hfill
For each Monte Carlo replication $\ei=1,\ldots,R_{\mathrm{MC}}$, generate standard scalar Brownian motions $\{W_{1,j,(\ei)}\}_{j=1}^{M}$ and $\{W_{2,j,(\ei)}\}_{j=1}^{\KX}$, independently across the two families and across replications, on a fine grid over $[0,1]$. To reproduce the demeaning of the integrated forcing variables, define $W^c_{2,j,(\ei)}(r)=W_{2,j,(\ei)}(r)-\int_0^1W_{2,j,(\ei)}(u)\,du$. The required synthetic paths are
\begin{equation}\label{eq_synthetic_BM} \widehat W_{U|\mathbf{x},(\ei)}(r)=\sum_{j=1}^{M}\widehat\lambda_j^{1/2}\widehat v_jW_{1,j,(\ei)}(r),\qquad \widehat W^c_{\mathbf{x},(\ei)}(r)=\sum_{j=1}^{\KX}\widehat\mu_j^{1/2}\widehat w_jW^c_{2,j,(\ei)}(r). \end{equation}
Only the forcing path $\widehat W^c_{\mathbf{x},(\ei)}$ is centered. No analogous centering is applied to $\widehat W_{U|\mathbf{x},(\ei)}$.

\item[Step 3: Monte Carlo approximation of the quantiles] \hfill
For a given anomaly range with weight $w_k$, define the scalar Brownian process $\widehat Z_{k,(\ei)}(r)=\langle\widehat W_{U|\mathbf{x},(\ei)}(r),w_k\rangle$, whose increments satisfy
\begin{equation}\label{eq_simulated_Z} d\widehat Z_{k,(\ei)}(r)=\sum_{\ell=1}^{M}\widehat\lambda_\ell^{1/2}\langle\widehat v_\ell,w_k\rangle\,dW_{1,\ell,(\ei)}(r). \end{equation}
Because $\HX=\mathbb R^{\KX}$, the operator $\int_0^1\widehat W^c_{\mathbf{x},(\ei)}(r)\otimes\widehat W^c_{\mathbf{x},(\ei)}(r)\,dr$ is represented by an ordinary $\KX\times\KX$ matrix. The simulated realization of the limiting random variable in \eqref{eq_local_limit} is therefore computed directly as
\begin{equation}\label{eq_simulated_local_limit} e_j'\left\{\int_0^1\widehat W^c_{\mathbf{x},(\ei)}(r)\widehat W^c_{\mathbf{x},(\ei)}(r)'\,dr\right\}^{-1}\int_0^1\widehat W^c_{\mathbf{x},(\ei)}(r)\,d\widehat Z_{k,(\ei)}(r). \end{equation}
The ordinary integrals are evaluated by Riemann sums on the simulation grid, and the stochastic integral is evaluated using Brownian increments and left-endpoint values of $\widehat W^c_{\mathbf{x},(\ei)}$. Repeating Steps~2 and~3 for $\ei=1,\ldots,R_{\mathrm{MC}}$ gives the empirical distribution of \eqref{eq_simulated_local_limit}. Its empirical $\tau$ quantile is $\widehat q_{\tau}(e_j,w_k)$; in particular, $\widehat q_{\alpha/2}(e_j,w_k)$ and $\widehat q_{1-\alpha/2}(e_j,w_k)$ are the two quantiles used in \eqref{eqlocalci}.

\end{description}


\section{Empirical Analysis of Distributional Responses}
\label{sec:distributional_responses}
\noindent This section first examines whether the CO$_2$ and non-CO$_2$ forcing portfolios carry distinct stochastic trends (no within-cointegration) and then whether their trends account for all stochastic trends in the anomaly-density process (between-cointegration). It then estimates the forcing-specific responses, uses their implied densities to describe changes in location, dispersion, and tail mass, examines the common-response restriction imposed by an aggregate anthropogenic forcing index, compares the fitted responses with matched scalar benchmarks, and localizes probability-mass reallocations across the anomaly support. Section~\ref{sec_app_coverage} of the Supplement examines whether the results are sensitive to changes in the spatial coverage of the temperature data over time.

\subsection{Anthropogenic forcing and persistent climate change}
\label{sec_test_stat}
The empirical specification involves two distinct questions, assessed in sequence. We first examine whether the forcing vector $\mathbf{x}_t=(\mathrm{F1}_t,\mathrm{F2}_t)'$ contains two distinct stochastic trends. If F1 and F2 are within-cointegrated, their persistent variation is driven by fewer than two stochastic trends, and they cannot be treated as distinct persistent coordinates. Given full nonstationary rank of $\mathbf{x}_t$, we then test whether its stochastic trends span all stochastic trends in the density process $Y_t$, leaving the disturbance $U_t$ in \eqref{eqmodel1} stationary.

Let $d_N(\mathbf{x}_t)$ denote the number of stochastic trends in $\mathbf{x}_t$, equivalently, the dimension of its nonstationary subspace. Since $\mathbf{x}_t$ is a bivariate vector with $\KX=2$, we apply the variance-ratio diagnostic of \citet{Breitung2002} directly to its demeaned observations. We test the null of no within-cointegration against the alternative of fewer than two stochastic trends: 
\begin{equation*} H_0:\ d_N(\mathbf{x}_t)=2,\qquad H_1:\ d_N(\mathbf{x}_t)<2. 
\end{equation*}
Under the null, no nontrivial linear combination of F1 and F2 is stationary. The variance-ratio statistic is $63.2050$, with a $p$-value of approximately $0.95$, far above conventional significance levels. This supports treating F1 and F2 as distinct persistent coordinates. The maintained full-nonstationary-rank condition also ensures that $\widehat C_{\mathbf{x}\mathbf{x}}$ and $\widehat\Omega_{\mathbf{x}\mathbf{x}}$ are invertible with probability approaching one, as required to define $\Bpre$ and $\widehat B$. It does not, however, imply that F1 and F2 have different distributional response functions; the corresponding common-response restriction is examined below.

Having retained the specification in which the CO$_2$ and non-CO$_2$ forcing portfolios carry linearly independent stochastic trends, we next ask whether these trends account for the persistent evolution of the density process. In our framework, this is the between-cointegration condition: after projecting the density process on the forcing vector, the remaining disturbance $U_t$ in \eqref{eqmodel1} must be stationary. We thus construct the residuals $\Upre_t=(Y_t-\bar Y_T)-\Bpre(\mathbf{x}_t-\bar{\mathbf{x}}_T)$ from the centered least-squares projection, where $\Bpre=\widehat C_{Y\mathbf{x}}\widehat C_{\mathbf{x}\mathbf{x}}^{-1}$, and first evaluate the unnormalized diagnostic $\widehat{\mathcal T}_K$ defined in \eqref{eqteststat_K_main}. We also use the covariance-normalized diagnostic $\widehat{\mathcal T}_V$ in \eqref{eqteststat_V_main} as a complementary check. Critical values are obtained from the plug-in Monte Carlo procedure described in Section~\ref{sec_betcointeg}. The relevant long-run and one-sided covariance operators are estimated from $(\Delta\mathbf{x}_t,\Upre_t)$ using the Parzen kernel with the bandwidth chosen as the nearest integer to $T^{1/4}$, and the two null distributions are approximated using 5{,}000 replications on a grid of 499 subintervals. The unnormalized statistic is $\widehat{\mathcal T}_K=0.0110$, with a simulated 5\% critical value of $0.0165$ and a Monte Carlo $p$-value of $0.1398$. The complementary covariance-normalized check gives $\widehat{\mathcal T}_V=0.2651$, with a simulated 5\% critical value of $0.3114$ and a Monte Carlo $p$-value of $0.0748$. Both statistics fall below their respective critical values, so neither test rejects between-cointegration at the 5\% level. The response analysis thus proceeds under the specification in \eqref{eqmodel1}.

\noindent Within the stochastic-trend framework adopted here, this finding also has a substantive climate interpretation. Our empirical formulation represents climate change as persistent evolution in the cross-sectional distribution of observed temperature anomalies. At the 5\% level, the tests do not detect an additional stochastic trend in the residual after accounting for the two anthropogenic forcing portfolios. Thus, conditional on the constructed forcing measures and the maintained long-run specification, the results are consistent with their trends spanning the persistent distributional component of observed climate change. This interpretation concerns only the persistent component; stationary natural variability and transitory shocks may remain in the residual.

\subsection{Forcing-specific responses and what aggregation misses}
\label{subsec:responses_and_margins}

\noindent Having retained F1 and F2 as distinct persistent coordinates and found no evidence against between-cointegration, we estimate the long-run response operator using the fully modified estimator in Section~\ref{sec_esti}. The forcing-specific estimates are $\widehat\beta_k=\widehat B(e_k)$ for $k\in\{1,2\}$. For illustration, Figure~\ref{Fig:RF_Response} reports $\Delta x_k\widehat\beta_k$, where $\Delta x_k$ is set equal to the sample standard deviation of forcing portfolio $k$ and serves only as a convenient reporting scale, together with marginal confidence intervals for the corresponding local-average CLR responses.

Because $\widehat\beta_k$ belongs to the CLR space, it describes a response on the centered log-density scale rather than directly on the density scale. To illustrate the corresponding change in density, we apply the fitted CLR response for the chosen increment $\Delta x_k$ to a common reference density and map the result back through the inverse CLR transformation. Following \citet{Nam2025}, we take the arithmetic average $f_0(s)=T^{-1}\sum_{t=1}^Tf_t(s)$ on $[a,b]$, where $a=-5.2527$ and $b=5.6224$, as the reference. Because the density process is nonstationary, $f_0$ is not interpreted as a stationary population mean or long-run equilibrium, but only as a representative data-based reference. Its choice affects the density-scale illustration and descriptive margins, but not the CLR response estimates or their local-average inference. The implied probability density following a change $\Delta x_k$ is
\begin{equation}\label{eq:clr_pushforward}
f_k(\,\cdot\,;\Delta x_k)=\operatorname{clr}^{-1}\!\left[\operatorname{clr}(f_0)+\Delta x_k\widehat\beta_k\right].
\end{equation}
In the finite-basis implementation used here, the inverse-CLR normalization ensures that $f_k(s;\Delta x_k)>0$ and $\int_a^b f_k(s;\Delta x_k)\,ds=1$, so $f_k(\cdot;\Delta x_k)$ is a proper probability density on $[a,b]$.

\begin{figure}[t]
\centering
\includegraphics[height=0.27\textwidth,trim={0.1cm 0.1cm 0.1cm 0.1cm},clip]{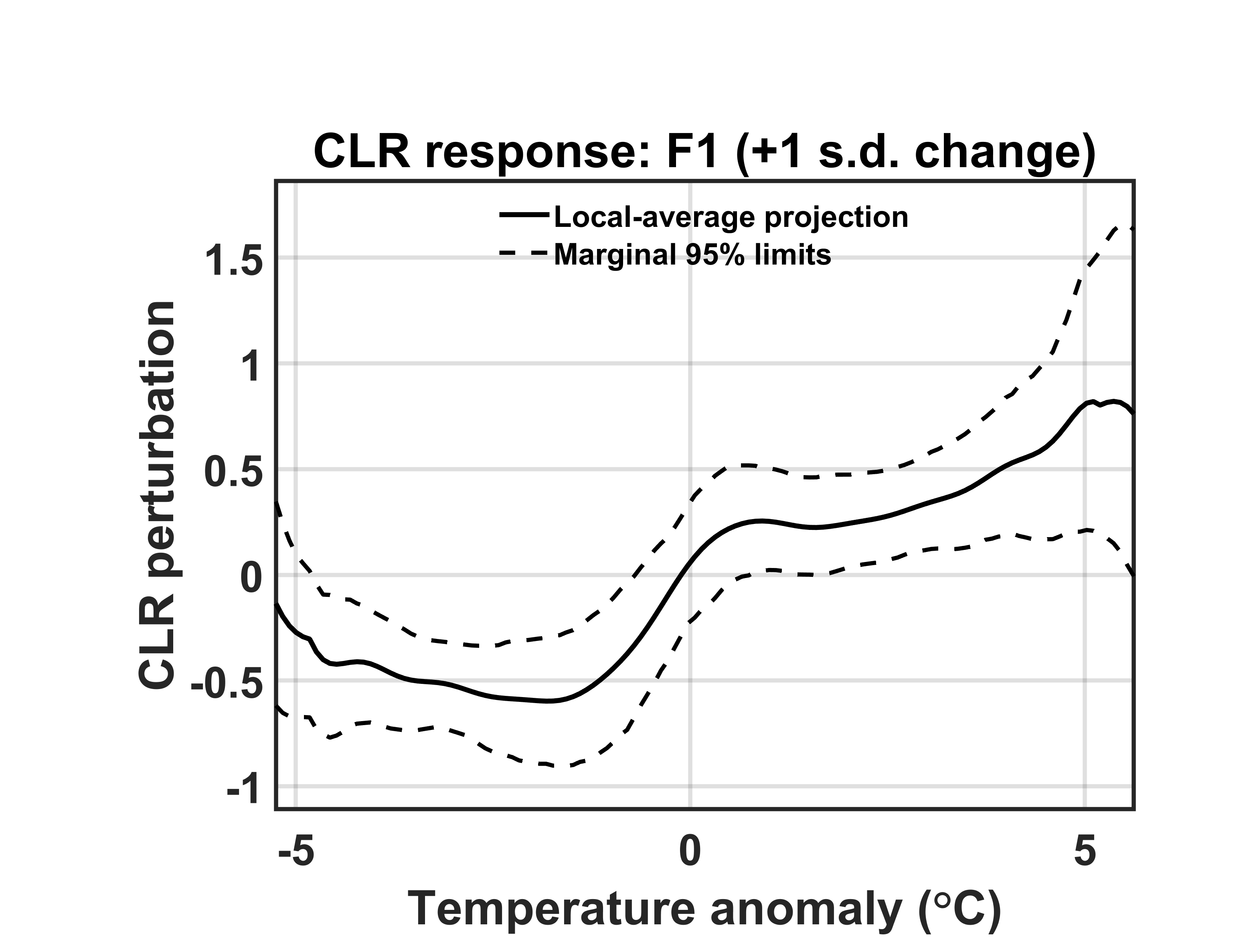}
\includegraphics[height=0.24\textwidth,trim={0.0cm 0.0cm 0.0cm 0.0cm},clip]{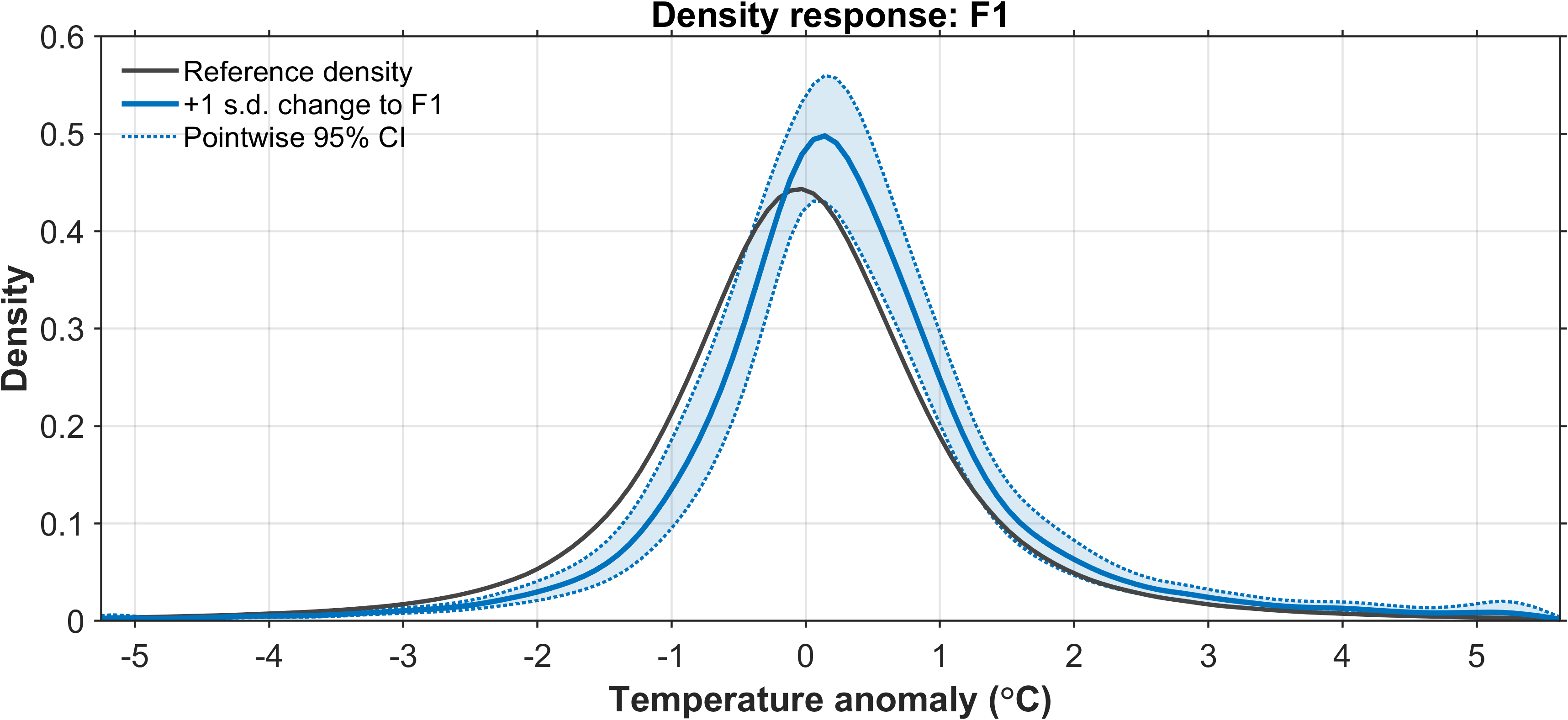}
\includegraphics[height=0.27\textwidth,trim={0.1cm 0.1cm 0.1cm 0.1cm},clip]{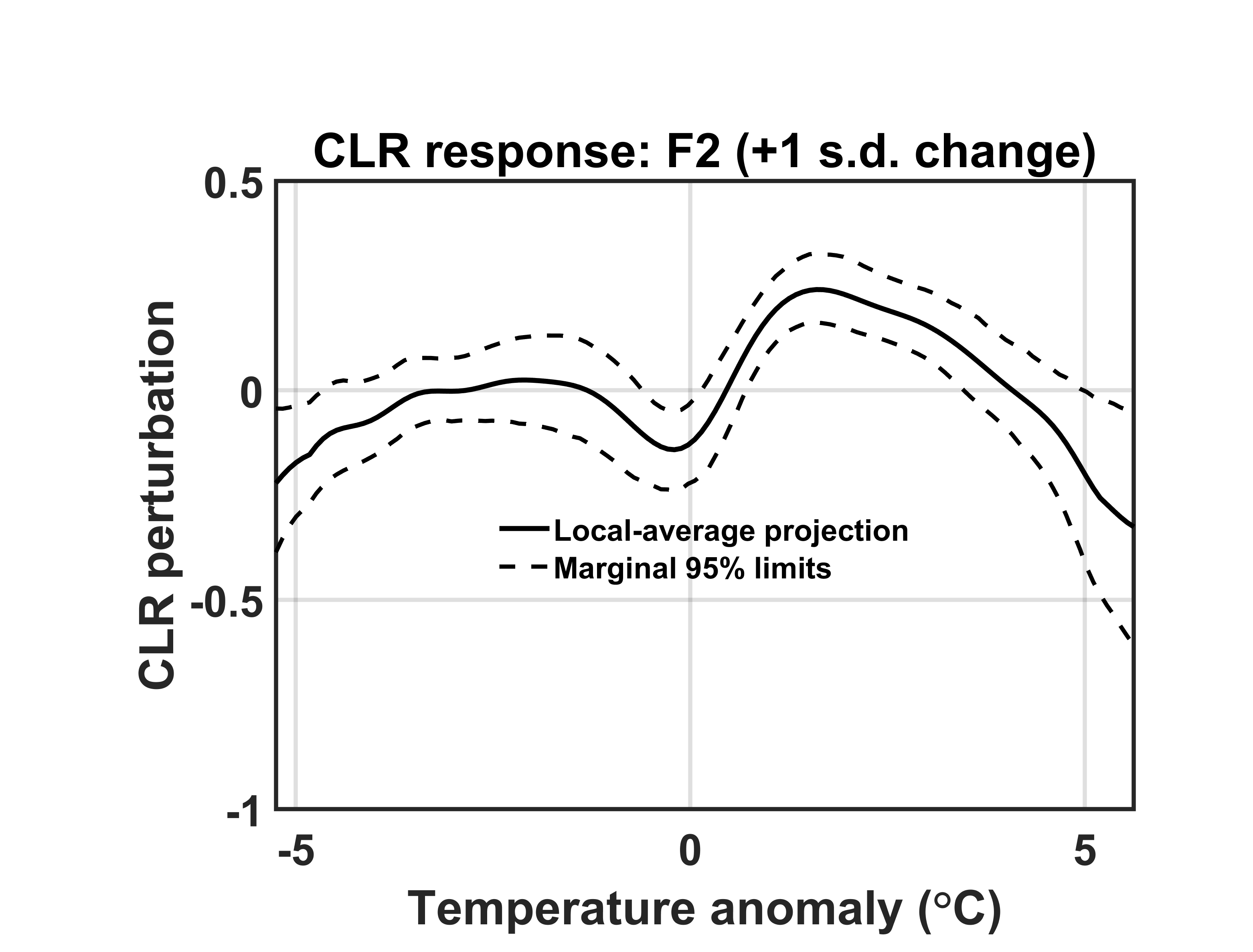}
\includegraphics[height=0.24\textwidth,trim={0.0cm 0.0cm 0.0cm 0.0cm},clip]{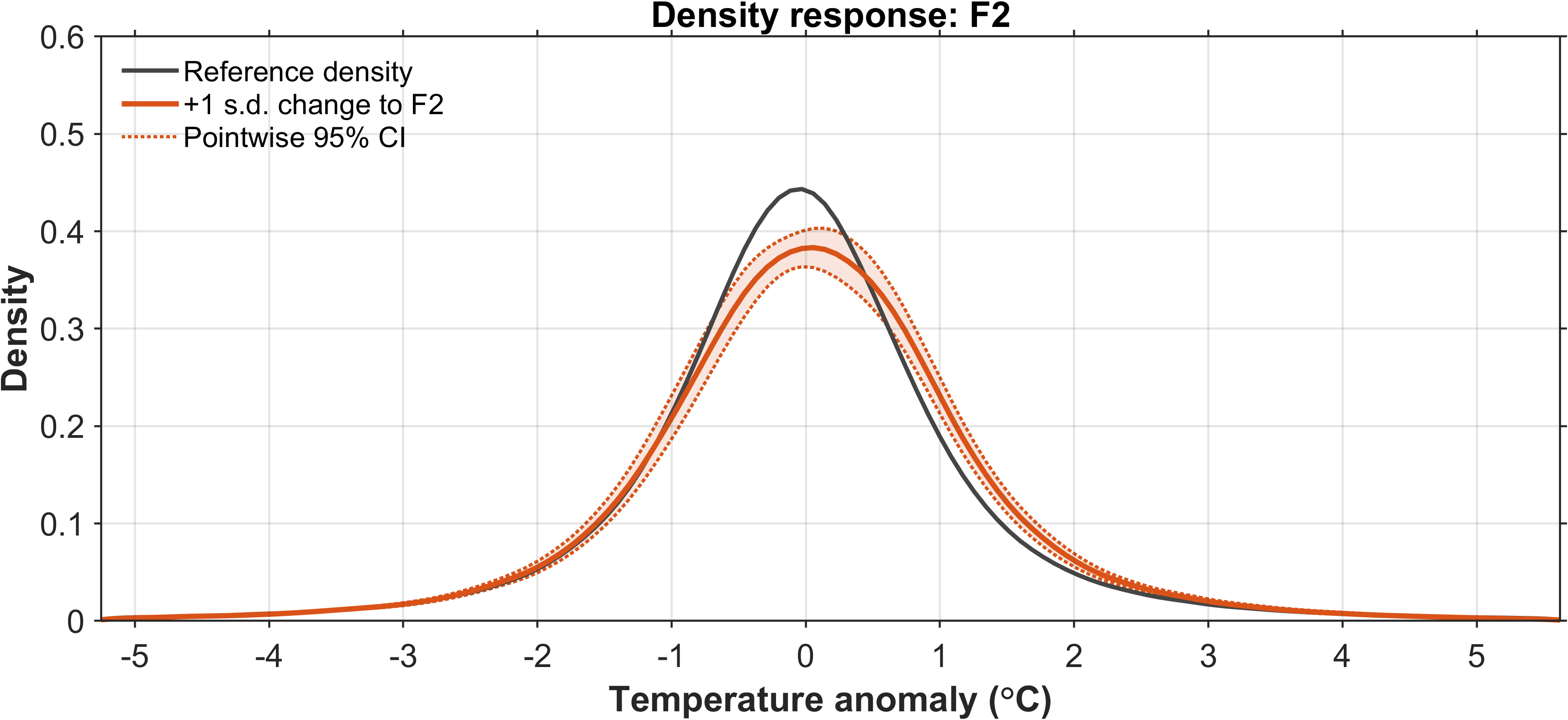}
\caption{Estimated responses to the specified increases in F1 (top) and F2 (bottom): local-average CLR responses with marginal 95\% confidence intervals (left), and the common reference density and fitted implied density (right). Shaded bands show central 95\% pointwise simulation envelopes for the implied densities, based on 50{,}000 joint draws; see Section~\ref{sec_app_density_ci} of the Supplement.}
\label{Fig:RF_Response}
\end{figure}

Figure~\ref{Fig:RF_Response} reports the local-average CLR responses scaled by the increments $\Delta x_k$ in the left panels and the implied densities alongside the reference density in the right panels.\footnote{We use an 11-grid-point window, spanning approximately $0.86^\circ\mathrm{C}$ at interior points, to smooth grid-level variation while retaining local response features.} The CLR panels use different vertical scales. At this scale, the implied F1 density shifts probability mass from colder to warmer anomalies while increasing central concentration and upper-tail mass relative to the reference. The implied F2 density shows a smaller warmward shift, greater dispersion, and a modest increase in upper-tail mass. For F1, the marginal intervals for the local-average CLR response lie entirely below zero over a broad cold-anomaly range and entirely above zero over a warm-anomaly range. For F2, they lie entirely below zero over a central range and entirely above zero over a broad range on the warm side.


\indent We next report several descriptive summaries of the implied density changes shown in the right panels of Figure~\ref{Fig:RF_Response}. These quantities summarize location, dispersion, and tail features and are not treated as separate inferential targets. The distinction between location and dispersion is empirically relevant and is consistent with \citet{chang2020evaluating}, who document stochastic trends in both features.

Let $Q_k(u)$ denote the quantile function of $f_k(\cdot;\Delta x_k)$, and let $Q_0(u)$ denote that of the reference density.\footnote{In practice, $Q_k(u)$ is obtained from the CDF evaluated on the discretized anomaly grid. We numerically enforce monotonicity, retain unique CDF values, and use linear interpolation for inversion.} The implied change in the mean is $\Delta\mu_k=\int_a^b s\{f_k(s;\Delta x_k)-f_0(s)\}\,ds$, and the central-half displacement is $\delta_k^{\mathrm{loc}}=(u_H-u_L)^{-1}\int_{u_L}^{u_H}\{Q_k(u)-Q_0(u)\}\,du$, where $(u_L,u_H)=(0.25,0.75)$. This quantity provides a location benchmark for typical anomaly states that is less sensitive to tail behavior than the mean. We also report the median shift $\Delta\mathrm{Median}_k=Q_k(0.5)-Q_0(0.5)$, the interquartile-range change $\Delta\mathrm{IQR}_k=\{Q_k(0.75)-Q_k(0.25)\}-\{Q_0(0.75)-Q_0(0.25)\}$, and the mean--location gap $\Delta\mu_k-\delta_k^{\mathrm{loc}}$, which compares the overall mean change with the central-half displacement.

For the tail summaries, fix the reference thresholds $r_q=Q_0(q)$ for $q\in\{0.05,0.95\}$. To distinguish reallocation between the two tails from a change in combined tail probability, define the cold-tail change as $\Delta\mathrm{ColdMass}_k=\int_a^{r_{0.05}}\{f_k(s;\Delta x_k)-f_0(s)\}\,ds$, the warm-tail change as $\Delta\mathrm{WarmMass}_k=\int_{r_{0.95}}^b\{f_k(s;\Delta x_k)-f_0(s)\}\,ds$, and their sum as $\Delta\mathrm{Extreme}_k=\Delta\mathrm{ColdMass}_k+\Delta\mathrm{WarmMass}_k$.

\begin{table}[t]
\centering
\caption{Descriptive summaries of the fitted density changes under one-standard-deviation increases. Temperature margins are in $^\circ$C; tail margins are probability changes.}
\label{Tab:MarginDecomp}
\renewcommand{\arraystretch}{1.15}
\setlength{\tabcolsep}{3pt}
\resizebox{\textwidth}{!}{
\begin{tabular}{lrrrrrrrr}
\toprule
Forcing
& $\Delta\mu$
& $\delta^{\mathrm{loc}}$
& $\Delta\mu-\delta^{\mathrm{loc}}$
& $\Delta\mathrm{Median}$
& $\Delta\mathrm{IQR}$
& $\Delta\mathrm{ColdMass}$
& $\Delta\mathrm{WarmMass}$
& $\Delta\mathrm{Extreme}$ \\
\midrule
F1 & 0.3125 & 0.2646 & 0.0479 & 0.2562 & $-$0.1120
& $-$0.0205 & 0.0223 & 0.0018 \\
F2 & 0.0789 & 0.0883 & $-$0.0094 & 0.0894 & 0.1500
& 0.0003 & 0.0080 & 0.0083 \\
\bottomrule
\end{tabular}}
\end{table}

\indent Table~\ref{Tab:MarginDecomp} provides numerical summaries of the changes in the implied densities shown in Figure~\ref{Fig:RF_Response}. We use their signs and relative magnitudes descriptively, not as separately tested quantities. At the chosen reporting scale, the larger central-half displacement and negative IQR point estimate for F1 correspond to its stronger warmward shift and greater central concentration, whereas the smaller displacement and positive IQR point estimate for F2 correspond to its weaker location shift and wider central range.

The table shows two further distinctions. First, the implied mean--location gap is positive for F1 but close to zero for F2. Thus, F1's mean change is not fully summarized by its central displacement, whereas the implied mean and central shifts for F2 are closely aligned. Second, the sum of the two reference-tail changes distinguishes changes in extreme-state composition from changes in combined tail probability. For F1, the cold-tail decrease nearly offsets the warm-tail increase, leaving $\Delta\mathrm{Extreme}_1=0.0018$; the implied tail composition therefore shifts from the cold to the warm tail with little change in combined tail probability. For F2, the cold tail is nearly unchanged, so the warm-tail increase carries through to $\Delta\mathrm{Extreme}_2=0.0083$ and is concentrated on the warm side.

\indent The estimated CLR profiles in Figure~\ref{Fig:RF_Response} differ visibly in shape and sign patterns, suggesting that aggregation may conceal substantial response heterogeneity. As a formal check using the projection inference developed above, we examine the common-response restriction $\beta_1=\beta_2$ introduced in Section~\ref{subsec:aggregation_cost} within the fixed response space used for estimation. This comparison differs from the no-within-cointegration diagnostic in Section~\ref{sec_test_stat}: that diagnostic assesses the number of stochastic trends in the forcing block, whereas the present comparison asks whether F1 and F2 have the same distributional response.

Because F1 and F2 are measured in the same $\mathrm{W\,m^{-2}}$ units, equality is assessed using the unscaled coefficient functions. Let $\mathbf c=e_1-e_2$ denote the contrast vector, so that $B\mathbf c=\beta_1-\beta_2$. Section~\ref{sec_app_common_response} of the Supplement uses the Cram\'{e}r--Wold device and Theorems~\ref{thmapp2}--\ref{thmapp3} to establish joint validity of the feasible approximation for any fixed finite collection of response directions. Let $\{\phi_m\}_{m=1}^J$, with $J=20$, denote the orthonormal nonconstant Fourier directions spanning the empirical response space, and let $P_J$ denote the associated projection. We test
\begin{equation*}
H_{0,J}:\ \langle B\mathbf c,\phi_m\rangle=0\ \text{for }m=1,\ldots,J,\qquad
\widehat{\mathcal S}_{\mathrm{eq},J}
=T\left\{\sum_{m=1}^J\langle\widehat B\mathbf c,\phi_m\rangle^2\right\}^{1/2}
=T\|P_J\widehat B\mathbf c\|_{L^2}.
\end{equation*}
Here $\|g\|_{L^2}=(\int_a^b |g(s)|^2\,ds)^{1/2}$ denotes the $L^2[a,b]$ norm. The null distribution is approximated using 50{,}000 joint draws from the feasible approximation to the intercept-model limit, with the same Brownian paths used across all $J$ coordinates within each replication and with $M=7$. The same section reports the numerical implementation, retained-variation shares, and sensitivity to $J$ and $M$. The observed statistic is 557.48, exceeding the simulated 95th percentile of 463.85; the corresponding add-one Monte Carlo $p$-value is 0.0199. Thus, the common-response restriction is rejected at the 5\% level within the fixed response space.

\indent We next use two scalar-response benchmarks to distinguish the information lost by reducing the density response to its mean from that lost by aggregating the forcing coordinates. The aggregate scalar specification regresses the density mean on total anthropogenic forcing, thereby collapsing both the response and predictor. The vector-to-mean specification retains F1 and F2 separately but models only the mean response. All three specifications use the same fully modified estimation approach.

For Table~\ref{Tab:BenchmarkComparison}, we construct a matched reporting scenario by setting the increase in aggregate anthropogenic forcing (F1+F2) equal to its sample standard deviation and allocating it between F1 and F2 in proportion to their sample standard deviations. This scaling serves only as a convenient empirical normalization.\footnote{Let $s_1$ and $s_2$ denote the sample standard deviations of F1 and F2, and let $s_{\mathrm{agg}}$ denote that of F1+F2. We set $\Delta_j=s_{\mathrm{agg}}s_j/(s_1+s_2)$ for $j=1,2$, so that $\Delta_1+\Delta_2=s_{\mathrm{agg}}$. The F1, F2, and joint rows apply $(\Delta_1,0)$, $(0,\Delta_2)$, and $(\Delta_1,\Delta_2)$, respectively.} The component rows isolate the two allocated changes, whereas the joint row combines them to represent the matched aggregate increase. This normalization differs from the coordinate-specific scenarios in Figure~\ref{Fig:RF_Response} and Table~\ref{Tab:MarginDecomp}, which change F1 and F2 separately by their respective sample standard deviations. For the matched joint change under the vector-to-density specification, the implied probability density is $f(\,\cdot\,;\Delta_1,\Delta_2)=\operatorname{clr}^{-1}[\operatorname{clr}(f_0)+\Delta_1\widehat\beta_1+\Delta_2\widehat\beta_2]$. Because the inverse CLR mapping is nonlinear, the density-derived component margins need not sum to the joint margin.

\begin{table}[t]
\centering
\caption{Matched responses to a one-standard-deviation increase in aggregate
anthropogenic forcing. Temperature margins are in $^\circ$C; WarmMass is a probability change.}
\label{Tab:BenchmarkComparison}
\renewcommand{\arraystretch}{1.15}
\setlength{\tabcolsep}{7pt}
\begin{tabular}{llrrr}
\toprule
Specification & Channel & $\Delta\mu$ & $\Delta\mathrm{IQR}$ & $\Delta\mathrm{WarmMass}$ \\
\midrule
Aggregate scalar  & F1+F2  & 0.3430 & \multicolumn{1}{c}{--} & \multicolumn{1}{c}{--} \\
\midrule
Vector-to-mean    & Matched F1 component & 0.2785 & \multicolumn{1}{c}{--} & \multicolumn{1}{c}{--} \\
                  & Matched F2 component & 0.0604 & \multicolumn{1}{c}{--} & \multicolumn{1}{c}{--} \\
                  & Matched joint change & 0.3389 & \multicolumn{1}{c}{--} & \multicolumn{1}{c}{--} \\
\midrule
Vector-to-density & Matched F1 component & 0.2874 & $-$0.1030 & 0.0202 \\
                  & Matched F2 component & 0.0716 & 0.1371 & 0.0073 \\
                  & Matched joint change & 0.3678 & $-$0.0098 & 0.0283 \\
\bottomrule
\end{tabular}
\end{table}

\indent
Table~\ref{Tab:BenchmarkComparison} provides a descriptive comparison of fitted values; no separate inference is conducted for its nonlinear margins. The $\Delta\mathrm{IQR}$ and $\Delta\mathrm{WarmMass}$ entries summarize features of the implied densities. The matched mean responses are $0.3430^\circ$C under the aggregate scalar specification and $0.3389^\circ$C under the vector-to-mean specification, compared with $0.3678^\circ$C under the vector-to-density specification. Thus, broadly similar fitted mean responses coexist with IQR and tail changes observable only under the density specification. The matched F1 and F2 components have opposite-signed IQR changes. When the two changes are applied jointly, the resulting density exhibits little change in central spread, with $\Delta\mathrm{IQR}=-0.0098^\circ\mathrm{C}$, while upper-tail probability increases by 2.83 percentage points. The joint implied density exhibits a meaningful change in tail behavior despite its small IQR change.

\subsection{Distributional displacement and reshaping}
\label{subsec:loss_decomposition_interpretation}

\noindent The global margins in Table~\ref{Tab:MarginDecomp} summarize changes in location, dispersion, and tail mass, but not where offsetting gains and losses of probability mass occur. Similar mean or IQR changes can reflect quite different reallocations across the anomaly support. Pointwise density differences are also less directly interpretable because they measure density height rather than probability and may be sensitive to grid-level variation. We therefore examine probability changes over fixed-width neighborhoods, providing an intermediate view between global summaries and pointwise comparisons.

We return to the coordinate-specific reporting scenarios and write $f_k^{\mathrm{full}}(s)\equiv f_k(s;\Delta x_k)$ for the full implied density in \eqref{eq:clr_pushforward}. To distinguish displacement from reshaping, we construct a location-only benchmark by shifting the reference density by the central-half displacement:
\begin{equation}\label{eq:loc_counterfactual_sec6}
f_k^{\mathrm{loc}}(s)\propto f_0\!\left(s-\delta_k^{\mathrm{loc}}\right),\qquad s\in[a,b].
\end{equation}
After extending $f_0$ by zero outside $[a,b]$, the shifted density is restricted to $[a,b]$ and renormalized. Apart from this boundary adjustment, the benchmark preserves the reference shape. The difference between $f_k^{\mathrm{full}}$ and $f_k^{\mathrm{loc}}$ therefore captures implied changes in dispersion, asymmetry, and mass allocation beyond a pure displacement.

For window width $\ell>0$ and center $r\in[a+\ell/2,b-\ell/2]$, let $I_\ell(r)=(r-\ell/2,r+\ell/2]$ and define
\begin{equation}\label{eq:local_loss_functional_sec6}
L_\ell(r;f)=\int_{I_\ell(r)}f(s)\,ds.
\end{equation}
Unlike density height, $L_\ell(r;f)$ is the probability assigned to a neighborhood of $r$. Define $\Delta L_{k,\ell}^{\mathrm{full}}(r)=L_\ell(r;f_k^{\mathrm{full}})-L_\ell(r;f_0)$, $\Delta L_{k,\ell}^{\mathrm{loc}}(r)=L_\ell(r;f_k^{\mathrm{loc}})-L_\ell(r;f_0)$, and $\Delta L_{k,\ell}^{\mathrm{dist}}(r)=\Delta L_{k,\ell}^{\mathrm{full}}(r)-\Delta L_{k,\ell}^{\mathrm{loc}}(r)$. Thus, the full local response is decomposed exactly into displacement and residual reshaping at each $r$. This decomposition is descriptive and benchmark-dependent rather than a unique structural separation. We set $\ell=0.50^\circ\mathrm{C}$ and evaluate the profiles on an overlapping grid with a $0.05^\circ\mathrm{C}$ step. Positive values indicate gains in local probability mass and negative values indicate losses. Because the windows overlap, the profiles are localized diagnostics rather than an additive probability partition.

\begin{figure}[t]
\centering
\includegraphics[height=0.35\textwidth, width=0.49\textwidth]{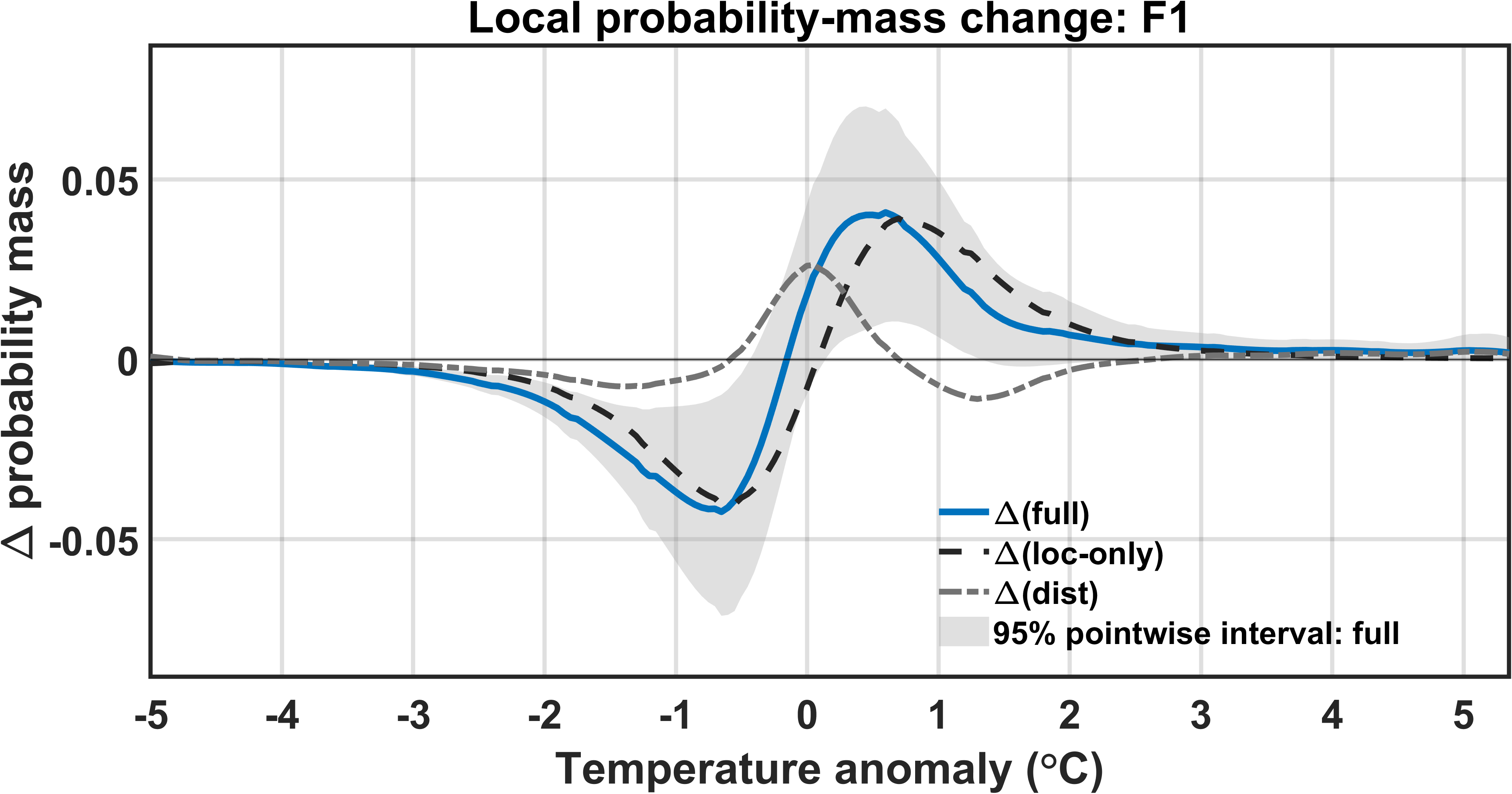}
\includegraphics[height=0.35\textwidth, width=0.49\textwidth]{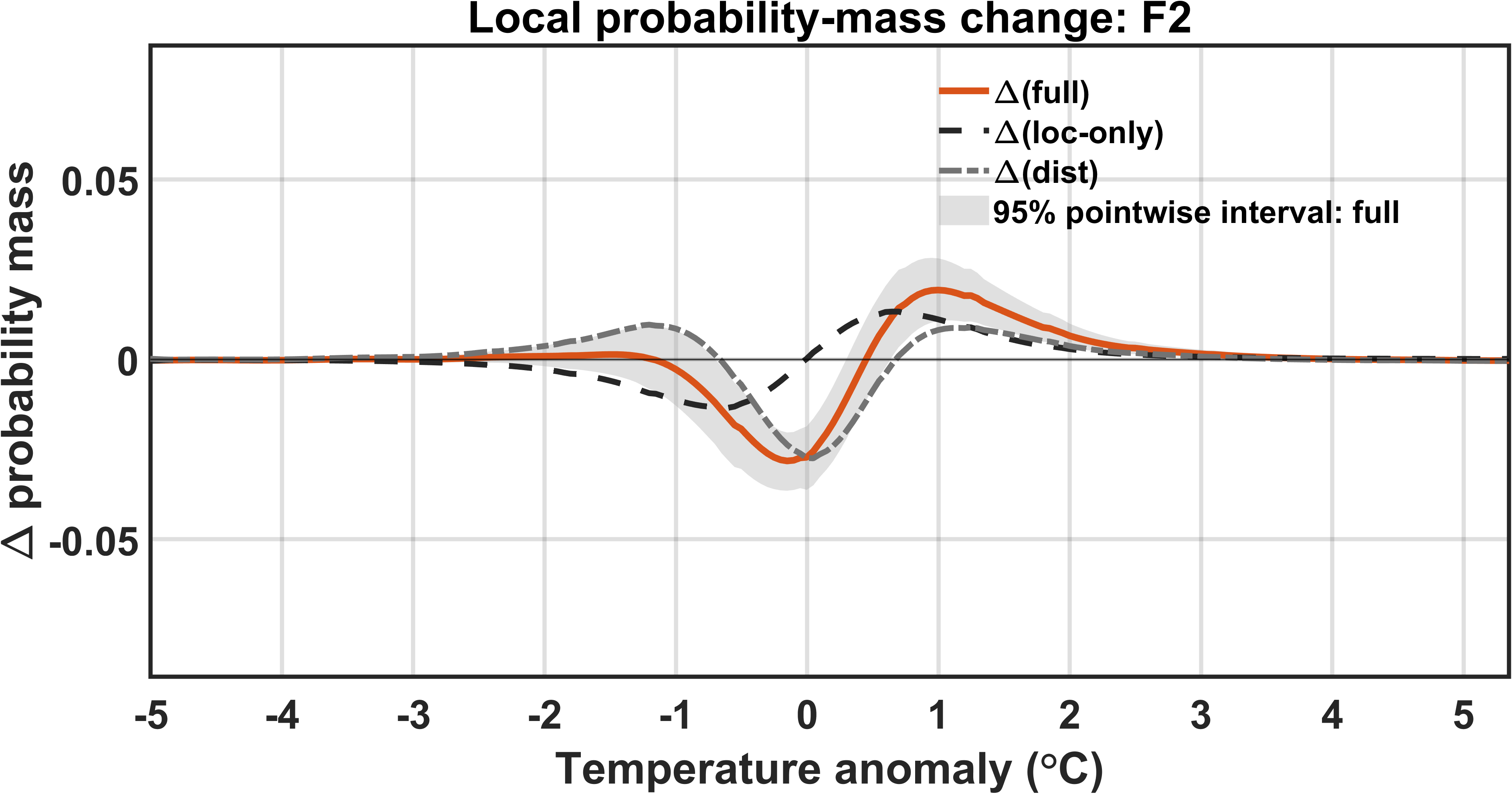}
\caption{Local probability-mass responses for F1 (left) and F2 (right): full fitted responses, location-only components, and residual reshaping components. Shaded bands are central 95\% pointwise simulation envelopes for the full responses, based on the same 50{,}000 joint draws as the implied-density panels of Figure~\ref{Fig:RF_Response}; see Section~\ref{sec_app_density_ci} of the Supplement.}
\label{Fig:BinLossDecomp}
\end{figure}

\indent
Figure~\ref{Fig:BinLossDecomp} reveals different mixtures of location-only displacement and residual reshaping. For F1, the location-only component determines the broad sign pattern: fitted probability mass falls over cold and near-central states and rises on the warm side. The reshaping component adds mass around $r\approx0$ and removes it near $r\approx1.5$, consistent with the negative IQR change reported in Table~\ref{Tab:MarginDecomp}. For F2, the location-only component is relatively small, whereas the reshaping component removes mass near the center and adds it on both flanks, around $r\approx-1$ and $r\approx1.5$, providing the localized counterpart of its positive IQR change. In distributional terms, the fitted CO$_2$ response is therefore dominated by a broad warmward displacement of the cross-sectional distribution of observed local temperature anomalies. By contrast, the fitted non-CO$_2$ response exhibits more pronounced reshaping, with greater dispersion and a higher relative prevalence of both cool and warm off-center states, consistent with the heterogeneous signs and spatial incidence of the components in this portfolio.

\section{The shape of warming: concluding remarks}\label{conclude}

\noindent By linking separate anthropogenic forcing portfolios to the anomaly density rather than its mean alone, our analysis gives the \emph{shape of warming} a specific empirical meaning. The specification evidence supports retaining the CO$_2$ and non-CO$_2$ portfolios as distinct persistent coordinates and is consistent with their trends jointly spanning the persistent evolution of the observed temperature-anomaly distribution. Within the fixed response space, rejection of the common-response restriction further indicates that the distributional response depends on the composition, not merely the aggregate level, of anthropogenic forcing.

\indent On the density scale, warming is not simply a uniform shift to the right. For the CO$_2$ portfolio, the fitted density moves warmward and becomes more concentrated around its center: colder states lose probability, the central peak sharpens, and warmer states gain probability, while the two tail changes nearly offset. For the non-CO$_2$ portfolio, the density moves less but spreads outward: central states lose mass, both flanks gain it, and additional extreme-state probability appears mainly in the warm tail. The two portfolios therefore trace distinct shapes of warming: displacement with concentration versus weaker displacement with widening and reshaping. Scalar benchmarks yield similar aggregate mean responses but compress these patterns into a single number. These fitted patterns represent persistent reduced-form associations and do not establish causality.

\bibliography{VtD_biblio}
\appendix

\counterwithin{assumption}{section}
\counterwithin{theorem}{section}
\renewcommand{\theassumption}{\thesection.\arabic{assumption}}
\renewcommand{\thetheorem}{\thesection.\arabic{theorem}}

\newpage 

\LARGE{\textbf{Supplementary Appendix}} \normalsize

\section{Additional details on the density-valued response}
\label{sec_app_density_error}
\noindent The empirical analysis in Section~\ref{subsec:response_data} first estimates an annual density of observed grid-cell-month temperature anomalies and then applies the CLR transformation in \eqref{eq_clr_response} to obtain the functional response. This section provides additional data and numerical details and clarifies how density-estimation error is accommodated by the regression model.

\indent We use the non-infilled HadCRUT.5.1.0.0 product downloaded from \url{https://www.metoffice.gov.uk/hadobs/hadcrut5/} on January~15,~2026. HadCRUT5 combines land-surface air temperatures from CRUTEM5 and sea-surface temperatures from HadSST4 on a $5^\circ\times5^\circ$ latitude--longitude grid \citep{morice2021updated}. The ensemble-mean observations are pooled within each year with equal weights, as described in the main article.

\indent Annual densities are estimated by Gaussian kernel smoothing using Silverman's rule-of-thumb bandwidth. The common support is $[-5.2527,5.6224]$, obtained by trimming 0.5\% from each tail of the pooled anomaly distribution. Each kernel estimate is restricted to this interval, renormalized to integrate to one, and evaluated on a common grid of 128 points before applying the CLR transformation.

\indent To distinguish the underlying and estimated objects, let $f_t^\circ$ denote the underlying annual density and $\widehat f_t$ its kernel estimator. The corresponding latent and empirical CLR responses are $Y_t^\circ=\operatorname{clr}(f_t^\circ)$ and $Y_t=\operatorname{clr}(\widehat f_t)$, respectively, so that
\begin{equation*}
Y_t=Y_t^\circ+e_t,\qquad
e_t:=\operatorname{clr}(\widehat f_t)-\operatorname{clr}(f_t^\circ).
\end{equation*}
If the latent response satisfies $Y_t^\circ(s)=\beta_0(s)+\sum_{k=1}^{\KX}\beta_k(s)x_{k,t}+U_t^\circ(s)$ for $s$ in the common support, then the empirical response satisfies
\begin{equation*}
Y_t(s)=\beta_0(s)+\sum_{k=1}^{\KX}\beta_k(s)x_{k,t}+U_t(s),
\qquad
U_t(s):=U_t^\circ(s)+e_t(s).
\end{equation*}
Thus, density-estimation error on the CLR scale is absorbed into the composite regression disturbance $U_t$. The asymptotic theory does not require $e_t$ to vanish separately as $T\to\infty$; rather, the assumptions and results are imposed directly on $U_t$. Any contribution of $e_t$ to temporal dependence or uncertainty is therefore incorporated into the long-run covariance of $U_t$, provided that the composite process satisfies the maintained stationarity, moment, and weak-dependence conditions. Section~\ref{sec_app_coverage} examines sensitivity to changing spatial coverage, a potential source of variation in $e_t$.

\section{Implementation of the between-cointegration tests}
\label{sec_app_test_implementation}

\noindent This section presents the plug-in Monte Carlo procedure used to compute feasible critical values and $p$-values for the between-cointegration tests introduced in Section~\ref{sec_betcointeg}. The required inputs are the projection residuals $\Upre_t$, the observed statistics $\widehat{\mathcal T}_K$ and $\widehat{\mathcal T}_V$, the residual covariance operator $\widehat{\mathcal V}$, and its dominant eigenvector $\hat v_V$, all computed as described there. We state the procedure for the specification with an intercept used in the empirical application. Section~\ref{sec_app_test} provides its theoretical justification, with the intercept modification established in Section~\ref{sec_app_test_1}. Section~\ref{sec_sim_between} reports simulation results on finite-sample size and power.

\begin{description}[style=nextline, leftmargin=1em, font=\bfseries,topsep=-2pt, itemsep=0pt]

\item[Step 1: Estimating the covariance inputs] \hfill
Using $\Delta\mathbf{x}_t$ and the projection residuals $\Upre_t$, compute the kernel estimators $\widehat{\Omega}_{\mathbf{x}\mathbf{x}}$, $\widehat{\Omega}_{U\mathbf{x}}$, and $\widehat{\Omega}_{U\mathbf{x}}^+$ defined in \eqref{eqsample1}--\eqref{eqsample2}, together with
\begin{equation}\label{eqlrveq01} \widehat{\Omega}_{UU}=\sum_{|j|\leq h}\mathrm{k}(j/h)\widehat{\Gamma}_{UU}^{(j)}. \end{equation}
Here $\widehat{\Gamma}_{UU}^{(j)}$ is the lag-$j$ sample covariance operator computed from $\Upre_t$, that is, $\widehat{\Gamma}_{UU}^{(j)}=T^{-1}\sum_{t=(1-j)\vee1}^{T\wedge(T-j)}\Upre_t\otimes\Upre_{t+j}$. The kernel and bandwidth conditions are stated later in Assumption~\ref{assum_test_kernel}. Consider the joint long-run covariance estimator
\begin{equation}\label{eq_app_test_omegaZZ} \widehat{\Omega}_{ZZ}=\begin{pmatrix}\widehat{\Omega}_{\mathbf{x}\mathbf{x}}&\widehat{\Omega}_{U\mathbf{x}}^{*}\\ \widehat{\Omega}_{U\mathbf{x}}&\widehat{\Omega}_{UU}\end{pmatrix} \end{equation}
on $\HX\times\HY$, where $\widehat{\Omega}_{U\mathbf{x}}^{*}$ denotes the adjoint of $\widehat{\Omega}_{U\mathbf{x}}$. Under Assumption~\ref{assum_test_kernel}, the common nonnegative-definite lag window makes $\widehat{\Omega}_{ZZ}$ a self-adjoint, nonnegative, finite-rank operator, so it can be used as the covariance operator for the joint Brownian paths simulated in Step~2.

\item[Step 2: Simulating the limiting null statistics] \hfill
Conditionally on the data, generate a jointly distributed pair of Brownian paths $(\widehat W_{\mathbf{x},(\ei)},\widehat W_{U,(\ei)})$ with joint covariance operator $\widehat{\Omega}_{ZZ}$, independently for each $\ei=1,\ldots,R_{\mathrm{MC}}$. Remark~\ref{rem_prac} describes their numerical construction. Because the statistics are computed from demeaned observations, form the centered paths
\begin{equation}\label{eq_app_test_centered_paths} \widehat W^c_{\mathbf{x},(\ei)}(r)=\widehat W_{\mathbf{x},(\ei)}(r)-\int_0^1\widehat W_{\mathbf{x},(\ei)}(u)\,du,\qquad \widehat W^c_{U,(\ei)}(r)=\widehat W_{U,(\ei)}(r)-r\widehat W_{U,(\ei)}(1). \end{equation}
For each replication, define
\begin{equation}\label{eq_app_test_sim_AJ} \widehat{\mathcal A}^c_{(\ei)}=\int_0^1\widehat W^c_{\mathbf{x},(\ei)}(r)\otimes d\widehat W^c_{U,(\ei)}(r)+\widehat{\Omega}_{U\mathbf{x}}^+,\qquad \widehat J^c_{(\ei)}=\int_0^1\widehat W^c_{\mathbf{x},(\ei)}(r)\otimes\widehat W^c_{\mathbf{x},(\ei)}(r)\,dr, \end{equation}
and
\begin{equation}\label{eq_app_test_sim_QR} \widehat Q^c_{(\ei)}(r)=\int_0^r\widehat W^c_{\mathbf{x},(\ei)}(u)\,du,\qquad \widehat{\mathcal R}^c_{(\ei)}(r)=\widehat W^c_{U,(\ei)}(r)-\widehat{\mathcal A}^c_{(\ei)}\{\widehat J^c_{(\ei)}\}^{-1}\widehat Q^c_{(\ei)}(r). \end{equation}
The corresponding simulated residual partial-sum operator is
\begin{equation}\label{eq_app_test_sim_KR} \widehat{\mathcal K}^c_{R,(\ei)}=\int_0^1\widehat{\mathcal R}^c_{(\ei)}(r)\otimes\widehat{\mathcal R}^c_{(\ei)}(r)\,dr. \end{equation}
The simulated null statistics are
\begin{equation}\label{eq_app_test_null_draws} \widehat{\mathcal T}_{K,0,(\ei)}=\Lambda_{\max}\left(\widehat{\mathcal K}^c_{R,(\ei)}\right),\qquad \widehat{\mathcal T}_{V,0,(\ei)}=\frac{\langle\widehat{\mathcal K}^c_{R,(\ei)}\hat v_V,\hat v_V\rangle}{\langle\widehat{\mathcal V}\hat v_V,\hat v_V\rangle}. \end{equation}
The sample quantities $\widehat{\mathcal V}$ and $\hat v_V$ are held fixed across Monte Carlo replications. All integrals are evaluated numerically on a sufficiently fine grid over $[0,1]$, using left-endpoint sums for the stochastic integral. The limiting interpretation and validity of this centered construction are established in Section~\ref{sec_app_test_1}.

\item[Step 3: Computing critical values and $p$-values] \hfill
Repeat the simulation in Step~2 for $\ei=1,\ldots,R_{\mathrm{MC}}$. At significance level $\alpha$, let $\widehat q_{K,\alpha}$ and $\widehat q_{V,\alpha}$ be the empirical $(1-\alpha)$ quantiles of $\{\widehat{\mathcal T}_{K,0,(\ei)}\}_{\ei=1}^{R_{\mathrm{MC}}}$ and $\{\widehat{\mathcal T}_{V,0,(\ei)}\}_{\ei=1}^{R_{\mathrm{MC}}}$, respectively. The corresponding upper-tail Monte Carlo $p$-values are
\begin{equation}\label{eq_app_test_mc_pvalues} \widehat p_K=\frac{1+\sum_{\ei=1}^{R_{\mathrm{MC}}}\mathbf{1}\{\widehat{\mathcal T}_{K,0,(\ei)}\geq\widehat{\mathcal T}_K\}}{R_{\mathrm{MC}}+1},\qquad \widehat p_V=\frac{1+\sum_{\ei=1}^{R_{\mathrm{MC}}}\mathbf{1}\{\widehat{\mathcal T}_{V,0,(\ei)}\geq\widehat{\mathcal T}_V\}}{R_{\mathrm{MC}}+1}. \end{equation}
At level $\alpha$, the unnormalized and normalized tests reject between-cointegration when $\widehat{\mathcal T}_K>\widehat q_{K,\alpha}$ and $\widehat{\mathcal T}_V>\widehat q_{V,\alpha}$, respectively. In the empirical application, we use the Parzen kernel with $h=4$, $R_{\mathrm{MC}}=5{,}000$ replications, and a simulation grid with 499 subintervals.
\end{description}

\noindent The theoretical justification is provided later in Section~\ref{sec_app_test}, which derives the limiting residual process simulated in Step~2 and establishes the validity of the plug-in critical values and consistency of the tests. Section~\ref{sec_app_test_1} provides the centering argument for the intercept specification used here.

\section{Mathematical preliminaries}\label{Sec_prelim}
\subsection{Basic notation}\label{Sec_prelim1}
We collect notation and basic concepts for generic Hilbert spaces $\mathcal H_1$ and $\mathcal H_2$, each of which represents either the function space $\HY$ or the Euclidean space $\HX=\mathbb R^{\KX}$ in our application. We use $\langle\cdot,\cdot\rangle$ and $\|\cdot\|$ for the inner product and norm on all relevant spaces. In particular, $\langle v_1,v_2\rangle=\int_a^bv_1(s)v_2(s)\,ds$ on $\HY$ and $\langle v_1,v_2\rangle=v_1' v_2$ on $\HX$. The symbol $I$ denotes the identity operator on the relevant space.

Let $\mathcal L_{\mathcal H_1,\mathcal H_2}$ denote the space of bounded linear operators from $\mathcal H_1$ to $\mathcal H_2$, equipped with the operator norm $\|A\|_{\op}=\sup_{\|x\|\leq1}\|A(x)\|$. We write $\mathcal L_{\mathcal H}=\mathcal L_{\mathcal H,\mathcal H}$. For $\zeta_1\in\mathcal H_1$ and $\zeta_2\in\mathcal H_2$, the rank-one operator $\zeta_1\otimes\zeta_2\in\mathcal L_{\mathcal H_1,\mathcal H_2}$ is defined by $(\zeta_1\otimes\zeta_2)(x)=\langle\zeta_1,x\rangle\zeta_2$. For $A\in\mathcal L_{\mathcal H_1,\mathcal H_2}$, its range, kernel, and adjoint are denoted by $\ran A$, $\ker A$, and $A^*\in\mathcal L_{\mathcal H_2,\mathcal H_1}$, respectively. An operator $A\in\mathcal L_{\mathcal H}$ is self-adjoint if $A=A^*$ and nonnegative if $\langle A\zeta,\zeta\rangle\geq0$ for every $\zeta\in\mathcal H$. An operator $A\in\mathcal L_{\mathcal H_1,\mathcal H_2}$ is compact if it maps bounded subsets of $\mathcal H_1$ into relatively compact subsets of $\mathcal H_2$. Equivalently, a compact operator between separable Hilbert spaces admits a singular-value representation $A=\sum_{j=1}^{r}a_je_j\otimes f_j$, where $r\in\mathbb N\cup\{\infty\}$, $\{e_j\}\subset\mathcal H_1$ and $\{f_j\}\subset\mathcal H_2$ are orthonormal systems, and $a_j>0$ are the singular values. If $r=\infty$, then $a_j\to0$, and the series converges in operator norm.


Let $X\in\mathcal H_1$ and $Y\in\mathcal H_2$ be random elements with finite second moments. Their expectations are characterized by $\langle\mathbb E[X],\zeta\rangle=\mathbb E[\langle X,\zeta\rangle]$ for every $\zeta\in\mathcal H_1$, and analogously for $Y$. Their covariance and cross-covariance operators are
$C_X=\mathbb E[(X-\mathbb E[X])\otimes(X-\mathbb E[X])]\in\mathcal L_{\mathcal H_1}$ and $C_{YX}=\mathbb E[(X-\mathbb E[X])\otimes(Y-\mathbb E[Y])]\in\mathcal L_{\mathcal H_1,\mathcal H_2}$, respectively. Under this convention, $C_{XY}=C_{YX}^*$. The covariance operator $C_X$ is self-adjoint, nonnegative, and trace class, with $\operatorname{tr}(C_X)=\mathbb E\|X-\mathbb E[X]\|^2$.


We employ several additional notational conventions throughout this supplement. We use the same symbol $\to_p$ for convergence in probability in both finite-dimensional Euclidean spaces and infinite-dimensional Hilbert spaces, with convergence understood under the norm of the relevant space. As will be detailed later, $X_T\Rightarrow X$ denotes weak convergence, or convergence in distribution, in the relevant space. Following standard conventions in the functional data literature \citep[see, e.g.,][]{seo2020functional}, stochastic convergence and bounds for bounded linear operators are evaluated directly in terms of the operator norm. Specifically, for a sequence of operators $A_T$, $A_T\to_p A$ means $\|A_T-A\|_{\op}\to_p0$. Similarly, for a deterministic sequence $a_T>0$, we write $A_T-A=O_p(a_T)$ or $A_T-A=o_p(a_T)$ to mean $\|A_T-A\|_{\op}=O_p(a_T)$ or $\|A_T-A\|_{\op}=o_p(a_T)$, respectively.

\subsection{Cointegrated linear processes in Hilbert space}\label{AP_FTS}
To establish the asymptotic properties of the system in \eqref{eqmodel1}, we first characterize the stochastic trends of $Y_t\in\HY$ and $\mathbf{x}_t\in\HX$ within a unified framework. Let $Z_t$ be a generic time series taking values in a Hilbert space $\mathcal H$. Following the literature on functional nonstationarity \citep[e.g.,][]{Phillips1992,Chang2016152,seo_2022}, suppose that
\begin{equation*} \Delta Z_t=\sum_{j=0}^{\infty}\psi_j\varepsilon_{t-j}, \end{equation*}
where $\psi_j\in\mathcal L_{\mathcal H}$ and $\{\varepsilon_t\}$ is an i.i.d. innovation sequence satisfying $\mathbb E[\varepsilon_t]=0$ and $\mathbb E\|\varepsilon_t\|^4<\infty$. Under the summability condition $\sum_{j=0}^{\infty}j\|\psi_j\|_{\op}<\infty$, define $\psi(1)=\sum_{j=0}^{\infty}\psi_j$ and $\widetilde\psi_j=-\sum_{k=j+1}^{\infty}\psi_k$. The Phillips--Solo decomposition then gives $\Delta Z_t=\psi(1)\varepsilon_t+\eta_t-\eta_{t-1}$, where $\eta_t=\sum_{j=0}^{\infty}\widetilde\psi_j\varepsilon_{t-j}$ is a mean-zero stationary process. Consequently, $Z_t=Z_0+\psi(1)\sum_{s=1}^t\varepsilon_s+\eta_t-\eta_0$.

To accommodate the deterministic common shape or functional level typically present in empirical functional time series, let $Z_0=\mu_Z+\eta_0$ for some deterministic $\mu_Z\in\mathcal H$. The preceding decomposition then becomes
\begin{equation*} Z_t=\mu_Z+\psi(1)\sum_{s=1}^t\varepsilon_s+\eta_t. \end{equation*}
Since $\mathbb E[\varepsilon_t]=0$ and $\mathbb E[\eta_t]=0$, we have $\mathbb E[Z_t]=\mu_Z$; hence, $\mu_Z=0$ gives the zero-mean normalization.

The operator $\psi(1)$ induces the orthogonal decomposition
\begin{equation*} \mathcal H=\mathcal H^N\oplus\mathcal H^S,\qquad \mathcal H^N=\overline{\ran\psi(1)},\qquad \mathcal H^S=(\mathcal H^N)^\perp=\ker\psi(1)^*, \end{equation*}
where $\overline{\ran\psi(1)}$ denotes the closure of $\ran\psi(1)$. Let $\PP^N$ and $\PP^S=I-\PP^N$ denote the orthogonal projections onto $\mathcal H^N$ and $\mathcal H^S$, respectively. Then $\PP^N(Z_t-\mu_Z)$ contains all stochastic trends, whereas $\PP^S(Z_t-\mu_Z)=\PP^S\eta_t$ is stationary. If $\psi(1)$ has finite rank, then $\mathcal H^N$ is finite-dimensional, and this representation yields the decomposition in \eqref{eqtimedecom} when $Z_t=Y_t$. Let $C_\varepsilon$ denote the covariance operator of $\varepsilon_t$. If $\langle C_\varepsilon h,h\rangle>0$ for every nonzero $h\in\ran\{\psi(1)^*\}$, then the scalar coordinate $\langle Z_t,v\rangle$ is stationary if and only if $v\in\mathcal H^S$ \citep{BSS2017}.

\subsection{$L^p$-$m$-approximability of Hilbert-valued time series}\label{sec_app_math_s1b}

Let $\mathcal H$ denote a Hilbert space equipped with the norm $\|\cdot\|$. For a positive integer $p$, we use the following centered and slightly strengthened version of $L^p$-$m$-approximability. A random sequence $\{\xi_t\}$ in $\mathcal H$ is said to be $L^p$-$m$-approximable if it satisfies the following conditions:
\begin{enumerate}[label=(\roman*)]
\item The sequence admits the Bernoulli-shift representation $\xi_t=\mathfrak F(\epsilon_t,\epsilon_{t-1},\ldots)$ for some measurable mapping $\mathfrak F$ and an i.i.d. innovation sequence $\{\epsilon_t\}$.
\item For each $t$ and $m$, define the coupled process $\xi_{t,m}=\mathfrak F(\epsilon_t,\ldots,\epsilon_{t-m+1},\epsilon_{t,t-m}^{(m)},\epsilon_{t,t-m-1}^{(m)},\ldots)$, where the replacement innovations are independent copies of the original innovations, independent across $t$ and independent of $\{\epsilon_t\}$. Let $\widetilde\xi_{t,m}=\xi_t-\xi_{t,m}$ and $a_m=\{\mathbb E\|\widetilde\xi_{t,m}\|_{\mathcal H}^{p+\delta}\}^{1/(p+\delta)}$. Then $\mathbb E[\xi_t]=0$ and, for some $\delta\in(0,1)$,
\begin{equation*} \mathbb E\|\xi_t\|_{\mathcal H}^{p+\delta}<\infty,\quad ma_m\to0,\quad \sum_{m=1}^{\infty}a_m<\infty. \end{equation*}
\end{enumerate}
This definition strengthens the usual notion of $L^p$-$m$-approximability by imposing the approximation conditions in $L^{p+\delta}$ for some $\delta\in(0,1)$ and by requiring the additional decay condition $ma_m\to0$. Such strengthened $m$-approximation conditions are useful in deriving convergence rates for kernel long-run covariance estimators; see, for example, \citet{berkes2013weak}, \citet{BERKES2016150}, and \citet{horvath2013estimation}.

\section{Testing between-cointegration}\label{sec_app_test}
\noindent The structural validity of model \eqref{eqmodel1} is grounded in the concept of between-cointegration between the processes \(Y_t\) and \(\mathbf{x}_t\). Unlike within-cointegration, which focuses on the internal dynamics of a single process, between-cointegration requires that the nonstationary stochastic trends in \(Y_t\) are entirely spanned by a linear transformation of the nonstationary components in \(\mathbf{x}_t\). In model \eqref{eqmodel1}, this requires that the disturbance term \(U_t\) be stationary in \(\HY\). In our climate application, between-cointegration implies that persistent, long-run shifts in the distribution of temperature anomalies are accounted for by radiative forcings, with only transitory fluctuations remaining as residuals. Verification of this relationship is a crucial prerequisite for our subsequent estimation.

\noindent To keep the baseline asymptotic arguments transparent, we first develop the testing, estimation, and inference procedures for the model without an intercept. Setting $\beta_0=0$ in \eqref{eqmodel1} gives
\begin{equation}\label{eqmodel1add}
Y_t=B(\mathbf{x}_t)+U_t.
\end{equation}
For this baseline specification, we impose the zero-mean normalization
\begin{equation}\label{eqzeromean}
\mathbb{E}[Y_t]=\mu_Y=0
\qquad\text{and}\qquad
\mathbb{E}[\mathbf{x}_t]=\mu_X=0,
\end{equation}
so that the deterministic levels in \eqref{eqtimedecom} and \eqref{eqtimedecomx} are zero. The corresponding results for the model with an intercept are obtained by applying the procedures to the demeaned variables $Y_t-\bar Y_T$ and $\mathbf{x}_t-\bar{\mathbf{x}}_T$. The extension of the between-cointegration test is provided in Section~\ref{sec_app_test_1}.

\subsection{Test statistics}

\noindent Under the zero-intercept specification \eqref{eqmodel1add} and the zero-mean normalization \eqref{eqzeromean}, we define the sample cross-covariance and covariance operators as
\begin{equation}\label{eqcov1}
\widehat C_{Y\mathbf{x}}=\frac{1}{T}\sum_{t=1}^T\mathbf{x}_t\otimes Y_t,\qquad
\widehat C_{\mathbf{x}\mathbf{x}}=\frac{1}{T}\sum_{t=1}^T\mathbf{x}_t\otimes\mathbf{x}_t.
\end{equation}
Since $\widehat C_{\mathbf{x}\mathbf{x}}$ is a nonnegative self-adjoint compact operator on the finite-dimensional space $\HX$, it admits the eigendecomposition
$\widehat C_{\mathbf{x}\mathbf{x}}=\sum_{j=1}^{\KX}\widehat\lambda_{\mathbf{x},j}\widehat v_{\mathbf{x},j}\otimes\widehat v_{\mathbf{x},j}$, where $\widehat\lambda_{\mathbf{x},1}\geq\cdots\geq\widehat\lambda_{\mathbf{x},\KX}\geq0$ \citep[see, e.g.,][]{Bosq2000}. Under the full nonstationary-rank condition imposed below, $\widehat C_{\mathbf{x}\mathbf{x}}$ is invertible with probability approaching one. We define the preliminary least-squares projection map $\Bpre:\HX\rightarrow\HY$ and the corresponding residuals by
\begin{equation}\label{eqls0}
\Bpre=\widehat C_{Y\mathbf{x}}\widehat C_{\mathbf{x}\mathbf{x}}^{-1},\qquad
\Upre_t=Y_t-\Bpre(\mathbf{x}_t).
\end{equation}
Equivalently, the fitted component can be computed as
\begin{equation}
\Bpre(\mathbf{x}_t)=\frac{1}{T}\sum_{s=1}^T\sum_{j=1}^{\KX}
\widehat\lambda_{\mathbf{x},j}^{-1}
\langle\mathbf{x}_t,\widehat v_{\mathbf{x},j}\rangle
\langle\mathbf{x}_s,\widehat v_{\mathbf{x},j}\rangle Y_s.
\end{equation}

Using the residuals $\Upre_t$, we construct the residual partial-sum operator and the residual covariance operator
\begin{equation}
\widehat{\mathcal K}=\frac{1}{T}\sum_{t=1}^T\left(\sum_{s=1}^t\Upre_s\right)\otimes\left(\sum_{s=1}^t\Upre_s\right),
\qquad
\widehat{\mathcal V}
=\frac{1}{T}\sum_{t=1}^T\Upre_t\otimes\Upre_t.
\end{equation}
The operator $\widehat{\mathcal K}$ measures accumulated residual persistence, whereas $\widehat{\mathcal V}$ measures the contemporaneous scale of the fitted residuals. Under between-cointegration, the fitted residuals have the partial-sum order associated with a stationary process, although estimation of $B$ affects their limiting process. Under the nondegenerate alternative considered below, the residuals retain a stochastic trend, causing the partial-sum operator to grow at a faster rate.

We first consider the unnormalized diagnostic
\begin{equation}\label{eqteststat_K} \widehat{\mathcal T}_{K}=T^{-1}\Lambda_{\max}(\widehat{\mathcal K})=\Lambda_{\max}(T^{-1}\widehat{\mathcal K}), \end{equation}
where $\Lambda_{\max}(A)$ denotes the largest eigenvalue of $A$. The statistic $\widehat{\mathcal T}_{K}$ depends on the scale of $\Upre_t$. For a covariance-normalized diagnostic, we use the residual covariance operator to select its dominant direction and normalize its scale. Let $\hat v_V$ be a unit eigenvector associated with the largest eigenvalue of $\widehat{\mathcal V}$. We define
\begin{equation}\label{eqteststat_V} \widehat{\mathcal T}_{V}=T^{-1}\frac{\langle\widehat{\mathcal K}\hat v_V,\hat v_V\rangle}{\langle\widehat{\mathcal V}\hat v_V,\hat v_V\rangle}. \end{equation}
This statistic evaluates residual persistence along the covariance-dominant residual direction and normalizes it by the contemporaneous residual variance in that direction. The statistic $\widehat{\mathcal T}_{V}$ is invariant to a common nonzero rescaling of the residual process: replacing $\Upre_t$ with $c\Upre_t$ for any $c\neq0$ leaves $\widehat{\mathcal T}_{V}$ unchanged. We report both statistics in the empirical analysis.

Throughout the remainder of this supplement, $X_T\Rightarrow X$ denotes weak convergence in the relevant space or product space. For stochastic-process-valued random elements, weak convergence is taken in the corresponding Skorohod space on $[0,1]$; for operator-valued random elements, it is taken under the topology induced by the operator norm. We use the following additional notation:
\begin{equation*} \mathcal A_T=\frac{1}{T}\sum_{t=1}^T\mathbf{x}_t\otimes U_t,\qquad J_T=\frac{1}{T^2}\sum_{t=1}^T\mathbf{x}_t\otimes\mathbf{x}_t,\qquad Q_T(r)=T^{-3/2}\sum_{s=1}^{\lfloor Tr\rfloor}\mathbf{x}_s. \end{equation*}
Here $\mathcal A_T\in\mathcal L_{\HX,\HY}$, $J_T\in\mathcal L_{\HX}$, and $Q_T(r)\in\HX$. Under between-cointegration, let $(W_{\mathbf{x}},W_U)$ denote the joint Brownian limit associated with $(\Delta\mathbf{x}_t,U_t)$ and define
\begin{equation*} \mathcal A=\int_0^1W_{\mathbf{x}}(r)\otimes dW_U(r)+\Omega_{U\mathbf{x}}^{+},\qquad J=\int_0^1W_{\mathbf{x}}(r)\otimes W_{\mathbf{x}}(r)\,dr,\qquad Q(r)=\int_0^rW_{\mathbf{x}}(u)\,du. \end{equation*}
Here $\Omega_{U\mathbf{x}}^{+}$ is the one-sided long-run covariance operator from $\HX$ to $\HY$ introduced in Section~\ref{sec_compest}; its kernel estimator $\widehat{\Omega}_{U\mathbf{x}}^{+}$ is given in \eqref{eqsample2}. Define the limiting residual process and its associated operator by
\begin{equation*} \mathcal R(r)=W_U(r)-\mathcal AJ^{-1}Q(r),\qquad \mathcal K_R=\int_0^1\mathcal R(r)\otimes\mathcal R(r)\,dr. \end{equation*}
For the covariance-normalized null limit, also define
\begin{equation*} C_{UU}=\mathbb E[U_t\otimes U_t], \end{equation*}
and let $v_U$ be a unit eigenvector associated with the largest eigenvalue of $C_{UU}$, under an arbitrary sign convention.

Under the alternative, let $(W_{\mathbf{x}},W_V)$ denote the joint Brownian limit associated with $(\Delta\mathbf{x}_t,V_t)$, where $V_t=\Delta U_t$, and define
\begin{equation*} W_{V|\mathbf{x}}(r)=W_V(r)-\left(\int_0^1W_{\mathbf{x}}(u)\otimes W_V(u)\,du\right)J^{-1}W_{\mathbf{x}}(r),\qquad W_{V|\mathbf{x}}^{(2)}(r)=\int_0^rW_{V|\mathbf{x}}(u)\,du. \end{equation*}
Finally, let
\begin{equation*} M_V=\int_0^1W_{V|\mathbf{x}}(r)\otimes W_{V|\mathbf{x}}(r)\,dr,\qquad M_K=\int_0^1W_{V|\mathbf{x}}^{(2)}(r)\otimes W_{V|\mathbf{x}}^{(2)}(r)\,dr. \end{equation*}

We impose the following conditions for the proposed tests.

\begin{assumption}\label{assum_test}
The following conditions hold:
\begin{enumerate}[label=(\roman*)]
\item\label{assum_test1} The model \eqref{eqmodel1add} holds, and the predictor satisfies $\mathbf{x}_t=\mathbf{x}_0+\sum_{s=1}^t\Delta\mathbf{x}_s$ with $T^{-1/2}\mathbf{x}_0=o_p(1)$. The increment process $\Delta\mathbf{x}_t$ is a mean-zero stationary $L^4$-$m$-approximable sequence satisfying
\begin{equation*} T^{-1/2}\sum_{s=1}^{\lfloor Tr\rfloor}\Delta\mathbf{x}_s\Rightarrow W_{\mathbf{x}}(r), \end{equation*}
where $W_{\mathbf{x}}$ is a $\KX$-dimensional Brownian motion with positive definite long-run covariance $\Omega_{\mathbf{x}\mathbf{x}}$. Consequently, $J$ is almost surely invertible.

\item\label{assum_test2} Under $H_0$ of between-cointegration, $U_t$ is a mean-zero stationary $\HY$-valued $L^4$-$m$-approximable sequence. Moreover, the joint sequence $Z_t^0=(\Delta\mathbf{x}_t,U_t)$ in $\HX\times\HY$ is mean-zero stationary and $L^4$-$m$-approximable, and satisfies
\begin{equation*} T^{-1/2}\sum_{s=1}^{\lfloor Tr\rfloor}Z_s^0=T^{-1/2}\sum_{s=1}^{\lfloor Tr\rfloor}\begin{pmatrix}\Delta\mathbf{x}_s\\ U_s\end{pmatrix}\Rightarrow\begin{pmatrix}W_{\mathbf{x}}(r)\\ W_U(r)\end{pmatrix}, \end{equation*}
where $W_U$ is an $\HY$-valued Brownian motion jointly distributed with $W_{\mathbf{x}}$. In addition,
\begin{equation*} \left(T^{-1/2}\sum_{s=1}^{\lfloor T\cdot\rfloor}U_s,\ \mathcal A_T,\ J_T,\ Q_T(\cdot)\right)\Rightarrow\left(W_U(\cdot),\ \mathcal A,\ J,\ Q(\cdot)\right). \end{equation*}
For the covariance-normalized statistic $\widehat{\mathcal T}_V$, the largest eigenvalue of $C_{UU}$ is assumed to be strictly positive and simple.

\item\label{assum_test3} Under $H_1$ of failure of between-cointegration, $U_t$ is nonstationary with mean-zero stationary first difference $V_t=\Delta U_t$. Specifically, $U_t=U_0+\sum_{s=1}^tV_s$, $T^{-1/2}U_0=o_p(1)$, and the joint sequence $Z_t^1=(\Delta\mathbf{x}_t,V_t)$ in $\HX\times\HY$ is mean-zero stationary and $L^4$-$m$-approximable, and satisfies
\begin{equation*} T^{-1/2}\sum_{s=1}^{\lfloor Tr\rfloor}Z_s^1=T^{-1/2}\sum_{s=1}^{\lfloor Tr\rfloor}\begin{pmatrix}\Delta\mathbf{x}_s\\ V_s\end{pmatrix}\Rightarrow\begin{pmatrix}W_{\mathbf{x}}(r)\\ W_V(r)\end{pmatrix}, \end{equation*}
where $W_V$ is an $\HY$-valued Brownian motion jointly distributed with $W_{\mathbf{x}}$. We restrict attention to the nondegenerate alternative satisfying $\mathbb P\{W_{V|\mathbf{x}}(\cdot)\equiv0\}=0$. For the covariance-normalized statistic $\widehat{\mathcal T}_V$, we additionally assume that the largest eigenvalue of $M_V$ is simple almost surely.
\end{enumerate}
\end{assumption}

Assumption~\ref{assum_test} collects the conditions needed for the residual-based diagnostics. Assumption~\ref{assum_test}\ref{assum_test1} imposes full nonstationary rank on the forcing vector and thereby rules out within-cointegration in $\mathbf{x}_t$. Assumptions~\ref{assum_test}\ref{assum_test2} and~\ref{assum_test}\ref{assum_test3} impose joint weak-dependence and functional central limit conditions under the null and the alternative, respectively. The additional joint convergence in Assumption~\ref{assum_test}\ref{assum_test2} is stated explicitly because the fitted residuals are constructed after projecting $Y_t$ on the integrated forcing vector; the resulting estimation effect contributes the adjustment term $\mathcal AJ^{-1}Q(r)$ to the null limit. The eigenvalue-separation conditions for $C_{UU}$ and $M_V$ ensure that the covariance-selected directions used in $\widehat{\mathcal T}_V$ have well-defined asymptotic limits. Finally, the nondegenerate alternative in Assumption~\ref{assum_test}\ref{assum_test3} implies that both $M_V$ and $M_K$ have strictly positive largest eigenvalues almost surely.

\begin{theorem}\label{thm1}
Suppose that Assumption~\ref{assum_test}\ref{assum_test1} holds. Under $H_0$, suppose also that Assumption~\ref{assum_test}\ref{assum_test2} holds. Then
\begin{equation*} \widehat{\mathcal T}_{K}\Rightarrow\mathcal T_{K,0}:=\Lambda_{\max}(\mathcal K_R),\qquad \widehat{\mathcal T}_{V}\Rightarrow\mathcal T_{V,0}:=\frac{\langle\mathcal K_Rv_U,v_U\rangle}{\langle C_{UU}v_U,v_U\rangle}. \end{equation*}
Under $H_1$, suppose instead that Assumption~\ref{assum_test}\ref{assum_test3} holds, and let $v_{V,1}$ be a unit eigenvector associated with $\Lambda_{\max}(M_V)$. Then
\begin{equation*} T^{-2}\widehat{\mathcal T}_{K}\Rightarrow\mathcal T_{K,1}:=\Lambda_{\max}(M_K),\qquad T^{-1}\widehat{\mathcal T}_{V}\Rightarrow\mathcal T_{V,1}:=\frac{\langle M_Kv_{V,1},v_{V,1}\rangle}{\langle M_Vv_{V,1},v_{V,1}\rangle}, \end{equation*}
where $\mathcal T_{K,1}>0$ and $\mathcal T_{V,1}>0$ almost surely. Consequently, $\widehat{\mathcal T}_{K}\to_p\infty$ and $\widehat{\mathcal T}_{V}\to_p\infty$ under $H_1$.
\end{theorem}

\begin{proof}[Proof of Theorem~\ref{thm1}]
Under $H_0$, since $Y_t=B(\mathbf{x}_t)+U_t$, we have $\widehat C_{Y\mathbf{x}}=B\widehat C_{\mathbf{x}\mathbf{x}}+\mathcal A_T$, $\Bpre=B+\mathcal A_T\widehat C_{\mathbf{x}\mathbf{x}}^{-1}$, and $\Upre_t=U_t-\mathcal A_T\widehat C_{\mathbf{x}\mathbf{x}}^{-1}\mathbf{x}_t$. Recall that $J_T=T^{-2}\sum_{t=1}^T\mathbf{x}_t\otimes\mathbf{x}_t$. Since $\widehat C_{\mathbf{x}\mathbf{x}}=TJ_T$, we have $\widehat C_{\mathbf{x}\mathbf{x}}^{-1}=T^{-1}J_T^{-1}$ with probability approaching one. Hence, for $r\in[0,1]$,
\begin{equation*} T^{-1/2}\sum_{s=1}^{\lfloor Tr\rfloor}\Upre_s=T^{-1/2}\sum_{s=1}^{\lfloor Tr\rfloor}U_s-\mathcal A_TJ_T^{-1}Q_T(r)\Rightarrow\mathcal R(r). \end{equation*}
By continuous mapping and a Riemann-sum approximation, $T^{-1}\widehat{\mathcal K}\Rightarrow\mathcal K_R$. Continuity of the largest-eigenvalue functional on the space of compact self-adjoint operators then gives $\widehat{\mathcal T}_{K}=\Lambda_{\max}(T^{-1}\widehat{\mathcal K})\Rightarrow\Lambda_{\max}(\mathcal K_R)$.

Moreover, $\Upre_t-U_t=-\mathcal A_TT^{-1}J_T^{-1}\mathbf{x}_t$. Under Assumption~\ref{assum_test}, $\mathcal A_T=O_p(1)$, $J_T^{-1}=O_p(1)$, and $\max_{1\leq t\leq T}\|\mathbf{x}_t\|=O_p(T^{1/2})$, and hence $\max_{1\leq t\leq T}\|\Upre_t-U_t\|=O_p(T^{-1/2})$. Together with the operator law of large numbers for $U_t\otimes U_t$, this yields $\widehat{\mathcal V}\to_p C_{UU}$. By the eigenvalue-separation condition for $C_{UU}$, $\widehat v_V\to_p v_U$ up to sign; see Lemma~4.3 of \citet{Bosq2000}. Since the relevant quadratic forms are invariant to this sign, $T^{-1}\langle\widehat{\mathcal K}\widehat v_V,\widehat v_V\rangle\Rightarrow\langle\mathcal K_Rv_U,v_U\rangle$ and $\langle\widehat{\mathcal V}\widehat v_V,\widehat v_V\rangle\to_p\langle C_{UU}v_U,v_U\rangle$. This establishes the two null limits.

Under $H_1$, standard continuous-mapping arguments applied to Assumption~\ref{assum_test}\ref{assum_test3} give $T^{-1}\mathcal A_T\Rightarrow\int_0^1W_{\mathbf{x}}(r)\otimes W_V(r)\,dr$. Since $\Upre_t=U_t-(T^{-1}\mathcal A_T)J_T^{-1}\mathbf{x}_t$, it follows that
\begin{equation*} T^{-1/2}\Upre_{\lfloor T\cdot\rfloor}\Rightarrow W_{V|\mathbf{x}}(\cdot),\qquad \left(T^{-1}\widehat{\mathcal V},T^{-3}\widehat{\mathcal K}\right)\Rightarrow(M_V,M_K). \end{equation*}
Consequently, $T^{-2}\widehat{\mathcal T}_{K}=\Lambda_{\max}(T^{-3}\widehat{\mathcal K})\Rightarrow\Lambda_{\max}(M_K)$. The nondegeneracy condition implies that $W_{V|\mathbf{x}}$ and $W_{V|\mathbf{x}}^{(2)}$ are not identically zero almost surely, so $\Lambda_{\max}(M_V)>0$ and $\Lambda_{\max}(M_K)>0$ almost surely.

Because the largest eigenvalue of $M_V$ is simple almost surely, the rank-one eigenprojection $\widehat v_V\otimes\widehat v_V$ converges weakly, jointly with the scaled operators above, to $v_{V,1}\otimes v_{V,1}$. Moreover, $\langle M_Vv_{V,1},v_{V,1}\rangle=\Lambda_{\max}(M_V)>0$ almost surely. The scalar process $\langle W_{V|\mathbf{x}}(\cdot),v_{V,1}\rangle$ is therefore not identically zero, which implies $\langle M_Kv_{V,1},v_{V,1}\rangle>0$ almost surely. It follows that $T^{-1}\widehat{\mathcal T}_{V}\Rightarrow\mathcal T_{V,1}>0$ almost surely. The two positive scaled limits imply $\widehat{\mathcal T}_{V}\to_p\infty$ and $\widehat{\mathcal T}_{K}\to_p\infty$. This completes the proof.
\end{proof}

Theorem~\ref{thm1} shows that both diagnostics are right-tailed. Their limiting null distributions are nonpivotal since the residual partial-sum limit $\mathcal R$ contains the estimation effect arising from the projection of $Y_t$ on the integrated forcing vector. We therefore approximate these null distributions by plug-in Monte Carlo simulation, as described in the next section.


\subsection{Implementation}\label{sec_app_test_implement}

We now formalize the plug-in Monte Carlo approximation to the null distributions in Theorem~\ref{thm1} for the zero-intercept specification. The corresponding centered construction for the model with an intercept is described in Section~\ref{sec_app_test_implementation} and justified in Section~\ref{sec_app_test_1} by an extension of the arguments below. The procedure uses the kernel estimators in \eqref{eqsample1}--\eqref{eqsample2}. The estimator $\widehat{\Omega}_{\mathbf{x}\mathbf{x}}$ is unchanged, whereas $\widehat{\Omega}_{U\mathbf{x}}$ and $\widehat{\Omega}_{U\mathbf{x}}^+$ are computed using the zero-intercept residuals $\Upre_t=Y_t-\Bpre(\mathbf{x}_t)$ in place of the demeaned residuals used there. Using the same residuals, we also define $\widehat{\Omega}_{UU}=\sum_{|j|\leq h}\mathrm{k}(j/h)\widehat{\Gamma}_{UU}^{(j)}$, where $\widehat{\Gamma}_{UU}^{(j)}$ is the lag-$j$ sample covariance operator defined in Section~\ref{sec_app_test_implementation}. Since the residuals used in each construction are clear from the specification under consideration, we use the same notation for the residual-based long-run covariance estimators in both cases.

\begin{assumption}\label{assum_test_kernel}
The bandwidth $h$ satisfies $h\to\infty$ and $h/\sqrt T\to0$. The kernel $\mathrm{k}:\mathbb R\to[-1,1]$ is an even, nonnegative-definite, twice continuously differentiable function such that $\mathrm{k}(x)=0$ for $|x|\geq1$, $\mathrm{k}(0)=1$, $\mathrm{k}'(0)=0$, $\mathrm{k}''(0)\neq0$, and $\lim_{|x|\uparrow1}\mathrm{k}(x)/(1-|x|)^2$ is finite.
\end{assumption}

Assumption~\ref{assum_test_kernel} provides the smoothness and bandwidth conditions required for consistent estimation of the long-run and one-sided covariance operators. The condition $h/\sqrt T\to0$ also ensures that replacing $U_t$ with $\Upre_t$ is asymptotically negligible under $H_0$. When all blocks are computed from $(\Delta\mathbf{x}_t,\Upre_t)$ using the same kernel and bandwidth, the nonnegative-definiteness of $\mathrm{k}$ ensures that the joint lag-window estimator is a nonnegative covariance operator. Using these estimators, form the block long-run covariance estimator on $\HX\times\HY$,
\begin{equation*} \widehat{\Omega}_{ZZ}=\begin{pmatrix}\widehat{\Omega}_{\mathbf{x}\mathbf{x}} & \widehat{\Omega}_{U\mathbf{x}}^{*}\\ \widehat{\Omega}_{U\mathbf{x}} & \widehat{\Omega}_{UU}\end{pmatrix}. \end{equation*}
Here $\widehat{\Omega}_{U\mathbf{x}}^{*}$ denotes the adjoint of $\widehat{\Omega}_{U\mathbf{x}}$. Under $H_0$, $\widehat{\Omega}_{ZZ}$ estimates the long-run covariance operator of $Z_t^0=(\Delta\mathbf{x}_t,U_t)$, with $U_t$ replaced by $\Upre_t$. Remark~\ref{rem_prac} describes the practical generation of Brownian paths from $\widehat{\Omega}_{ZZ}$.

In the present zero-intercept construction, no centering is applied to the simulated Brownian paths, and the resulting simulated quantities are written without the superscript $c$. For each $\ei=1,\ldots,R_{\mathrm{MC}}$, conditionally on the data and independently across replications, generate a jointly distributed pair of Brownian paths $(\widehat W_{\mathbf{x},(\ei)},\widehat W_{U,(\ei)})$ with joint covariance operator $\widehat{\Omega}_{ZZ}$. Define
\begin{equation*} \widehat{\mathcal A}_{(\ei)}=\int_0^1\widehat W_{\mathbf{x},(\ei)}(r)\otimes d\widehat W_{U,(\ei)}(r)+\widehat{\Omega}_{U\mathbf{x}}^{+},\qquad \widehat J_{(\ei)}=\int_0^1\widehat W_{\mathbf{x},(\ei)}(r)\otimes\widehat W_{\mathbf{x},(\ei)}(r)\,dr \end{equation*}
and
\begin{equation*} \widehat Q_{(\ei)}(r)=\int_0^r\widehat W_{\mathbf{x},(\ei)}(u)\,du,\qquad \widehat{\mathcal R}_{(\ei)}(r)=\widehat W_{U,(\ei)}(r)-\widehat{\mathcal A}_{(\ei)}\widehat J_{(\ei)}^{-1}\widehat Q_{(\ei)}(r). \end{equation*}
The simulated residual partial-sum operator is
\begin{equation*} \widehat{\mathcal K}_{R,(\ei)}=\int_0^1\widehat{\mathcal R}_{(\ei)}(r)\otimes\widehat{\mathcal R}_{(\ei)}(r)\,dr. \end{equation*}
The resulting simulated null statistics are
\begin{equation*} \widehat{\mathcal T}_{K,0,(\ei)}=\Lambda_{\max}\bigl(\widehat{\mathcal K}_{R,(\ei)}\bigr),\qquad \widehat{\mathcal T}_{V,0,(\ei)}=\frac{\langle\widehat{\mathcal K}_{R,(\ei)}\hat v_V,\hat v_V\rangle}{\langle\widehat{\mathcal V}\hat v_V,\hat v_V\rangle}. \end{equation*}
The sample quantities $\widehat{\mathcal V}$ and $\hat v_V$ are held fixed across Monte Carlo replications.

Let $\widehat q_{K,\alpha}$ and $\widehat q_{V,\alpha}$ be the empirical $(1-\alpha)$ quantiles of $\{\widehat{\mathcal T}_{K,0,(\ei)}\}_{\ei=1}^{R_{\mathrm{MC}}}$ and $\{\widehat{\mathcal T}_{V,0,(\ei)}\}_{\ei=1}^{R_{\mathrm{MC}}}$, respectively. We also let $q_{K,\alpha}$ and $q_{V,\alpha}$ denote the corresponding $(1-\alpha)$ quantiles of $\mathcal T_{K,0}$ and $\mathcal T_{V,0}$.

\begin{proposition}[Validity of the plug-in critical values]\label{prop_plugin_test}
Suppose that Assumptions~\ref{assum_test}\ref{assum_test1} and~\ref{assum_test_kernel} hold and that $R_{\mathrm{MC}}\to\infty$ as $T\to\infty$. Under $H_0$, suppose also that Assumption~\ref{assum_test}\ref{assum_test2} holds. If the $(1-\alpha)$ quantiles $q_{K,\alpha}$ and $q_{V,\alpha}$ are uniquely identified and the distribution functions of $\mathcal T_{K,0}$ and $\mathcal T_{V,0}$ are continuous at these quantiles, then
\begin{equation*} \widehat q_{K,\alpha}\to_p q_{K,\alpha},\qquad \widehat q_{V,\alpha}\to_p q_{V,\alpha}, \end{equation*}
and
\begin{equation*} \mathbb P\{\widehat{\mathcal T}_{K}>\widehat q_{K,\alpha}\}\to\alpha,\qquad \mathbb P\{\widehat{\mathcal T}_{V}>\widehat q_{V,\alpha}\}\to\alpha. \end{equation*}
Under $H_1$, suppose instead that Assumption~\ref{assum_test}\ref{assum_test3} holds. Then
\begin{equation*} \mathbb P\{\widehat{\mathcal T}_{K}>\widehat q_{K,\alpha}\}\to1,\qquad \mathbb P\{\widehat{\mathcal T}_{V}>\widehat q_{V,\alpha}\}\to1. \end{equation*}
\end{proposition}

\begin{proof}
We first consider $H_0$. From the proof of Theorem~\ref{thm1}, $\max_{1\leq t\leq T}\|\Upre_t-U_t\|=O_p(T^{-1/2})$. Consequently, the difference between each residual-based kernel estimator and its infeasible counterpart based on $U_t$ is $O_p(h/\sqrt T)=o_p(1)$. The joint $L^4$-$m$-approximability of $Z_t^0=(\Delta\mathbf{x}_t,U_t)$ and standard long-run covariance consistency results for functional time series, including Theorem~2 of \citet{horvath2013estimation} and the analogous one-sided argument, therefore give consistency of $\widehat{\Omega}_{\mathbf{x}\mathbf{x}}$, $\widehat{\Omega}_{U\mathbf{x}}$, $\widehat{\Omega}_{UU}$, and $\widehat{\Omega}_{U\mathbf{x}}^+$ for their population counterparts.

These results yield operator-norm consistency of $\widehat{\Omega}_{ZZ}$. Moreover, the trace of the block estimator is the sum of the traces of its diagonal blocks, $\operatorname{tr}(\widehat{\Omega}_{ZZ})=\operatorname{tr}(\widehat{\Omega}_{\mathbf{x}\mathbf{x}})+\operatorname{tr}(\widehat{\Omega}_{UU})$. For each lag $j$, the traces of the corresponding sample covariance operators are sample averages of the scalar inner products $\langle\Delta\mathbf{x}_t,\Delta\mathbf{x}_{t+j}\rangle$ and $\langle\Upre_t,\Upre_{t+j}\rangle$. Hence, $\operatorname{tr}(\widehat{\Omega}_{ZZ})$ is a scalar lag-window estimator of the total long-run variance $\operatorname{tr}(\Omega_{ZZ})$. The same weak-dependence and bandwidth conditions, together with the asymptotic negligibility of replacing $U_t$ by $\Upre_t$, therefore give $\operatorname{tr}(\widehat{\Omega}_{ZZ})\to_p\operatorname{tr}(\Omega_{ZZ})$. Since both operators are nonnegative, operator-norm convergence together with trace convergence implies convergence in nuclear norm and hence $\bigl\|\widehat{\Omega}_{ZZ}^{1/2}-\Omega_{ZZ}^{1/2}\bigr\|_{\mathrm{HS}}\to_p0$. 
Under a coupling based on a common cylindrical Brownian motion, this convergence implies that the conditional law of the simulated joint Brownian paths converges weakly in probability to the law of $(W_{\mathbf{x}},W_U)$. The Itô isometry further yields joint convergence of the stochastic integral in $\widehat{\mathcal A}_{(\ei)}$ to $\int_0^1W_{\mathbf{x}}(r)\otimes dW_U(r)$.

Moreover, $\widehat{\mathcal V}\to_p C_{UU}$, and the eigenvalue-separation condition gives $\hat v_V\otimes\hat v_V\to_p v_U\otimes v_U$, avoiding dependence on the arbitrary signs of the eigenvectors. Together with $\widehat{\Omega}_{U\mathbf{x}}^+\to_p\Omega_{U\mathbf{x}}^+$ and the almost-sure invertibility of $J$, these results imply that the conditional joint distribution of $(\widehat{\mathcal T}_{K,0,(\ei)},\widehat{\mathcal T}_{V,0,(\ei)})$ converges weakly in probability to the distribution of $(\mathcal T_{K,0},\mathcal T_{V,0})$. The convergence of the empirical Monte Carlo quantiles follows from the uniqueness and continuity conditions and $R_{\mathrm{MC}}\to\infty$. The size statements then follow from Theorem~\ref{thm1}.

We next consider $H_1$. All stochastic orders involving simulated objects below are understood conditionally on the data. Under the nondegenerate alternative, $T^{-1/2}\Upre_{\lfloor T\cdot\rfloor}\Rightarrow W_{V|\mathbf{x}}(\cdot)$, and hence $\max_{1\leq t\leq T}\|\Upre_t\|=O_p(\sqrt T)$. The kernel covariance estimators constructed from $(\Delta\mathbf{x}_t,\Upre_t)$ no longer estimate stationary long-run covariance objects, but their orders remain controlled. Since $\Delta\mathbf{x}_t$ is stationary, $\|\widehat{\Omega}_{\mathbf{x}\mathbf{x}}\|_{\op}=O_p(1)$. For each $|j|\leq h$, $\|\widehat{\Gamma}_{U\mathbf{x}}^{(j)}\|_{\op}$ is bounded by $T^{-1}\sum_t\|\Upre_t\|\,\|\Delta\mathbf{x}_{t-j}\|=O_p(\sqrt T)$, uniformly over these lags. It follows that $\|\widehat{\Omega}_{U\mathbf{x}}\|_{\op}=O_p(h\sqrt T)$ and $\|\widehat{\Omega}_{U\mathbf{x}}^+\|_{\op}=O_p(h\sqrt T)$. Similarly, $\|\widehat{\Gamma}_{UU}^{(j)}\|_{\op}=O_p(T)$ uniformly for $|j|\leq h$, so $\|\widehat{\Omega}_{UU}\|_{\op}=O_p(hT)$. The same calculation in nuclear norm gives $\operatorname{tr}(\widehat{\Omega}_{UU})=O_p(hT)$. The fixed-dimensional $\mathbf{x}$-component of the simulated Brownian motion therefore satisfies $\sup_{0\leq r\leq1}\|\widehat W_{\mathbf{x},(\ei)}(r)\|=O_p(1)$, while the conditional maximal inequality and the trace bound give $\sup_{0\leq r\leq1}\|\widehat W_{U,(\ei)}(r)\|=O_p(\sqrt{hT})$. 
Because $\widehat{\Omega}_{\mathbf{x}\mathbf{x}}\to_p\Omega_{\mathbf{x}\mathbf{x}}$ and $\Omega_{\mathbf{x}\mathbf{x}}$ is positive definite, $\|\widehat J_{(\ei)}\|_{\op}=O_p(1)$, $\|\widehat J_{(\ei)}^{-1}\|_{\op}=O_p(1)$, and $\sup_{0\leq r\leq1}\|\widehat Q_{(\ei)}(r)\|=O_p(1)$. By the Itô isometry, the stochastic integral in $\widehat{\mathcal A}_{(\ei)}$ is $O_p(\sqrt{hT})$, whereas $\widehat{\Omega}_{U\mathbf{x}}^+=O_p(h\sqrt T)$. Hence, $\|\widehat{\mathcal A}_{(\ei)}\|_{\op}=O_p(h\sqrt T)$ and $\sup_{0\leq r\leq1}\|\widehat{\mathcal R}_{(\ei)}(r)\|=O_p(h\sqrt T)$. It follows that $\|\widehat{\mathcal K}_{R,(\ei)}\|_{\op}=O_p(h^2T)$. The corresponding conditional $(1-\alpha)$ quantile is therefore $O_p(h^2T)$, and $R_{\mathrm{MC}}\to\infty$ gives $\widehat q_{K,\alpha}=O_p(h^2T)$.

For the normalized statistic, the proof of Theorem~\ref{thm1} gives $T^{-1}\langle\widehat{\mathcal V}\hat v_V,\hat v_V\rangle\Rightarrow\Lambda_{\max}(M_V)>0$ almost surely. Combining this denominator rate with $\|\widehat{\mathcal K}_{R,(\ei)}\|_{\op}=O_p(h^2T)$ shows that the conditional $(1-\alpha)$ quantile of $\widehat{\mathcal T}_{V,0,(\ei)}$ is $O_p(h^2)$ and hence that $\widehat q_{V,\alpha}=O_p(h^2)$. Since $h/\sqrt T\to0$, it follows that $T^{-2}\widehat q_{K,\alpha}\to_p0$ and $T^{-1}\widehat q_{V,\alpha}\to_p0$. On the other hand, Theorem~\ref{thm1} gives $T^{-2}\widehat{\mathcal T}_{K}\Rightarrow\mathcal T_{K,1}$ and $T^{-1}\widehat{\mathcal T}_{V}\Rightarrow\mathcal T_{V,1}$, where $\mathcal T_{K,1}>0$ and $\mathcal T_{V,1}>0$ almost surely. The two rejection probabilities therefore converge to one.
\end{proof}

\begin{remark}\label{rem_prac}
At the operator level, $\widehat{\Omega}_{ZZ}$ is a finite-rank nonnegative self-adjoint covariance operator on the product Hilbert space $\HX\times\HY$. Let $\widehat{\Omega}_{ZZ}=\sum_{j\geq1}\widehat\mu_j\widehat\psi_j\otimes\widehat\psi_j$ be its spectral decomposition, where only finitely many $\widehat\mu_j$ are nonzero. For each Monte Carlo replication $\ei$, the increment of $\widehat W_{Z,(\ei)}=(\widehat W_{\mathbf{x},(\ei)},\widehat W_{U,(\ei)})$ over the $\ell$th interval of a grid with $m$ equal intervals may be generated as
\begin{equation*} \Delta\widehat W_{Z,\ell,(\ei)}=m^{-1/2}\sum_{j\geq1}\widehat\mu_j^{1/2}\xi_{\ell j,(\ei)}\widehat\psi_j,\qquad \ell=1,\ldots,m, \end{equation*}
where the $\xi_{\ell j,(\ei)}$ are independent standard normal variables across $\ell$, $j$, and $\ei$. Cumulative sums of these increments give the simulated Brownian paths. In practical computation, the functional objects are represented by finitely many coefficient vectors, so this construction reduces to forming the matrix representation of $\widehat{\Omega}_{ZZ}$, taking its symmetric square root, and generating Gaussian increments with that covariance. For the model with an intercept, the centered paths are formed from these initially generated paths as described in Section~\ref{sec_app_test_implementation}.
\end{remark}

\subsection{Testing between-cointegration in the model with an intercept}
\label{sec_app_test_1}

Theorem~\ref{thm1} and Proposition~\ref{prop_plugin_test} are stated for the zero-intercept specification. For the model with an intercept, the same arguments are applied to the demeaned variables $\widetilde Y_t=Y_t-\bar Y_T$ and $\widetilde{\mathbf{x}}_t=\mathbf{x}_t-\bar{\mathbf{x}}_T$. Under $H_0$, $\widetilde Y_t=B(\widetilde{\mathbf{x}}_t)+\widetilde U_t$, where $\widetilde U_t=U_t-\bar U_T$. Demeaning affects the regressor-level and disturbance partial-sum limits differently:
\begin{equation*}
T^{-1/2}\widetilde{\mathbf{x}}_{\lfloor Tr\rfloor}\Rightarrow W_{\mathbf{x}}^c(r),
\qquad
T^{-1/2}\sum_{s=1}^{\lfloor Tr\rfloor}\widetilde U_s\Rightarrow W_U^c(r),
\end{equation*}
where
\begin{equation*}
W_{\mathbf{x}}^c(r)=W_{\mathbf{x}}(r)-\int_0^1W_{\mathbf{x}}(u)\,du,
\qquad
W_U^c(r)=W_U(r)-rW_U(1).
\end{equation*}

Define $\mathcal A_T^c=T^{-1}\sum_{t=1}^T\widetilde{\mathbf{x}}_t\otimes\widetilde U_t$, $J_T^c=T^{-2}\sum_{t=1}^T\widetilde{\mathbf{x}}_t\otimes\widetilde{\mathbf{x}}_t$, and $Q_T^c(r)=T^{-3/2}\sum_{s=1}^{\lfloor Tr\rfloor}\widetilde{\mathbf{x}}_s$. Since $\sum_{t=1}^T\widetilde{\mathbf{x}}_t=0$, replacing $\widetilde U_t$ with $U_t$ in $\mathcal A_T^c$ has no effect, and
\begin{equation*}
\mathcal A_T^c=\mathcal A_T-\bar{\mathbf{x}}_T\otimes\bar U_T.
\end{equation*}
Let $\bar W_{\mathbf{x}}=\int_0^1W_{\mathbf{x}}(r)\,dr$. The joint convergence in Assumption~\ref{assum_test}\ref{assum_test2} implies that $T^{-1/2}\bar{\mathbf{x}}_T\Rightarrow\bar W_{\mathbf{x}}$ and $T^{1/2}\bar U_T\Rightarrow W_U(1)$. Together with the continuous mapping theorem, this yields $\mathcal A_T^c\Rightarrow\mathcal A^c$, $J_T^c\Rightarrow J^c$, and $Q_T^c(\cdot)\Rightarrow Q^c(\cdot)$, where
\begin{align*}
\mathcal A^c
&=\mathcal A-\bar W_{\mathbf{x}}\otimes W_U(1)
=\int_0^1W_{\mathbf{x}}(r)\otimes dW_U(r)
+\Omega_{U\mathbf{x}}^+
-\bar W_{\mathbf{x}}\otimes W_U(1),\\
J^c
&=\int_0^1W_{\mathbf{x}}^c(r)\otimes W_{\mathbf{x}}^c(r)\,dr,
\qquad
Q^c(r)=\int_0^rW_{\mathbf{x}}^c(u)\,du.
\end{align*}
Positive definiteness of $\Omega_{\mathbf{x}\mathbf{x}}$ implies that $J^c$ is almost surely invertible.

To connect this expression with the centered simulation in Section~\ref{sec_app_test_implementation}, interpret the centered stochastic integral through the corresponding left-endpoint-sum limit. Since $W_U^c(r)=W_U(r)-rW_U(1)$ and $\int_0^1W_{\mathbf{x}}^c(r)\,dr=0$,
\begin{equation*}
\int_0^1W_{\mathbf{x}}^c(r)\otimes dW_U^c(r)=\int_0^1W_{\mathbf{x}}(r)\otimes dW_U(r)-\bar W_{\mathbf{x}}\otimes W_U(1).
\end{equation*}
Hence, $\mathcal A^c=\int_0^1W_{\mathbf{x}}^c(r)\otimes dW_U^c(r)+\Omega_{U\mathbf{x}}^+$, and $\widehat{\mathcal A}_{(\ei)}^c$ in \eqref{eq_app_test_sim_AJ} is its direct plug-in counterpart.

The centered residual partial-sum limit and its associated operator are
\begin{equation*}
\mathcal R^c(r)=W_U^c(r)-\mathcal A^c(J^c)^{-1}Q^c(r),
\qquad
\mathcal K_R^c=\int_0^1\mathcal R^c(r)\otimes\mathcal R^c(r)\,dr.
\end{equation*}
In particular, $\mathcal R^c(1)=0$, consistent with the fact that the least-squares residuals from the model with an intercept sum to zero. Consequently, under $H_0$,
\begin{equation*}
\widehat{\mathcal T}_K\Rightarrow\Lambda_{\max}(\mathcal K_R^c),
\qquad
\widehat{\mathcal T}_V\Rightarrow
\frac{\langle\mathcal K_R^cv_U,v_U\rangle}
{\langle C_{UU}v_U,v_U\rangle}.
\end{equation*}

The centered simulation defined in \eqref{eq_app_test_centered_paths}--\eqref{eq_app_test_null_draws} provides direct plug-in counterparts of these limiting objects. Accordingly, the proof of Proposition~\ref{prop_plugin_test} applies with the two null limits replaced by their centered counterparts, provided that the corresponding quantiles satisfy the same uniqueness and continuity conditions. The simulated critical values are therefore consistent, and both tests have asymptotic size $\alpha$ under the model with an intercept.

Under $H_1$, let
\begin{equation*}
W_V^c(r)=W_V(r)-\int_0^1W_V(u)\,du
\end{equation*}
and define
\begin{equation*}
W_{V|\mathbf{x}}^c(r)
=
W_V^c(r)
-
\left(\int_0^1W_{\mathbf{x}}^c(u)\otimes W_V^c(u)\,du\right)
(J^c)^{-1}W_{\mathbf{x}}^c(r).
\end{equation*}
Define $M_V^c$ and $M_K^c$ from $W_{V|\mathbf{x}}^c$ in the same way that $M_V$ and $M_K$ are defined from $W_{V|\mathbf{x}}$. For the intercept specification, the corresponding conditions in Assumption~\ref{assum_test}\ref{assum_test3} are replaced by $\mathbb P\{W_{V|\mathbf{x}}^c(\cdot)\equiv0\}=0$ and the requirement that the largest eigenvalue of $M_V^c$ be simple almost surely. The divergence rates remain $T^2$ for $\widehat{\mathcal T}_K$ and $T$ for $\widehat{\mathcal T}_V$, so the power conclusions of Proposition~\ref{prop_plugin_test} continue to hold.

\subsection{Finite-sample performance}\label{sec_sim_between}
\noindent We use simulation experiments to examine the size and power of the between-cointegration tests $\widehat{\mathcal T}_K$ and $\widehat{\mathcal T}_V$ under different degrees of serial dependence and endogeneity. The functional response is represented by $J=20$ nonconstant orthonormal Fourier functions on $[0,1]$, with $\phi_{2m-1}(s)=\sqrt{2}\sin(2\pi ms)$ and $\phi_{2m}(s)=\sqrt{2}\cos(2\pi ms)$ for $m=1,\ldots,10$. These basis functions integrate to zero, as do the CLR responses in our empirical analysis. All calculations are performed directly in Fourier coordinates, whose Euclidean inner products coincide with the corresponding $L^2$ inner products.

\indent Let $\varepsilon_t\sim N(0,I_2)$, $z_t\sim N(0,I_J)$, and $\nu_t\sim N(0,1)$ be mutually independent sequences of independent innovations, where $I_k$ denotes the $k\times k$ identity matrix. Starting from $\mathbf{x}_0=0$, the predictor vector is generated by
\begin{equation*}
\mathbf{x}_t=\mathbf{x}_{t-1}+\Delta\mathbf{x}_t,\quad \Delta\mathbf{x}_t=\rho\Delta\mathbf{x}_{t-1}+\sqrt{1-\rho^2}L_x\varepsilon_t,\quad L_x=\begin{pmatrix}1&0\\0.3&\sqrt{1-0.3^2}\end{pmatrix}.
\end{equation*}
Under the stationary initialization specified below, $\Delta\mathbf{x}_t$ has unit marginal variances and contemporaneous correlation $0.3$. Its long-run covariance is positive definite, ensuring that the predictor vector carries two linearly independent stochastic trends. To introduce endogeneity, define $\zeta_{m,t}=\eta\varepsilon_{m,t}+\sqrt{1-\eta^2}z_{m,t}$ for $m=1,2$, and $\zeta_{m,t}=z_{m,t}$ for $m=3,\ldots,J$. The stationary component of the disturbance is
\begin{equation*}
U_t^S(s)=\sum_{m=1}^{J}\frac{a_{m,t}}{m}\phi_m(s),\quad a_{m,t}=\rho a_{m,t-1}+\sqrt{1-\rho^2}\zeta_{m,t}.
\end{equation*}
The increment and coefficient processes are initialized jointly from their stationary Gaussian distribution. Thus, $U_t^S$ has contemporaneous covariance eigenvalues $m^{-2}$ and a unique leading direction $\phi_1$. We consider three settings for $(\Delta\mathbf{x}_t,U_t^S)$: $(\rho,\eta)=(0,0)$, with i.i.d. observations (IID); $(0.3,0)$, with serial dependence (Serial); and $(0.3,0.3)$, with serial dependence and endogeneity (Serial + endogeneity). The response is generated as
\begin{equation*}
Y_t(s)=\beta_0(s)+\beta_1(s)x_{1,t}+\beta_2(s)x_{2,t}+U_t(s),\quad U_t(s)=U_t^S(s)+\delta q_t\phi_1(s),\quad q_t=\sum_{v=1}^{t}\nu_v,
\end{equation*}
where $\beta_0=0.2\phi_1-0.1\phi_2$, $\beta_1=\phi_1+0.5\phi_3$, and $\beta_2=0.8\phi_2-0.4\phi_4$. Setting $\delta=0$ gives the null of between-cointegration. For power, we set $\delta=0.3$, introducing an independent stochastic trend along $\phi_1$ that cannot be accounted for by the predictor vector.

\indent For each design and sample size $T\in\{200,400,800\}$, we conduct 1,000 Monte Carlo replications with independently generated samples. Both tests use residuals from the least-squares regression with an intercept and are calibrated as described in Section~\ref{sec_app_test_implementation}. The covariance operators are re-estimated for every sample using the Parzen kernel with $h$ chosen as the nearest integer to $T^{1/4}$, giving $h=4,4,5$ at the three sample sizes. All $J$ response directions are retained. Each calibration uses 2,000 auxiliary Brownian draws on a grid of 499 subintervals, incorporating the demeaning, projection, and one-sided covariance corrections. We use fewer auxiliary draws than those used for the empirical tests to reduce the computational cost of recalibrating the tests in every Monte Carlo replication. We generate standard Brownian paths independently of the data and reuse them across samples, using each sample's covariance estimates to simulate its null distribution. The leading eigenvector of the residual covariance operator and its eigenvalue are re-estimated for each sample and held fixed during these simulations.

\indent Tables~\ref{tab:bc_size} and~\ref{tab:bc_power} report size and power at the 5\% nominal level. The unnormalized statistic $\widehat{\mathcal T}_K$ is the primary diagnostic, with the covariance-normalized statistic $\widehat{\mathcal T}_V$ reported for comparison. Entries are rejection frequencies in percent, with Monte Carlo standard errors in parentheses, conditional on the common auxiliary draws. The tables assess null calibration and sensitivity to the additional stochastic trend across dependence designs and sample sizes. Power is reported without size adjustment and is interpreted alongside the size results.

\begin{table}[h!]
\centering
\renewcommand{\arraystretch}{0.7}
\caption{Empirical size of the between-cointegration tests. Rejection frequencies in percent at the 5\% nominal level.}
\label{tab:bc_size}
\begin{tabular}{lccccc}
\toprule
Design & $\rho$ & $\eta$ & $T$  & $\widehat{\mathcal T}_K$ & $\widehat{\mathcal T}_V$ \\
\midrule
IID & 0 & 0 & 200 &  3.1 (0.5) & 2.8 (0.5) \\
IID & 0 & 0 & 400 &  4.0 (0.6) & 4.0 (0.6) \\
IID & 0 & 0 & 800 &  5.3 (0.7) & 4.6 (0.7) \\
Serial & 0.3 & 0 & 200 &  7.2 (0.8) & 6.5 (0.8) \\
Serial & 0.3 & 0 & 400 &  8.9 (0.9) & 8.6 (0.9) \\
Serial & 0.3 & 0 & 800 &  5.3 (0.7) & 6.0 (0.8) \\
Serial + endogeneity & 0.3 & 0.3 & 200 &  6.8 (0.8) & 7.0 (0.8) \\
Serial + endogeneity & 0.3 & 0.3 & 400 &  6.9 (0.8) & 6.0 (0.8) \\
Serial + endogeneity & 0.3 & 0.3 & 800 &  6.1 (0.8) & 6.0 (0.8) \\
\bottomrule
\end{tabular}
\par\smallskip\begin{minipage}{\textwidth}\footnotesize Each cell uses 1,000 independent data replications. Critical values are reestimated for each sample using 2,000 auxiliary draws and the Parzen kernel with $h$ chosen as the nearest integer to $T^{1/4}$. Parentheses give Monte Carlo standard errors in percentage points, conditional on the common auxiliary draws. \end{minipage}
\end{table}

\begin{table}[h!]
\centering
\renewcommand{\arraystretch}{0.7}
\caption{Empirical power of the between-cointegration tests with $\delta=0.3$. Rejection frequencies in percent at the 5\% nominal level.}
\label{tab:bc_power}
\begin{tabular}{lccccc}
\toprule
Design & $\rho$ & $\eta$ & $T$  & $\widehat{\mathcal T}_K$ & $\widehat{\mathcal T}_V$ \\
\midrule
IID & 0 & 0 & 200 &  89.8 (1.0) & 90.7 (0.9) \\
IID & 0 & 0 & 400 & 98.9 (0.3) & 99.0 (0.3) \\
IID & 0 & 0 & 800 & 99.9 (0.1) & 99.9 (0.1) \\
Serial & 0.3 & 0 & 200 & 87.3 (1.1) & 88.0 (1.0) \\
Serial & 0.3 & 0 & 400 &  97.7 (0.5) & 97.9 (0.5) \\
Serial & 0.3 & 0 & 800 &  100.0 (0.0) & 100.0 (0.0) \\
Serial + endogeneity & 0.3 & 0.3 & 200 & 87.7 (1.0) & 88.1 (1.0) \\
Serial + endogeneity & 0.3 & 0.3 & 400 & 97.8 (0.5) & 97.8 (0.5) \\
Serial + endogeneity & 0.3 & 0.3 & 800 & 99.9 (0.1) & 99.9 (0.1) \\
\bottomrule
\end{tabular}
\par\smallskip\begin{minipage}{\textwidth}\footnotesize The Monte Carlo settings are as in Table~\ref{tab:bc_size}. Power is the rejection frequency under $U_t=U_t^S+0.3q_t\phi_1$.\end{minipage}
\end{table}

\indent Table~\ref{tab:bc_size} shows that both tests are conservative under the IID design at $T=200$, with rejection rates closer to the nominal level at larger sample sizes. Serial dependence produces some overrejection at $T=200$ and $T=400$, reaching 8.9\%. At $T=800$, rejection rates across all designs range from 4.6\% to 6.1\%.

\indent Table~\ref{tab:bc_power} shows high power against the additional stochastic trend. Rejection rates range from 87.3\% to 90.7\% at $T=200$, increase to 97.7\%--99.0\% at $T=400$, and reach at least 99.9\% at $T=800$. The covariance-normalized statistic performs similarly to the unnormalized statistic under this alternative. These power results are interpreted alongside the size distortions.

\section{Inference for the model without an intercept}\label{sec_just_1}

\noindent This section provides the formal theoretical justification for estimation and inference under the zero-intercept specification \eqref{eqmodel1add} and the zero-mean normalization \eqref{eqzeromean}. The corresponding results for the model with an intercept are developed in Section~\ref{sec_app_det}.

\subsection{Proposed estimator and asymptotic properties}\label{sec_app_est1}

Under between-cointegration, we proceed to estimate the long-run operator $B$. The standard least-squares estimator exhibits second-order biases induced by endogeneity and serial correlation. To remove these biases, we adapt the fully modified least-squares approach of \citet{phillips1995fully}.

We impose the following conditions throughout this section.

\begin{assumption}[Conditions for estimation under between-cointegration]\label{assum_est}
The following conditions hold:
\begin{enumerate}[label=(\roman*)]
\item\label{assum_est1} The model \eqref{eqmodel1add} holds, $\mathbf{x}_t=\mathbf{x}_0+\sum_{s=1}^t\Delta\mathbf{x}_s$, and $T^{-1/2}\mathbf{x}_0=o_p(1)$. The disturbance $U_t$ is stationary, so $Y_t$ and $\mathbf{x}_t$ are between-cointegrated. The long-run covariance matrix $\Omega_{\mathbf{x}\mathbf{x}}$ is positive definite, so $\mathbf{x}_t$ has full nonstationary rank and is not within-cointegrated.

\item\label{assum_est2} Let $\mathcal H_Z=\HX\times\HY$, equipped with the product inner product $\langle \cdot,\cdot \rangle_{\mathcal H_Z}$, and define $Z_t=(\Delta\mathbf{x}_t,U_t)'\in\mathcal H_Z$. The joint process $Z_t$ is $L^4$-$m$-approximable and satisfies
\begin{equation*}
T^{-1/2}\sum_{s=1}^{\lfloor Tr\rfloor}Z_s
=
T^{-1/2}\sum_{s=1}^{\lfloor Tr\rfloor}
\begin{pmatrix}
\Delta\mathbf{x}_s\\
U_s
\end{pmatrix}
\Rightarrow
W_Z(r):=
\begin{pmatrix}
W_{\mathbf{x}}(r)\\
W_U(r)
\end{pmatrix}.
\end{equation*}

\item\label{assum_est3} For every $v\in\HX$ and $z\in\mathcal H_Z$, the following convergence holds jointly over any finite collection of such pairs and jointly with the FCLT in part~\ref{assum_est2}:
\begin{equation*}
\frac{1}{T}\sum_{t=1}^T\langle\mathbf{x}_{t-1},v\rangle\langle Z_t,z\rangle_{\mathcal H_Z}\Rightarrow\int_0^1\langle W_{\mathbf{x}}(r),v\rangle\,d\langle W_Z(r),z\rangle_{\mathcal H_Z}+\sum_{j=1}^{\infty}\mathbb E\!\left[\langle\Delta\mathbf{x}_t,v\rangle\langle Z_{t+j},z\rangle_{\mathcal H_Z}\right].
\end{equation*}

\item\label{assum_est4} The bandwidth $h$ and kernel $\mathrm{k}$ satisfy Assumption~\ref{assum_test_kernel}.
\end{enumerate}
\end{assumption}

Assumptions~\ref{assum_est}\ref{assum_est1}, \ref{assum_est}\ref{assum_est2}, and~\ref{assum_est}\ref{assum_est4} parallel the conditions imposed under the null hypothesis for the between-cointegration tests. Assumption~\ref{assum_est}\ref{assum_est3} explicitly imposes the cross-moment convergence required to derive the limiting distribution of the fully modified estimator; conditions of this type are discussed in \citet{seong2021functional} and \citet{Nam2025}.

As discussed in Section~\ref{sec_compest}, constructing the fully modified estimator requires the long-run and one-sided long-run covariance operators $\Omega_{\mathbf{x}\mathbf{x}}$, $\Omega_{\mathbf{x}\mathbf{x}}^+$, $\Omega_{U\mathbf{x}}$, and $\Omega_{U\mathbf{x}}^+$. The following proposition establishes operator-norm consistency of their kernel estimators defined in \eqref{eqsample1}--\eqref{eqsample2}. For the subsequent simulation-based inference, it also establishes consistency of $\widehat{\Omega}_{UU}$ defined in \eqref{eq_lrv_UU}, computed here using the zero-intercept residuals $\Upre_t=Y_t-\Bpre(\mathbf{x}_t)$.

\begin{proposition}\label{prop1}
Suppose that Assumption~\ref{assum_est} holds. Let the kernel estimators $\widehat{\Omega}_{\mathbf{x}\mathbf{x}}$, $\widehat{\Omega}_{\mathbf{x}\mathbf{x}}^+$, $\widehat{\Omega}_{U\mathbf{x}}$, $\widehat{\Omega}_{U\mathbf{x}}^+$, and $\widehat{\Omega}_{UU}$ be defined as in \eqref{eqsample1}--\eqref{eqsample2} and \eqref{eq_lrv_UU}, with all residual-based terms computed using $\Upre_t=Y_t-\Bpre(\mathbf{x}_t)$. Then these estimators are consistent in operator norm for their population counterparts. Moreover, $\widehat{\Omega}_{U|\mathbf{x}}\to_p\Omega_{U|\mathbf{x}}$ in operator norm, where
\begin{equation*}
\widehat{\Omega}_{U|\mathbf{x}}=\widehat{\Omega}_{UU}-\widehat{\Omega}_{U\mathbf{x}}\widehat{\Omega}_{\mathbf{x}\mathbf{x}}^{-1}\widehat{\Omega}_{\mathbf{x}U},\qquad \Omega_{U|\mathbf{x}}=\Omega_{UU}-\Omega_{U\mathbf{x}}\Omega_{\mathbf{x}\mathbf{x}}^{-1}\Omega_{\mathbf{x}U}.
\end{equation*}
\end{proposition}

\begin{proof}[Proof of Proposition~\ref{prop1}]
The consistency of $\widehat{\Omega}_{\mathbf{x}\mathbf{x}}$ follows from the $L^4$-$m$-approximability of $\Delta\mathbf{x}_t$, the kernel and bandwidth conditions in Assumption~\ref{assum_est}\ref{assum_est4}, and standard long-run covariance consistency results for functional time series, such as Theorem~2 of \citet{horvath2013estimation}. The consistency of $\widehat{\Omega}_{\mathbf{x}\mathbf{x}}^+$ follows from the analogous one-sided argument.

It remains to justify using the fitted residuals to estimate covariance operators involving $U_t$. Observe that
\begin{equation*}
\Bpre=\widehat C_{Y\mathbf{x}}\widehat C_{\mathbf{x}\mathbf{x}}^{-1}=B+\left(\frac{1}{T}\sum_{t=1}^T\mathbf{x}_t\otimes U_t\right)\widehat C_{\mathbf{x}\mathbf{x}}^{-1}.
\end{equation*}
Under Assumption~\ref{assum_est}, $T\widehat C_{\mathbf{x}\mathbf{x}}^{-1}=O_p(1)$. Standard moment bounds for the jointly $L^4$-$m$-approximable process, together with the finite dimensionality of $\HX$, also give $T^{-1}\sum_{t=1}^T\mathbf{x}_t\otimes U_t=O_p(1)$ in operator norm. Hence, $\Bpre-B=O_p(T^{-1})$. Since $\max_{1\leq t\leq T}\|\mathbf{x}_t\|=O_p(T^{1/2})$,
\begin{equation*}
\max_{1\leq t\leq T}\|\Upre_t-U_t\|\leq\|\Bpre-B\|_{\op}\max_{1\leq t\leq T}\|\mathbf{x}_t\|=O_p(T^{-1/2}).
\end{equation*}

Let $\widehat{\Omega}_{0,U\mathbf{x}}$ denote the infeasible estimator obtained by replacing $\Upre_t$ with $U_t$ in $\widehat{\Omega}_{U\mathbf{x}}$. The preceding bound, the stationarity of $\Delta\mathbf{x}_t$, and $|\mathrm{k}(x)|\leq1$ give
\begin{equation*}
\|\widehat{\Omega}_{U\mathbf{x}}-\widehat{\Omega}_{0,U\mathbf{x}}\|_{\op}=O_p(hT^{-1/2})=o_p(1).
\end{equation*}
The infeasible estimator $\widehat{\Omega}_{0,U\mathbf{x}}$ is consistent for $\Omega_{U\mathbf{x}}$ by applying the same long-run covariance argument to $Z_t=(\Delta\mathbf{x}_t,U_t)'$. The analogous one-sided argument yields $\widehat{\Omega}_{U\mathbf{x}}^+\to_p\Omega_{U\mathbf{x}}^+$.

Let $\widehat{\Omega}_{0,UU}$ denote the infeasible estimator obtained by replacing $\Upre_t$ with $U_t$ in $\widehat{\Omega}_{UU}$. Writing $D_t=\Upre_t-U_t$ and expanding $\Upre_t\otimes\Upre_{t+j}-U_t\otimes U_{t+j}$ into its two cross terms and the term $D_t\otimes D_{t+j}$ gives
\begin{equation*}
\|\widehat{\Omega}_{UU}-\widehat{\Omega}_{0,UU}\|_{\op}=O_p(hT^{-1/2})+O_p(hT^{-1})=o_p(1).
\end{equation*}
Standard long-run covariance consistency gives $\widehat{\Omega}_{0,UU}\to_p\Omega_{UU}$, and therefore $\widehat{\Omega}_{UU}\to_p\Omega_{UU}$.

Finally, positive definiteness of $\Omega_{\mathbf{x}\mathbf{x}}$ on the finite-dimensional space $\HX$ gives $\widehat{\Omega}_{\mathbf{x}\mathbf{x}}^{-1}\to_p\Omega_{\mathbf{x}\mathbf{x}}^{-1}$. Consistency of the adjoint estimators and the continuous mapping theorem then yield $\widehat{\Omega}_{U|\mathbf{x}}\to_p\Omega_{U|\mathbf{x}}$.
\end{proof}

Using the consistent long-run covariance estimators from Proposition~\ref{prop1}, define
\begin{equation}\label{eqest01}
Z_{1,t}=Y_t-\widehat{\Omega}_{U\mathbf{x}}\widehat{\Omega}_{\mathbf{x}\mathbf{x}}^{-1}\Delta\mathbf{x}_t,\qquad
\widehat{\Upsilon}=\widehat{\Omega}_{U\mathbf{x}}^+-\widehat{\Omega}_{U\mathbf{x}}\widehat{\Omega}_{\mathbf{x}\mathbf{x}}^{-1}\widehat{\Omega}_{\mathbf{x}\mathbf{x}}^+.
\end{equation}
The resulting functional fully modified least-squares estimator, extending the approach of \citet{phillips1995fully}, is
\begin{equation*}
\widehat B=\bigl(\widehat C_{Z_1\mathbf{x}}-\widehat\Upsilon\bigr)\widehat C_{\mathbf{x}\mathbf{x}}^{-1},\qquad
\widehat C_{Z_1\mathbf{x}}=\frac{1}{T}\sum_{t=1}^T\mathbf{x}_t\otimes Z_{1,t}.
\end{equation*}

\begin{theorem}\label{thm2} 
Suppose that Assumption~\ref{assum_est} holds. Then $\widehat B$ is superconsistent, with $\|\widehat B-B\|_{\op}=O_p(T^{-1})$. Moreover, for every fixed $v\in\HX$ and $w\in\HY$, 
\begin{equation}\label{eqconv1} 
\langle T(\widehat B-B)v,w\rangle\Rightarrow\left\langle\left(\int_0^1W_{\mathbf{x}}(r)\otimes dW_{U|\mathbf{x}}(r)\right)\left(\int_0^1W_{\mathbf{x}}(r)\otimes W_{\mathbf{x}}(r)\,dr\right)^{-1}v,w\right\rangle. 
\end{equation} 
Here $W_{U|\mathbf{x}}=W_U-\Omega_{U\mathbf{x}}\Omega_{\mathbf{x}\mathbf{x}}^{-1}W_{\mathbf{x}}$ is an $\HY$-valued Brownian motion independent of $W_{\mathbf{x}}$, with covariance operator $\Omega_{U|\mathbf{x}}=\Omega_{UU}-\Omega_{U\mathbf{x}}\Omega_{\mathbf{x}\mathbf{x}}^{-1}\Omega_{\mathbf{x}U}$, where $\Omega_{\mathbf{x}U}=\Omega_{U\mathbf{x}}^*$. 
\end{theorem} 
 
\begin{proof}[Proof of Theorem~\ref{thm2}] 
Write $\widehat{\Omega}_L=\widehat{\Omega}_{U\mathbf{x}}\widehat{\Omega}_{\mathbf{x}\mathbf{x}}^{-1}$ and $\Omega_L=\Omega_{U\mathbf{x}}\Omega_{\mathbf{x}\mathbf{x}}^{-1}$. From the definition of $\widehat B$, 
\begin{equation}\label{eqproof01} 
T(\widehat B-B)=\left[\frac{1}{T}\sum_{t=1}^T\mathbf{x}_t\otimes U_t-\widehat\Upsilon-\frac{1}{T}\sum_{t=1}^T\mathbf{x}_t\otimes\widehat{\Omega}_L\Delta\mathbf{x}_t\right]T\widehat C_{\mathbf{x}\mathbf{x}}^{-1}. 
\end{equation} 
By Proposition~\ref{prop1}, $\widehat{\Omega}_L\to_p\Omega_L$ and $\widehat\Upsilon-\Upsilon=o_p(1)$, where $\Upsilon=\Omega_{U\mathbf{x}}^+-\Omega_L\Omega_{\mathbf{x}\mathbf{x}}^+$. Since $\HX$ is finite-dimensional, Assumption~\ref{assum_est}\ref{assum_est3} also gives $T^{-1}\sum_{t=1}^T\mathbf{x}_t\otimes\Delta\mathbf{x}_t=O_p(1)$ in operator norm. Therefore, replacing $\widehat{\Omega}_L$ with $\Omega_L$ in \eqref{eqproof01} contributes only $o_p(1)$, and 
\begin{equation}\label{eqestexpan} 
T(\widehat B-B)=\left[\left(\frac{1}{T}\sum_{t=1}^T\mathbf{x}_t\otimes U_t-\Omega_{U\mathbf{x}}^+\right)-\Omega_L\left(\frac{1}{T}\sum_{t=1}^T\mathbf{x}_t\otimes\Delta\mathbf{x}_t-\Omega_{\mathbf{x}\mathbf{x}}^+\right)+o_p(1)\right]T\widehat C_{\mathbf{x}\mathbf{x}}^{-1}. 
\end{equation} 
 
For fixed $v\in\HX$ and $w\in\HY$, take $z=(0,w)'\in\mathcal H_Z$ in Assumption~\ref{assum_est}\ref{assum_est3}. Since $\mathbf{x}_t=\mathbf{x}_{t-1}+\Delta\mathbf{x}_t$, 
\begin{align} 
\frac{1}{T}\sum_{t=1}^T\langle\mathbf{x}_t,v\rangle\langle U_t,w\rangle 
&=\frac{1}{T}\sum_{t=1}^T\langle\mathbf{x}_{t-1},v\rangle\langle U_t,w\rangle+\frac{1}{T}\sum_{t=1}^T\langle\Delta\mathbf{x}_t,v\rangle\langle U_t,w\rangle\notag\\ 
&\Rightarrow\left\langle\left(\int_0^1W_{\mathbf{x}}(r)\otimes dW_U(r)\right)v,w\right\rangle+\langle\Omega_{U\mathbf{x}}^+v,w\rangle. 
\label{eqpf04} 
\end{align} 
Likewise, taking $z=(\Omega_L^*w,0)'\in\mathcal H_Z$ gives 
\begin{equation}\label{eqpf05} 
\frac{1}{T}\sum_{t=1}^T\langle\mathbf{x}_t,v\rangle\langle\Omega_L\Delta\mathbf{x}_t,w\rangle\Rightarrow\left\langle\left(\int_0^1W_{\mathbf{x}}(r)\otimes d(\Omega_LW_{\mathbf{x}}(r))\right)v,w\right\rangle+\langle\Omega_L\Omega_{\mathbf{x}\mathbf{x}}^+v,w\rangle. 
\end{equation} 
Jointly with these limits, the FCLT and continuous mapping theorem give 
\begin{equation*} 
T\widehat C_{\mathbf{x}\mathbf{x}}^{-1}v\Rightarrow\left(\int_0^1W_{\mathbf{x}}(r)\otimes W_{\mathbf{x}}(r)\,dr\right)^{-1}v. 
\end{equation*} 
Because $\HX$ is finite-dimensional, the preceding limits may be applied jointly to a basis of $\HX$ and combined with this random inverse. Combining \eqref{eqestexpan}--\eqref{eqpf05} and using $dW_{U|\mathbf{x}}=dW_U-d(\Omega_LW_{\mathbf{x}})$ yields \eqref{eqconv1}. 
 
Because the preceding convergence holds jointly over arbitrary finite collections of $v\in\HX$ and $w\in\HY$, Lemma~S5.1(i) of \citet{seo2020functional}, which follows from Theorem~3.1 in Chapter~1 of \citet{skorohod2001}, implies that $T(\widehat B-B)$ is stochastically bounded in operator norm. Hence, $\|\widehat B-B\|_{\op}=O_p(T^{-1})$. 
 
Finally, for all $r,s\in[0,1]$, the cross-covariance operator between $W_{U|\mathbf{x}}(r)$ and $W_{\mathbf{x}}(s)$ is 
\begin{equation*} 
(r\wedge s)(\Omega_{U\mathbf{x}}-\Omega_L\Omega_{\mathbf{x}\mathbf{x}})=0. 
\end{equation*} 
Since the processes are jointly Gaussian, $W_{U|\mathbf{x}}$ and $W_{\mathbf{x}}$ are independent. A direct covariance calculation gives $\Omega_{U|\mathbf{x}}=\Omega_{UU}-\Omega_{U\mathbf{x}}\Omega_{\mathbf{x}\mathbf{x}}^{-1}\Omega_{\mathbf{x}U}$. This completes the proof. 
\end{proof} 
 
\begin{remark}[Standard least-squares estimator]\label{rem_ls_limit} 
The same expansion used in the proof of Theorem~\ref{thm2} shows that, for fixed $v\in\HX$ and $w\in\HY$, the limiting distribution of the least-squares estimator $\Bpre$ contains nuisance-dependent terms. Let $\Omega_L=\Omega_{U\mathbf{x}}\Omega_{\mathbf{x}\mathbf{x}}^{-1}$, $\Upsilon=\Omega_{U\mathbf{x}}^+-\Omega_L\Omega_{\mathbf{x}\mathbf{x}}^+$, and $J=\int_0^1W_{\mathbf{x}}(r)\otimes W_{\mathbf{x}}(r)\,dr$. Then 
\begin{equation}\label{eqconv1add} 
\langle T(\Bpre-B)v,w\rangle\Rightarrow\left\langle\left(\int_0^1W_{\mathbf{x}}(r)\otimes dW_{U|\mathbf{x}}(r)\right)J^{-1}v,w\right\rangle+\mathcal G_1+\mathcal G_2, 
\end{equation} 
where 
\begin{equation*} 
\mathcal G_1=\langle\Upsilon J^{-1}v,w\rangle,\qquad \mathcal G_2=\left\langle\left(\int_0^1W_{\mathbf{x}}(r)\otimes d(\Omega_LW_{\mathbf{x}}(r))+\Omega_L\Omega_{\mathbf{x}\mathbf{x}}^+\right)J^{-1}v,w\right\rangle. 
\end{equation*} 
The subtraction of $\widehat\Upsilon$ removes the one-sided bias term $\mathcal G_1$, while the transformation from $Y_t$ to $Z_{1,t}$ removes the long-run endogeneity component $\mathcal G_2$. Thus, following the approach of \citet{phillips1995fully}, the fully modified estimator $\widehat B$ eliminates both nuisance-dependent terms and retains the centered limit in \eqref{eqconv1}. 
\end{remark} 
 
The limit in \eqref{eqconv1} is driven by the independent Brownian motions $W_{\mathbf{x}}$ and $W_{U|\mathbf{x}}$. Feasible inference therefore approximates these processes using consistent estimates of the spectral decompositions of their covariance operators $\Omega_{\mathbf{x}\mathbf{x}}$ and $\Omega_{U|\mathbf{x}}$.

\subsection{Asymptotic approximation for inference}\label{sec_app_est2}
To implement inference based on Theorem~\ref{thm2}, we approximate the distribution of the limit in \eqref{eqconv1}. As a covariance operator on $\HY$, $\Omega_{U|\mathbf{x}}$ is self-adjoint, nonnegative, and trace class, and hence compact. The finite-dimensional covariance operator $\Omega_{\mathbf{x}\mathbf{x}}$ has the same properties and is positive definite. Therefore, these operators admit the spectral decompositions \citep[p.~34]{Bosq2000} $\Omega_{U|\mathbf{x}}=\sum_{j=1}^{\infty}\lambda_jv_j\otimes v_j$ and $\Omega_{\mathbf{x}\mathbf{x}}=\sum_{j=1}^{\KX}\mu_jw_j\otimes w_j$, where $\lambda_1\geq\lambda_2\geq\cdots\geq0$, $\mu_1\geq\cdots\geq\mu_{\KX}>0$, and $\{v_j\}$ and $\{w_j\}$ are the corresponding orthonormal eigenvectors in $\HY$ and $\HX$, respectively.

Let $\{W_{1,j}\}_{j\geq1}$ and $\{W_{2,j}\}_{j=1}^{\KX}$ be two mutually independent families of independent standard scalar Brownian motions. Then
\begin{equation}\label{eqpopbm}
W_{U|\mathbf{x}}(r)=\sum_{j=1}^{\infty}\lambda_j^{1/2}v_jW_{1,j}(r),\qquad W_{\mathbf{x}}(r)=\sum_{j=1}^{\KX}\mu_j^{1/2}w_jW_{2,j}(r).
\end{equation}
For a positive integer $M$, define the truncated response Brownian motion by retaining the first $M$ covariance eigencomponents:
\begin{equation*}
W_{U|\mathbf{x}}^{(M)}(r)=\sum_{j=1}^{M}\lambda_j^{1/2}v_jW_{1,j}(r),\qquad r\in[0,1].
\end{equation*}
Thus, $M$ controls the number of response directions used in the approximation. Because $\Omega_{U|\mathbf{x}}$ is trace class, $\sum_{j\geq1}\lambda_j=\operatorname{tr}(\Omega_{U|\mathbf{x}})<\infty$. Doob's maximal inequality therefore implies that $W_{U|\mathbf{x}}^{(M)}$ converges to $W_{U|\mathbf{x}}$ in mean square under the uniform path norm as $M\to\infty$. No truncation is needed for $W_{\mathbf{x}}$, whose representation contains only $\KX$ components. If the population eigenelements were known, the distribution in \eqref{eqconv1} could therefore be approximated by replacing $W_{U|\mathbf{x}}$ with $W_{U|\mathbf{x}}^{(M)}$ and simulating its $M$ scalar Brownian components together with all $\KX$ components of $W_{\mathbf{x}}$.

For feasible implementation, we replace the population eigenelements with their sample counterparts. To control the resulting approximation error as $M$ increases with $T$, we impose the following convergence-rate conditions on the long-run covariance estimators.
\begin{assumption}\label{assumomega2}
The kernel-based long-run covariance estimators satisfy $\|\widehat{\Omega}_{\mathbf{x}\mathbf{x}}-\Omega_{\mathbf{x}\mathbf{x}}\|_{\op}=O_p(\sqrt{h/T})$, $\|\widehat{\Omega}_{U\mathbf{x}}-\Omega_{U\mathbf{x}}\|_{\op}=O_p(\sqrt{h/T})$, and $\|\widehat{\Omega}_{UU}-\Omega_{UU}\|_{\op}=O_p(\sqrt{h/T})$.
\end{assumption}
Under Assumption~\ref{assumomega2} and the positive definiteness of $\Omega_{\mathbf{x}\mathbf{x}}$, standard inverse and product perturbation arguments also give $\|\widehat{\Omega}_{U|\mathbf{x}}-\Omega_{U|\mathbf{x}}\|_{\op}=O_p(\sqrt{h/T})$. Assumption~\ref{assumomega2} is a high-level rate condition that incorporates both sampling variation and lag-window bias. Rates of this order are available under suitable weak-dependence, smoothness, and bandwidth conditions in finite-dimensional and functional time series settings \citep[see, e.g.,][]{BERKES2016150,NSS2}. 

To control estimation of the retained eigenelements and the truncation error, we impose the following spectral condition.
\begin{assumption}\label{assum4}
There exist constants $c,C>0$ and $\rho>1$ such that $\lambda_j\leq Cj^{-\rho}$ and $\lambda_j-\lambda_{j+1}\geq cj^{-\rho-1}$ for every $j\geq1$.
\end{assumption}
Let $\widehat{\Omega}_{U|\mathbf{x}}=\sum_{j=1}^{\infty}\widehat\lambda_j\widehat v_j\otimes\widehat v_j$, where $\widehat\lambda_1\geq\widehat\lambda_2\geq\cdots\geq0$. For the theoretical coupling, orient $\widehat v_j$ so that $\langle\widehat v_j,v_j\rangle\geq0$ and use the same auxiliary Brownian motions as in \eqref{eqpopbm}. With a truncation level $M=M_T$, define
\begin{equation}\label{eqwun}
\widehat W_{U|\mathbf{x}}(r)=\sum_{j=1}^{M}\widehat\lambda_j^{1/2}\widehat v_jW_{1,j}(r).
\end{equation}
The use of the same Brownian motions is only a coupling device for the consistency argument; independent copies are drawn across Monte Carlo replications in implementation. Similarly, using $\widehat{\Omega}_{\mathbf{x}\mathbf{x}}=\sum_{j=1}^{\KX}\widehat\mu_j\widehat w_j\otimes\widehat w_j$, define
\begin{equation}\label{eqwn}
\widehat W_{\mathbf{x}}(r)=\sum_{j=1}^{\KX}\widehat\mu_j^{1/2}\widehat w_jW_{2,j}(r).
\end{equation}
No truncation is required for $\widehat W_{\mathbf{x}}$ because $\HX$ is finite-dimensional. Moreover, consistency of this component follows directly from convergence of the finite-dimensional covariance operators and does not require the eigenvalues of $\Omega_{\mathbf{x}\mathbf{x}}$ to be distinct. In what follows, write $\|W\|_\infty=\sup_{0\leq r\leq1}\|W(r)\|$.

\begin{theorem}\label{thm3}
Suppose that Assumptions~\ref{assum_est}, \ref{assumomega2}, and~\ref{assum4} hold. Let $M=M_T\to\infty$. The population and plug-in Brownian paths can be constructed on a common probability space such that $\|\widehat W_{\mathbf{x}}-W_{\mathbf{x}}\|_\infty=o_p(1)$. If $M^{(\rho+3)/2}\sqrt{h/T}\to0$, then $\|\widehat W_{U|\mathbf{x}}-W_{U|\mathbf{x}}\|_\infty=o_p(1)$. Consequently, for fixed $v\in\HX$ and $w\in\HY$, define
\begin{equation}\label{eqinference}
\widehat{\mathcal Z}_{v,w}:=\left\langle\left(\int_0^1\widehat W_{\mathbf{x}}(r)\otimes d\widehat W_{U|\mathbf{x}}(r)\right)\left(\int_0^1\widehat W_{\mathbf{x}}(r)\otimes\widehat W_{\mathbf{x}}(r)\,dr\right)^{-1}v,w\right\rangle.
\end{equation}
Conditionally on the sample, the distribution of $\widehat{\mathcal Z}_{v,w}$ converges weakly in probability to the distribution of the right-hand side of \eqref{eqconv1}. Hence, it consistently approximates the limiting distribution of $\langle T(\widehat B-B)v,w\rangle$.
\end{theorem}

\begin{proof}[Proof of Theorem~\ref{thm3}]
Let $\delta_T=\sqrt{h/T}$. Assumption~\ref{assumomega2}, the positive definiteness of $\Omega_{\mathbf{x}\mathbf{x}}$, and standard inverse and product perturbation bounds give
\begin{equation*}
\|\widehat{\Omega}_{\mathbf{x}\mathbf{x}}^{-1}-\Omega_{\mathbf{x}\mathbf{x}}^{-1}\|_{\op}=O_p(\delta_T),\qquad \|\widehat{\Omega}_{U|\mathbf{x}}-\Omega_{U|\mathbf{x}}\|_{\op}=O_p(\delta_T).
\end{equation*}
We first consider the $\HX$-valued Brownian motion. Suppose initially that the eigenvalues of $\Omega_{\mathbf{x}\mathbf{x}}$ are distinct. After choosing the signs of $\widehat w_j$ appropriately, Lemmas~4.2 and~4.3 of \citet{Bosq2000} give $|\widehat\mu_j-\mu_j|=O_p(\delta_T)$ and $\|\widehat w_j-w_j\|=O_p(\delta_T)$ for each $j$. Coupling \eqref{eqpopbm} and \eqref{eqwn} using the same scalar Brownian motions gives
\begin{equation*}
W_{\mathbf{x}}(r)-\widehat W_{\mathbf{x}}(r)=\sum_{j=1}^{\KX}\mu_j^{1/2}(w_j-\widehat w_j)W_{2,j}(r)+\sum_{j=1}^{\KX}(\mu_j^{1/2}-\widehat\mu_j^{1/2})\widehat w_jW_{2,j}(r).
\end{equation*}
Since $\KX$ is fixed and $\mu_{\KX}>0$, both sums are $O_p(\delta_T)$ under the uniform path norm. Hence, $\|\widehat W_{\mathbf{x}}-W_{\mathbf{x}}\|_\infty=o_p(1)$. If $\Omega_{\mathbf{x}\mathbf{x}}$ has repeated eigenvalues, its individual eigenvectors within the corresponding eigenspaces need not be consistently identified. The conditional distribution of the process in \eqref{eqwn}, however, depends only on $\widehat{\Omega}_{\mathbf{x}\mathbf{x}}$. We may therefore use the distributionally equivalent square-root coupling $W_{\mathbf{x}}=\Omega_{\mathbf{x}\mathbf{x}}^{1/2}B_{\mathbf{x}}$ and $\widehat W_{\mathbf{x}}=\widehat{\Omega}_{\mathbf{x}\mathbf{x}}^{1/2}B_{\mathbf{x}}$, 
where $B_{\mathbf{x}}$ is a common standard $\KX$-dimensional Brownian motion independent of the sample and of the auxiliary Brownian motions used to construct $W_{U|\mathbf{x}}$. Positive definiteness of $\Omega_{\mathbf{x}\mathbf{x}}$ and operator-norm consistency of $\widehat{\Omega}_{\mathbf{x}\mathbf{x}}$ imply
\begin{equation*}
\|\widehat{\Omega}_{\mathbf{x}\mathbf{x}}^{1/2}-\Omega_{\mathbf{x}\mathbf{x}}^{1/2}\|_{\op}=O_p(\delta_T).
\end{equation*}
It follows that
\begin{equation*}
\|\widehat W_{\mathbf{x}}-W_{\mathbf{x}}\|_\infty\leq\|\widehat{\Omega}_{\mathbf{x}\mathbf{x}}^{1/2}-\Omega_{\mathbf{x}\mathbf{x}}^{1/2}\|_{\op}\|B_{\mathbf{x}}\|_\infty=o_p(1).
\end{equation*}
Because this coupling preserves the conditional distribution of \eqref{eqwn}, the first conclusion holds without requiring distinct eigenvalues.

We next consider $W_{U|\mathbf{x}}$. Let $\mathcal F_T$ be the $\sigma$-field generated by the observed sample and write $\mathbb E^*(\cdot)=\mathbb E(\cdot\mid\mathcal F_T)$ for expectation over the auxiliary Brownian motions. For $j\leq M$, define
\begin{equation*}
c_{j,T}=\widehat\lambda_j^{1/2}\widehat v_j-\lambda_j^{1/2}v_j,
\end{equation*}
and set $c_{j,T}=-\lambda_j^{1/2}v_j$ for $j>M$. Then
\begin{equation*}
\widehat W_{U|\mathbf{x}}-W_{U|\mathbf{x}}=\sum_{j=1}^{\infty}c_{j,T}W_{1,j}.
\end{equation*}
Choose the signs of $\widehat v_j$ so that $\langle\widehat v_j,v_j\rangle\geq0$. Assumption~\ref{assum4} and Lemma~4.3 of \citet{Bosq2000} give, uniformly over $1\leq j\leq M$,
\begin{equation*}
\|\widehat v_j-v_j\|\leq Cj^{\rho+1}\|\widehat{\Omega}_{U|\mathbf{x}}-\Omega_{U|\mathbf{x}}\|_{\op}=O_p(\delta_Tj^{\rho+1}).
\end{equation*}
It follows from $\lambda_j\leq Cj^{-\rho}$ that
\begin{equation*}
\sum_{j=1}^{M}\lambda_j\|\widehat v_j-v_j\|^2=O_p(\delta_T^2M^{\rho+3}).
\end{equation*}
Moreover, since $\lambda_j\to0$ and $\lambda_j-\lambda_{j+1}\geq cj^{-\rho-1}$,
\begin{equation*}
\lambda_j=\sum_{k=j}^{\infty}(\lambda_k-\lambda_{k+1})\geq c_0j^{-\rho}
\end{equation*}
for some $c_0>0$. Lemma~4.2 of \citet{Bosq2000} gives $|\widehat\lambda_j-\lambda_j|\leq\|\widehat{\Omega}_{U|\mathbf{x}}-\Omega_{U|\mathbf{x}}\|_{\op}=O_p(\delta_T)$ uniformly in $j$. Therefore,
\begin{equation*}
|\widehat\lambda_j^{1/2}-\lambda_j^{1/2}|^2\leq\frac{|\widehat\lambda_j-\lambda_j|^2}{\lambda_j}=O_p(\delta_T^2j^\rho),
\end{equation*}
and hence
\begin{equation*}
\sum_{j=1}^{M}|\widehat\lambda_j^{1/2}-\lambda_j^{1/2}|^2=O_p(\delta_T^2M^{\rho+1}).
\end{equation*}
Finally, the eigenvalue decay condition gives $
\sum_{j>M}\lambda_j=O(M^{1-\rho}).$ 
Combining these bounds yields
\begin{equation*}
\sum_{j=1}^{\infty}\|c_{j,T}\|^2=O_p\!\left(\delta_T^2M^{\rho+3}+\delta_T^2M^{\rho+1}+M^{1-\rho}\right)=o_p(1),
\end{equation*}
where the final equality follows from $M\to\infty$, $\rho>1$, and $M^{(\rho+3)/2}\delta_T\to0$. The conditional maximal inequality now gives
\begin{equation*}
\mathbb E^*\|\widehat W_{U|\mathbf{x}}-W_{U|\mathbf{x}}\|_\infty^2\leq4\sum_{j=1}^{\infty}\|c_{j,T}\|^2=o_p(1).
\end{equation*}
The conditional Markov inequality therefore implies $\|\widehat W_{U|\mathbf{x}}-W_{U|\mathbf{x}}\|_\infty=o_p(1)$.

It remains to establish convergence of the stochastic-integral functional. Define
\begin{equation*}
\widehat S=\int_0^1\widehat W_{\mathbf{x}}(r)\otimes d\widehat W_{U|\mathbf{x}}(r),\qquad S=\int_0^1W_{\mathbf{x}}(r)\otimes dW_{U|\mathbf{x}}(r),
\end{equation*}
and
\begin{equation*}
\widehat J=\int_0^1\widehat W_{\mathbf{x}}(r)\otimes\widehat W_{\mathbf{x}}(r)\,dr,\qquad J=\int_0^1W_{\mathbf{x}}(r)\otimes W_{\mathbf{x}}(r)\,dr.
\end{equation*}
Uniform path convergence gives $\|\widehat J-J\|_{\op}=o_p(1)$. Since $J$ is almost surely invertible, $\|\widehat J^{-1}-J^{-1}\|_{\op}=o_p(1)$.
Set $a=J^{-1}v$. We decompose
\begin{align*}
\langle(\widehat S-S)a,w\rangle&=\int_0^1\langle\widehat W_{\mathbf{x}}(r)-W_{\mathbf{x}}(r),a\rangle\,d\langle W_{U|\mathbf{x}}(r),w\rangle+\int_0^1\langle\widehat W_{\mathbf{x}}(r),a\rangle\,d\langle\widehat W_{U|\mathbf{x}}(r)-W_{U|\mathbf{x}}(r),w\rangle \\&=:I_{1,T}+I_{2,T}.
\end{align*}
Conditionally on the sample and the auxiliary $\HX$-valued Brownian paths, $a$ and the scalar integrands are fixed, while $\{W_{1,j}\}_{j\geq1}$ remain independent standard Brownian motions. The \Ito{} isometry gives $\mathbb E^*(I_{1,T}^2)\leq\langle\Omega_{U|\mathbf{x}}w,w\rangle\|a\|^2\|\widehat W_{\mathbf{x}}-W_{\mathbf{x}}\|_\infty^2=o_p(1)$ 
and
 $\mathbb E^*(I_{2,T}^2)\leq\left(\int_0^1\langle\widehat W_{\mathbf{x}}(r),a\rangle^2\,dr\right)\sum_{j=1}^{\infty}\langle c_{j,T},w\rangle^2\leq\|a\|^2\|\widehat W_{\mathbf{x}}\|_\infty^2\|w\|^2\sum_{j=1}^{\infty}\|c_{j,T}\|^2=o_p(1).$ 
Since $\|a\|$ and $\|\widehat W_{\mathbf{x}}\|_\infty$ are stochastically bounded, the conditional Markov inequality yields $I_{1,T}=o_p(1)$ and $I_{2,T}=o_p(1)$. Thus, $\langle(\widehat S-S)J^{-1}v,w\rangle=o_p(1).$ 
The stated rate condition also implies $M\delta_T=o(1)$. Consequently,
\begin{equation*}
\sum_{j=1}^{M}\widehat\lambda_j\leq\operatorname{tr}(\Omega_{U|\mathbf{x}})+M\|\widehat{\Omega}_{U|\mathbf{x}}-\Omega_{U|\mathbf{x}}\|_{\op}=O_p(1).
\end{equation*}
The \Ito{} isometry and the finite dimensionality of $\HX$ therefore give $\|\widehat S\|_{\op}=O_p(1)$. It follows that
\begin{equation*}
\langle\widehat S\widehat J^{-1}v-SJ^{-1}v,w\rangle=\langle(\widehat S-S)J^{-1}v,w\rangle+\langle\widehat S(\widehat J^{-1}-J^{-1})v,w\rangle=o_p(1).
\end{equation*}
The couplings used above preserve the conditional joint distribution of the plug-in Brownian paths. Hence, conditionally on the sample, the distribution of $\widehat{\mathcal Z}_{v,w}$ converges weakly in probability to that of $\langle SJ^{-1}v,w\rangle$. By Theorem~\ref{thm2}, the latter is the limiting distribution of $\langle T(\widehat B-B)v,w\rangle$. This completes the proof.
\end{proof}
Theorem~\ref{thm3} establishes that, for fixed $v\in\HX$ and $w\in\HY$, the conditional distribution of the feasible quantity in \eqref{eqinference} consistently approximates the limiting distribution in \eqref{eqconv1}. The approximation replaces the population eigenelements of $\Omega_{U|\mathbf{x}}$ and $\Omega_{\mathbf{x}\mathbf{x}}$ with their sample counterparts $\{(\widehat\lambda_j,\widehat v_j)\}_{j=1}^M$ and $\{(\widehat\mu_j,\widehat w_j)\}_{j=1}^{\KX}$. For each Monte Carlo replication $\ei$, we simulate independent standard scalar Brownian motions $\{W_{1,j,(\ei)}\}_{j=1}^M$ and $\{W_{2,j,(\ei)}\}_{j=1}^{\KX}$ and evaluate \eqref{eqinference}. Let $\widehat q_\tau(v,w)$ denote the empirical $\tau$ quantile of the resulting Monte Carlo draws. The simulation error can be made small by using sufficiently many replications. This procedure provides inference for fixed projections of the functional coefficients. In particular, taking $v=e_j$ and $w=w_k$, the quantiles $\widehat q_{\alpha/2}(e_j,w_k)$ and $\widehat q_{1-\alpha/2}(e_j,w_k)$ yield the $(1-\alpha)$ local-average confidence interval in \eqref{eqlocalci}.

\section{Inference for the model with an intercept}\label{sec_app_det}
For empirical applications, we allow the model to include an intercept. We briefly describe how the preceding estimation and inference results extend to this specification. Let $\widetilde Y_t=Y_t-\bar Y_T$ and $\widetilde{\mathbf{x}}_t=\mathbf{x}_t-\bar{\mathbf{x}}_T$. Demeaning removes the intercept while leaving the slope operator $B$ unchanged. The procedures developed in Section~\ref{sec_just_1} can therefore be applied to $(\widetilde Y_t,\widetilde{\mathbf{x}}_t)$. The sample moment operators used to estimate $B$ are formed from these demeaned variables, while the corresponding asymptotic representations account for the centering of the integrated regressor process, as detailed below.

\subsection{Proposed estimator and asymptotic properties}\label{sec_app_det2}
The preliminary least-squares residuals are $\Upre_t=\widetilde Y_t-\Bpre(\widetilde{\mathbf{x}}_t)$, where $\Bpre=\widehat C_{Y\mathbf{x}}\widehat C_{\mathbf{x}\mathbf{x}}^{-1}$ is computed from the demeaned variables. The long-run covariance estimators $\widehat{\Omega}_{\mathbf{x}\mathbf{x}}$, $\widehat{\Omega}_{\mathbf{x}\mathbf{x}}^+$, $\widehat{\Omega}_{U\mathbf{x}}$, $\widehat{\Omega}_{U\mathbf{x}}^+$, and $\widehat{\Omega}_{UU}$ are computed as in \eqref{eqsample1}--\eqref{eqsample2} and \eqref{eq_lrv_UU}, using $\Upre_t$ in place of $U_t$ wherever $U_t$ enters the corresponding formula. Since $\Delta\widetilde{\mathbf{x}}_t=\Delta\mathbf{x}_t$, the differenced regressor is unaffected by demeaning.
Define $Z_{1,t}=\widetilde Y_t-\widehat{\Omega}_{U\mathbf{x}}\widehat{\Omega}_{\mathbf{x}\mathbf{x}}^{-1}\Delta\mathbf{x}_t$ and $\widehat\Upsilon=\widehat{\Omega}_{U\mathbf{x}}^+-\widehat{\Omega}_{U\mathbf{x}}\widehat{\Omega}_{\mathbf{x}\mathbf{x}}^{-1}\widehat{\Omega}_{\mathbf{x}\mathbf{x}}^+$. The fully modified estimator of $B$ under the intercept specification is
\begin{equation*}
\widehat B=\bigl(\widehat C_{Z_1\mathbf{x}}-\widehat\Upsilon\bigr)\widehat C_{\mathbf{x}\mathbf{x}}^{-1},\qquad \widehat C_{Z_1\mathbf{x}}=\frac{1}{T}\sum_{t=1}^T\widetilde{\mathbf{x}}_t\otimes Z_{1,t}.
\end{equation*}
We impose the following conditions for the model with an intercept. Let $\mathcal H_Z=\HX\times\HY$, equipped with the product inner product $\langle\cdot,\cdot\rangle_{\mathcal H_Z}$, and write $Z_t=(\Delta\mathbf{x}_t,U_t)'$ and $W_Z=(W_{\mathbf{x}},W_U)'$. Define $\bar W_{\mathbf{x}}=\int_0^1W_{\mathbf{x}}(r)\,dr$ and $W_{\mathbf{x}}^c(r)=W_{\mathbf{x}}(r)-\bar W_{\mathbf{x}}$. In part~\ref{assum_app_int3} below, the centered stochastic integral is understood as the corresponding left-endpoint-sum limit, equivalently,
\begin{equation*}
\int_0^1\langle W_{\mathbf{x}}^c(r),v\rangle\,d\langle W_Z(r),z\rangle_{\mathcal H_Z}:=\int_0^1\langle W_{\mathbf{x}}(r),v\rangle\,d\langle W_Z(r),z\rangle_{\mathcal H_Z}-\langle\bar W_{\mathbf{x}},v\rangle\langle W_Z(1),z\rangle_{\mathcal H_Z}.
\end{equation*}
\begin{assumption}[Conditions for estimation with an intercept]\label{assum_app_int}
The following conditions hold:
\begin{enumerate}[label=(\roman*)]
\item\label{assum_app_int1} The model $Y_t=\beta_0+B(\mathbf{x}_t)+U_t$ holds for some $\beta_0\in\HY$, where $\mathbf{x}_t=\mathbf{x}_0+\sum_{s=1}^t\Delta\mathbf{x}_s$, $T^{-1/2}\mathbf{x}_0=o_p(1)$, and $U_t$ is stationary. The long-run covariance matrix $\Omega_{\mathbf{x}\mathbf{x}}$ is positive definite, so $\mathbf{x}_t$ has full nonstationary rank and is not within-cointegrated. Thus, $Y_t$ and $\mathbf{x}_t$ are between-cointegrated.
\item\label{assum_app_int2} The joint process $Z_t=(\Delta\mathbf{x}_t,U_t)'$ satisfies the conditions in Assumption~\ref{assum_est}\ref{assum_est2}, with Brownian limit $W_Z=(W_{\mathbf{x}},W_U)'$.
\item\label{assum_app_int3} For any finite collection of pairs $(v,z)\in\HX\times\mathcal H_Z$, the following convergence holds jointly with the FCLT in part~\ref{assum_app_int2}:
\begin{equation*}
\frac{1}{T}\sum_{t=1}^T\langle\mathbf{x}_{t-1}-\bar{\mathbf{x}}_T,v\rangle\langle Z_t,z\rangle_{\mathcal H_Z}\Rightarrow\int_0^1\langle W_{\mathbf{x}}^c(r),v\rangle\,d\langle W_Z(r),z\rangle_{\mathcal H_Z}+\sum_{j=1}^{\infty}\mathbb E\!\left[\langle\Delta\mathbf{x}_t,v\rangle\langle Z_{t+j},z\rangle_{\mathcal H_Z}\right].
\end{equation*}
\item\label{assum_app_int4} The bandwidth $h$ and kernel $\mathrm{k}$ satisfy Assumption~\ref{assum_test_kernel}.
\end{enumerate}
\end{assumption}
Assumption~\ref{assum_app_int} is the intercept analogue of Assumption~\ref{assum_est}. The principal change is that the integrated regressor path is replaced by its centered limit $W_{\mathbf{x}}^c$. Under these conditions, the argument used to prove Proposition~\ref{prop1} applies to residuals computed from the demeaned variables and gives the same operator-norm consistency results for the long-run covariance estimators.
Although demeaning replaces $U_t$ with $\widetilde U_t=U_t-\bar U_T$, it does not replace the error Brownian motion in the slope limit with its centered counterpart. Indeed, $\sum_{t=1}^T\widetilde{\mathbf{x}}_t=0$ implies
\begin{equation*}
\sum_{t=1}^T\widetilde{\mathbf{x}}_t\otimes(U_t-\bar U_T)=\sum_{t=1}^T\widetilde{\mathbf{x}}_t\otimes U_t,
\end{equation*}
so demeaning $U_t$ has no effect on the relevant cross-product. Consequently, the fully modified limit is driven by $W_{U|\mathbf{x}}$ rather than a centered version of it.
\begin{theorem}\label{thmapp2}
Suppose that Assumption~\ref{assum_app_int} holds. Then $\widehat B$ is superconsistent, with $\|\widehat B-B\|_{\op}=O_p(T^{-1})$. Moreover, for every fixed $v\in\HX$ and $w\in\HY$,
\begin{equation}\label{eqconv1app}
\langle T(\widehat B-B)v,w\rangle\Rightarrow\left\langle\left(\int_0^1W_{\mathbf{x}}^c(r)\otimes dW_{U|\mathbf{x}}(r)\right)\left(\int_0^1W_{\mathbf{x}}^c(r)\otimes W_{\mathbf{x}}^c(r)\,dr\right)^{-1}v,w\right\rangle.
\end{equation}
Here $W_{U|\mathbf{x}}=W_U-\Omega_{U\mathbf{x}}\Omega_{\mathbf{x}\mathbf{x}}^{-1}W_{\mathbf{x}}$ is independent of $W_{\mathbf{x}}$ and hence also of $W_{\mathbf{x}}^c$.
\end{theorem}
\begin{proof}[Proof of Theorem~\ref{thmapp2}]
The argument follows the proof of Theorem~\ref{thm2}, applied to the demeaned variables $\widetilde Y_t=Y_t-\bar Y_T$ and $\widetilde{\mathbf{x}}_t=\mathbf{x}_t-\bar{\mathbf{x}}_T$. Demeaning removes the intercept without changing the slope operator $B$, and
\begin{equation*}
T^{-1/2}\widetilde{\mathbf{x}}_{\lfloor T\cdot\rfloor}\Rightarrow W_{\mathbf{x}}^c(\cdot),\qquad T^{-1}\widehat C_{\mathbf{x}\mathbf{x}}\Rightarrow\int_0^1W_{\mathbf{x}}^c(r)\otimes W_{\mathbf{x}}^c(r)\,dr.
\end{equation*}
The limiting operator on the right is almost surely invertible by the positive definiteness of $\Omega_{\mathbf{x}\mathbf{x}}$. Moreover, $\sum_{t=1}^T\widetilde{\mathbf{x}}_t=0$, so replacing $U_t$ with $U_t-\bar U_T$ has no effect on the relevant cross-products. Assumption~\ref{assum_app_int}\ref{assum_app_int3} therefore gives the centered analogues of the cross-moment limits in \eqref{eqpf04}--\eqref{eqpf05}. Substituting these limits into the fully modified expansion \eqref{eqestexpan} yields \eqref{eqconv1app}, with $W_{\mathbf{x}}$ replaced by $W_{\mathbf{x}}^c$ while the stochastic integral remains driven by $dW_{U|\mathbf{x}}$. The same stochastic-boundedness argument used in the proof of Theorem~\ref{thm2} gives $\|\widehat B-B\|_{\op}=O_p(T^{-1})$.
\end{proof}

\subsection{Asymptotic approximation for inference}\label{sec_app_det3}
The feasible approximation in Theorem~\ref{thm3} extends directly to the intercept specification. Generate $\widehat W_{\mathbf{x}}$ and $\widehat W_{U|\mathbf{x}}$ using the estimated eigenelements of $\widehat{\Omega}_{\mathbf{x}\mathbf{x}}$ and $\widehat{\Omega}_{U|\mathbf{x}}$, as in \eqref{eqwun}--\eqref{eqwn}, and center the simulated regressor path according to
\begin{equation*}
\widehat W_{\mathbf{x}}^c(r)=\widehat W_{\mathbf{x}}(r)-\int_0^1\widehat W_{\mathbf{x}}(u)\,du.
\end{equation*}
The simulated error path $\widehat W_{U|\mathbf{x}}$ is left unchanged. The feasible counterpart of the limiting distribution in \eqref{eqconv1app} is obtained by replacing $W_{\mathbf{x}}^c$ and $W_{U|\mathbf{x}}$ with $\widehat W_{\mathbf{x}}^c$ and $\widehat W_{U|\mathbf{x}}$, respectively.
\begin{theorem}\label{thmapp3}
Suppose that Assumptions~\ref{assum_app_int} and~\ref{assum4} hold and that Assumption~\ref{assumomega2} holds for the covariance estimators computed from the demeaned variables. Let $M=M_T\to\infty$. Then $\|\widehat W_{\mathbf{x}}^c-W_{\mathbf{x}}^c\|_\infty=o_p(1)$. If $M^{(\rho+3)/2}\sqrt{h/T}\to0$, then $\|\widehat W_{U|\mathbf{x}}-W_{U|\mathbf{x}}\|_\infty=o_p(1)$. Consequently, for fixed $v\in\HX$ and $w\in\HY$, define
\begin{equation}\label{eqinferenceadd}
\widehat{\mathcal Z}_{v,w}^c:=\left\langle\left(\int_0^1\widehat W_{\mathbf{x}}^c(r)\otimes d\widehat W_{U|\mathbf{x}}(r)\right)\left(\int_0^1\widehat W_{\mathbf{x}}^c(r)\otimes\widehat W_{\mathbf{x}}^c(r)\,dr\right)^{-1}v,w\right\rangle.
\end{equation}
Conditionally on the sample, the distribution of $\widehat{\mathcal Z}_{v,w}^c$ converges weakly in probability to the distribution on the right-hand side of \eqref{eqconv1app}. Hence, it consistently approximates the limiting distribution of $\langle T(\widehat B-B)v,w\rangle$.
\end{theorem}
\begin{proof}[Proof of Theorem~\ref{thmapp3}]
Let $\mathcal C$ denote the centering map $\mathcal C f(r)=f(r)-\int_0^1f(u)\,du$. This map is continuous under the uniform norm because $\|\mathcal C f-\mathcal C g\|_\infty\leq2\|f-g\|_\infty$. The coupling argument in the proof of Theorem~\ref{thm3} gives $\|\widehat W_{\mathbf{x}}-W_{\mathbf{x}}\|_\infty=o_p(1)$ and therefore $\|\widehat W_{\mathbf{x}}^c-W_{\mathbf{x}}^c\|_\infty=o_p(1)$. Under the stated rate condition, the same argument gives $\|\widehat W_{U|\mathbf{x}}-W_{U|\mathbf{x}}\|_\infty=o_p(1)$. Moreover, positive definiteness of $\Omega_{\mathbf{x}\mathbf{x}}$ ensures that $\int_0^1W_{\mathbf{x}}^c(r)\otimes W_{\mathbf{x}}^c(r)\,dr$ is almost surely invertible. Centering the regressor path preserves its conditional independence from the simulated error path. Applying the stochastic-integral argument from the proof of Theorem~\ref{thm3} with $W_{\mathbf{x}}$ and $\widehat W_{\mathbf{x}}$ replaced by $W_{\mathbf{x}}^c$ and $\widehat W_{\mathbf{x}}^c$ establishes the stated conditional convergence.
\end{proof}
Thus, relative to the zero-intercept procedure, the only additional Monte Carlo step is to center each simulated regressor path; $\widehat W_{U|\mathbf{x}}$ remains uncentered. Local-average confidence intervals are obtained from \eqref{eqlocalci} by taking $v=e_j$ and $w=w_k$.
\begin{remark}[Joint inference for a fixed collection of projections]\label{rem_joint_projection}
Let $\mathcal Z_{v,w}^c$ denote the limiting random variable on the right-hand side of \eqref{eqconv1app}. Fix $v\in\HX$ and a finite collection $w_1,\ldots,w_J\in\HY$, where $J$ does not increase with $T$. The scalar convergence in Theorem~\ref{thmapp2} extends jointly to the vector of projected estimation errors. Indeed, for every fixed $c=(c_1,\ldots,c_J)'\in\mathbb R^J$,
\begin{equation*}
\sum_{m=1}^Jc_m\langle T(\widehat B-B)v,w_m\rangle=\left\langle T(\widehat B-B)v,\sum_{m=1}^Jc_mw_m\right\rangle\Rightarrow\mathcal Z_{v,\sum_{m=1}^Jc_mw_m}^c=\sum_{m=1}^Jc_m\mathcal Z_{v,w_m}^c.
\end{equation*}
Because $\sum_{m=1}^Jc_mw_m$ is a fixed element of $\HY$, the Cramér--Wold device gives joint convergence of $(\langle T(\widehat B-B)v,w_m\rangle)_{m=1}^J$. If the same auxiliary Brownian paths are used for all $J$ coordinates within each Monte Carlo replication, applying Theorem~\ref{thmapp3} to every fixed linear combination and using the conditional Cramér--Wold device likewise give weak convergence in probability of the joint conditional distribution of $(\widehat{\mathcal Z}_{v,w_m}^c)_{m=1}^J$. Consequently, any continuous function of this finite-dimensional vector, including its Euclidean norm, can be calibrated using the joint Monte Carlo draws, provided that its limiting distribution is continuous at the relevant quantile. This result applies to fixed $J$ and does not assert weak convergence of $T(\widehat B-B)v$ in the full space $\HY$ or justify an increasing-dimensional norm statistic.
\end{remark}

\section{Additional empirical implementation and diagnostics}
\label{sec_app_empirical}
\subsection{Joint calibration of the common-response comparison}
\label{sec_app_common_response}
\noindent We provide implementation details for the fixed-basis common-response comparison reported in Section~\ref{subsec:responses_and_margins}. Let $a=z_1<\cdots<z_G=b$ denote the common anomaly grid. For grid functions $u$ and $v$, each inner product is evaluated by composite trapezoidal quadrature,
\[
\langle u,v\rangle\approx\sum_{g=1}^{G-1}\frac{z_{g+1}-z_g}{2}\{u(z_g)v(z_g)+u(z_{g+1})v(z_{g+1})\}.
\]
This rule is used to evaluate the projected coordinates $\langle\widehat B\mathbf c,\phi_m\rangle$ entering the projected norm.
Let $\mathbf c=e_1-e_2$ denote the common-response contrast. Under $H_{0,J}:\langle B\mathbf c,\phi_m\rangle=0$ for $m=1,\ldots,J$, define the coefficient-error coordinate vector
\[\mathbf Z_{T,J}:=\bigl(T\langle(\widehat B-B)\mathbf c,\phi_1\rangle,\ldots,T\langle(\widehat B-B)\mathbf c,\phi_J\rangle\bigr)'.\]
Under the null, the empirical coordinate vector $\bigl(T\langle\widehat B\mathbf c,\phi_1\rangle,\ldots,T\langle\widehat B\mathbf c,\phi_J\rangle\bigr)'$ equals $\mathbf Z_{T,J}$. For every fixed $\widetilde c=(c_1,\ldots,c_J)'\in\mathbb R^J$,
\[
\widetilde c'\mathbf Z_{T,J}=\sum_{m=1}^{J}c_m\langle T(\widehat B-B)\mathbf c,\phi_m\rangle=\left\langle T(\widehat B-B)\mathbf c,\sum_{m=1}^{J}c_m\phi_m\right\rangle.
\]
Theorems~\ref{thmapp2} and~\ref{thmapp3}, Slutsky's theorem, and the Cramér--Wold device give the joint limit of $\mathbf Z_{T,J}$; see Remark~\ref{rem_joint_projection}. In each Monte Carlo replication, the same simulated Brownian paths are used across all $J$ directions, and the Euclidean norm of the resulting joint coordinate vector provides the simulated counterpart of the projected norm.
The baseline response space uses $J=20$. Its retained share of the temporally centered sample CLR $L^2$ variation, ${\sum_{t=1}^T\|P_J(Y_t-\bar Y_T)\|^2}/{\sum_{t=1}^T\|Y_t-\bar Y_T\|^2}$, is 90.7\%. For the finite-sample calibration, we use $M=7$, the nearest integer to $\sqrt{T/h}$, where $h$ is chosen as the nearest integer to $T^{1/4}$. This choice of $M$ is heuristic and is not intended as an asymptotic selection rule. The first seven eigencomponents account for ${\sum_{j=1}^{7}\widehat\lambda_j}/{\operatorname{tr}(\widehat\Omega_{U|\mathbf{x}})}=94.4\%$ of the estimated conditional response long-run covariance trace.
All sensitivity calculations use 50{,}000 joint draws. Re-estimating the fully modified model over $J\in\{8,12,16,20,24,28,32,40\}$ with $M=7$ gives add-one Monte Carlo $p$-values ranging from 0.0174 to 0.0218. Holding $J=20$ and recalibrating over $M\in\{3,5,7,10,15,20\}$ gives $p$-values ranging from 0.0172 to 0.0221. Every configuration rejects the fixed-basis equality restriction at the 5\% level.

\subsection{Sensitivity to changing spatial coverage}
\label{sec_app_coverage}

\noindent
The number of observed cells rises from 955 of 2{,}592 in 1850 to 2{,}093 in
2024, so the annual densities are constructed from a changing cross-section of
locations. We evaluate sensitivity to spatial coverage under two alternative
constructions, repeating kernel density estimation, the CLR transformation,
the basis representation, fully modified estimation, the margin summaries,
and the between-cointegration test. The forcing vector is unchanged.

The first construction follows \citet{chang2020evaluating}: Northern- and
Southern-Hemisphere densities are estimated separately and averaged with equal
weight. The second restricts the cross-section to the 190 cells that report at
least one monthly observation in every year. These cells account for 7.3\% of
the grid and are concentrated in the Northern Hemisphere; the median annual
number of grid-cell-month observations falls from 15{,}600 to 2{,}203.

\begin{table}[h]
\centering
\caption{Coverage diagnostics. Temperature margins are in $^\circ$C; WarmMass
and Extreme are warm- and combined-tail probability changes.}
\label{Tab:Balanced}
\renewcommand{\arraystretch}{1.15}
\setlength{\tabcolsep}{6pt}
\begin{tabular}{llrrrrrr}
\toprule
Cross-section & Forcing & $\Delta\mu$ & $\delta^{\mathrm{loc}}$
& $\Delta\mathrm{IQR}$ & $\Delta\mathrm{WarmMass}$
& $\Delta\mathrm{Extreme}$ & $p(\widehat{\mathcal T}_V)$ \\
\midrule
Baseline & F1 & 0.3125 & 0.2646 & $-$0.1120 & 0.0223 & 0.0018 & \multirow{2}{*}{0.0748} \\
         & F2 & 0.0789 & 0.0883 & 0.1500 & 0.0080 & 0.0083 & \\
\midrule
Hemispheric & F1 & 0.2955 & 0.2546 & $-$0.1235 & 0.0186 & $-$0.0030 & \multirow{2}{*}{0.1170} \\
equal weight & F2 & 0.0816 & 0.0902 & 0.1601 & 0.0108 & 0.0128 & \\
\midrule
Fixed grid & F1 & 0.2956 & 0.2650 & $-$0.1642 & 0.0090 & $-$0.0230 & \multirow{2}{*}{0.7317} \\
(190 cells) & F2 & 0.1584 & 0.1635 & 0.2118 & 0.0297 & 0.0323 & \\
\bottomrule
\end{tabular}
\end{table}

\noindent
Hemispheric weighting produces mean, location, and IQR estimates close to the
baseline values. The fixed-grid construction concerns a smaller, time-invariant
set of locations, and its estimates differ more in magnitude, especially for
the F2 mean response. Both alternatives preserve the negative F1 and positive F2 IQR point
estimates, and the between-cointegration test does not reject at the 5\%
level under either construction.

\subsection{Pointwise uncertainty for implied densities and local mass responses}\label{sec_app_density_ci}
\noindent We describe uncertainty in the implied densities and full local probability-mass responses by propagating joint draws of the CLR-response estimation errors through the corresponding transformations. Both calculations use the feasible approximation developed in Section~\ref{sec_app_det3} and proceed in four steps.
\begin{description}[style=nextline, leftmargin=1em, font=\bfseries,topsep=-2pt, itemsep=0pt]
\item[Step 1: Fixing the response representation and reporting inputs] \hfill
Let $\{\phi_m\}_{m=1}^J$, with $J=20$, denote the fixed orthonormal Fourier directions used in Section~\ref{sec_app_common_response}, and let $P_J$ denote the projection onto their span. For $k=1,2$, define $\widehat\beta_{k,J}=P_J\widehat B(e_k)$. As in Section~\ref{subsec:responses_and_margins}, $f_0$ is the arithmetic average of the annual densities, and $\Delta x_k$ is the chosen reporting increment, set equal to the sample standard deviation of forcing portfolio $k$. These inputs are held fixed throughout the simulation.
\item[Step 2: Simulating joint response-error draws] \hfill
Using the feasible intercept construction in Section~\ref{sec_app_det3}, generate $R_{\mathrm{MC}}=50{,}000$ joint draws of $\{\widehat{\mathcal Z}_{e_k,\phi_m}^c:k=1,2,\ m=1,\ldots,J\}$ in \eqref{eqinferenceadd}. We use the Parzen kernel with $h=4$ and the rank-$M$ covariance approximation with $M=7$, as described and assessed in Section~\ref{sec_app_common_response}. Within each replication $\ei$, the same simulated Brownian paths are used across all $J$ response directions and both forcing coordinates, with only the forcing path centered. Reconstruct the projected response-error draw as
\begin{equation*} E_{k,J}^{(\ei)}(s)=\frac{1}{T}\sum_{m=1}^J\widehat{\mathcal Z}_{e_k,\phi_m}^{c,(\ei)}\phi_m(s),\qquad \ei=1,\ldots,R_{\mathrm{MC}}. \end{equation*}
The factor $T^{-1}$ converts the simulated limiting draw to the response-error scale (see Theorem~\ref{thmapp2}).
\item[Step 3: Mapping the draws to densities] \hfill
For each replication, apply the normalized inverse-CLR map to obtain
\begin{equation*} f_k^{(\ei)}(s)=\frac{f_0(s)\exp\!\left\{\Delta x_k[\widehat\beta_{k,J}(s)-E_{k,J}^{(\ei)}(s)]\right\}}{\displaystyle\int_a^bf_0(u)\exp\!\left\{\Delta x_k[\widehat\beta_{k,J}(u)-E_{k,J}^{(\ei)}(u)]\right\}\,du}. \end{equation*}
The fitted Fourier response is used without the local averaging applied in the left panels of Figure~\ref{Fig:RF_Response}. Each $f_k^{(\ei)}$ is positive and integrates to one, with normalization performed separately for every draw.
\item[Step 4: Computing the pointwise simulation intervals] \hfill
Let $\widehat q_{k,\tau}(s)$ denote the empirical $\tau$ quantile of $\{f_k^{(\ei)}(s):\ei=1,\ldots,R_{\mathrm{MC}}\}$. At each fixed anomaly value $s$, the central 95\% pointwise simulation interval for the implied density is
\begin{equation*} \bigl[\widehat q_{k,0.025}(s),\widehat q_{k,0.975}(s)\bigr]. \end{equation*}
Because $f_0$ is held fixed, the corresponding interval for the implied density change $f_k(s)-f_0(s)$ is obtained by subtracting $f_0(s)$ from both endpoints. Figure~\ref{Fig:RF_Response} displays these intervals on the density scale. Their pointwise boundaries need not integrate to one.
For the full local probability-mass responses in Figure~\ref{Fig:BinLossDecomp}, use the same joint density draws from Step~3 and define
\begin{equation*} D_{k,\ell}^{(\ei)}(r)=\int_{I_\ell(r)}\{f_k^{(\ei)}(s)-f_0(s)\}\,ds,\qquad \ei=1,\ldots,R_{\mathrm{MC}}. \end{equation*}
The integrals are evaluated numerically on a fine grid. At each fixed window center $r$, the empirical 2.5th and 97.5th percentiles of $\{D_{k,\ell}^{(\ei)}(r):\ei=1,\ldots,R_{\mathrm{MC}}\}$ form the pointwise simulation interval for $\Delta L_{k,\ell}^{\mathrm{full}}(r)$. These intervals are reported only for the full local responses, not for the location-only or residual reshaping components. Taking quantiles after transforming, normalizing, and integrating the joint draws preserves the dependence represented by the covariance approximation; transforming marginal CLR interval endpoints separately would not generally reproduce these quantiles.
\end{description}
For fixed $f_0$, $\Delta x_k$, response space, rank-$M$ covariance approximation, and window center, the normalized inverse-CLR map is smooth in the retained Fourier coordinates, while integration over a fixed window is a continuous linear functional of the resulting density. Remark~\ref{rem_joint_projection} and the finite-dimensional delta method therefore justify the first-order propagation of joint CLR-response estimation uncertainty to density values and local probability masses under this calibration. The resulting intervals are pointwise rather than simultaneous and condition on the reference density, reporting increments, response-space choice, and covariance truncation.
\bibliography{VtD_biblio}
\end{document}